\documentclass{tlp}

\usepackage{aopmath}
\usepackage{color}
\usepackage{bbm}
\usepackage{psfrag}
\usepackage{epsf}
\usepackage{multicol}
\usepackage{multirow}
\usepackage{graphicx}
\usepackage{subfig}
\usepackage{setspace}
\usepackage{mathpartir}
\usepackage{fancyvrb}
\usepackage{algorithmic}
\usepackage{algorithm}
\usepackage[breaklinks=true]{hyperref}

\newtheorem{definition}{Definition} 
\newtheorem{example}{Example} 

\newcommand{\bench}[1]{{#1}}

\def\TG{\mathit{TG}}
\def\varset{\mathit{vars}}
\newcommand{\vars}[1]{\ensuremath{\mathit{vars}(#1)}}
\newcommand{\ep}{\ensuremath{\langle {atom}_1, \ldots, {atom}_n \rangle}}
\newcommand{\args}[2][X]{\ensuremath{#1_1, \ldots, #1_#2}}
\newcommand{\inargs}[1]{\ensuremath{\mathit{in\_args}(#1)}}
\newcommand{\outargs}[1]{\ensuremath{\mathit{out\_args}(#1)}}
\newcommand{\pp}[1]{\ensuremath{\mathit{pp}(#1)}}

\def\LVbefore{\mathit{LV}_{\mathit{before}}}
\def\LVafter{\mathit{LV}_{\mathit{after}}}
\newcommand{\LVb}[1]{\ensuremath{\mathit{LV}_{\mathit{before}}(#1)}}
\newcommand{\LVa}[1]{\ensuremath{\mathit{LV}_{\mathit{after}}(#1)}}
\newcommand{\VV}[1]{\ensuremath{\mathit{VV}(#1)}}

\def\LRbefore{\mathit{LR}_{before}}
\def\LRafter{\mathit{LR}_{after}}
\newcommand{\LRb}[1]{\ensuremath{\mathit{LR}_{before}(#1)}}
\newcommand{\LRa}[1]{\ensuremath{\mathit{LR}_{after}(#1)}}
\newcommand{\VR}[1]{\ensuremath{\mathit{VR}(#1)}}

\newcommand{\bornR}[1]{\ensuremath{{bornR}(#1)}}
\newcommand{\deadR}[1]{\ensuremath{{deadR}(#1)}}
\newcommand{\localR}[1]{\ensuremath{{localR}(#1)}}
\newcommand{\inputR}[1]{\ensuremath{{inputR}(#1)}}
\newcommand{\outputR}[1]{\ensuremath{{outputR}(#1)}}
\newcommand{\allocR}[1]{\ensuremath{{allocR}(#1)}}
\newcommand{\outlivedR}[1]{\ensuremath{{outlivedR}(#1)}}
\newcommand{\allocation}[1]{\ensuremath{{allocation}(#1)}}

\def\bornRegs{\mathit{bornR}}
\def\deadRegs{\mathit{deadR}}
\def\localRegs{\mathit{localR}}
\def\inputRegs{\mathit{inputR}}
\def\outputRegs{\mathit{outputR}}
\def\allocRegs{\mathit{allocR}}
\def\outlivedRegs{\mathit{outlivedR}}
\def\allocationRegs{\mathit{allocation}}

\newcommand{\atomat}[1]{\ensuremath{{atom\_at}(#1)}}
\newcommand{\Reach}[1]{\ensuremath{{Reach}(#1)}}

\newcommand{\create}[1]{\ensuremath{{create}(#1)}}
\newcommand{\remove}[1]{\ensuremath{{remove}(#1)}}

\newcommand{\code}[1]{{\tt#1}}

\title[Region-Based Memory Management for Mercury Programs]
    {Region-Based Memory Management\\for Mercury Programs}

\author[Q. Phan, G. Janssens and Z. Somogyi]
{QUAN PHAN, GERDA JANSSENS \\
Department of Computer Science \\
Katholieke Universiteit Leuven \\
Celestijnenlaan, 200A \\
B-3001 Leuven, Belgium \\
\email{quan.leuven@gmail.com,gerda.janssens@cs.kuleuven.be}
\and
ZOLTAN SOMOGYI \\
National ICT Australia and\\
Department of Computer Science and Software Engineering\\
The University of Melbourne\\
Victoria, 3010 Australia\\
\email{zs@unimelb.edu.au}
}

\pagerange{\pageref{firstpage}--\pageref{lastpage}}
\volume{\textbf{10} (3):}
\jdate{March 2002}
\pubyear{2012}

\submitted{4 December 2009}
\revised{16 November 2011}
\accepted{18 November 2011}

\begin{document}
\bibliographystyle{acmtrans}
\label{firstpage}
\maketitle

\begin{abstract}
Region-based memory management (RBMM)
is a form of compile time memory management,
well-known from the functional programming world.
In this paper we describe our work on implementing RBMM
for the logic programming language Mercury.
One interesting point about Mercury is that it is designed with
strong type, mode, and determinism systems.
These systems not only provide
Mercury programmers with several direct software engineering benefits,
such as self-documenting code and clear program logic,
but also give language implementors
a large amount of information that is useful for program analyses.
In this work, we make use of this information to develop program analyses
that determine the distribution of data into regions
and transform Mercury programs by inserting into them
the necessary region operations.
We prove the correctness of our program analyses and transformation.
To execute the annotated programs, we have implemented runtime support
that tackles the two main challenges posed by backtracking.
First, backtracking can require regions
removed during forward execution to be ``resurrected'';
and second,
any memory allocated during a computation that has been backtracked over
must be recovered promptly and without waiting
for the regions involved to come to the end of their life.
We describe in detail our solution of both these problems.
We study in detail how our RBMM system performs
on a selection of benchmark programs,
including some well-known difficult cases for RBMM.
Even with these difficult cases, our RBMM-enabled Mercury system
obtains clearly faster runtimes for 15 out of 18 benchmarks
compared to the base Mercury system with its Boehm runtime garbage collector,
with an average runtime speedup of 24\%,
and an average reduction in memory requirements of 95\%.
In fact, our system achieves optimal memory consumption in some programs.

A shorter version of this paper, without proofs,
is to appear in Theory and Practice of Logic Programming (TPLP).
\end{abstract}


\section{Introduction}

Memory management is an integral part
of all practical programming language systems.
Traditionally, memory has been left to the programmer to manage
using constructs such as C's \code{malloc} and \code{free},
but experience has shown that such manual systems
require a large amount of quite tedious of work from programmers,
and are very hard to use correctly.
More recent programming languages therefore automate memory management.
The standard way to implement automatic memory management
is runtime garbage collection.
This provides memory safety, good memory reuse, and reasonable performance,
but it does have a significant downside, which is that
decisions about which parts of memory can be reused
are made completely at runtime, which can incur significant overheads.

Region-based memory management or RBMM~\cite{TofteTalpin97}
is a recent technique for avoiding these overheads
by moving decisions from runtime to compile time,
thus shifting most of the responsibility to the compiler.
RBMM is based on the idea of
putting each group of heap objects that have the same lifetime
into their own regions,
the motive being that reclaiming entire regions at the end of their lifetime
makes collection very fast.
A typical scenario is a function storing its intermediate results in a region
that is freed once the final result of the function has been computed.
All the decisions about which objects are allocated into which regions
and when each region should be created and removed are made at compile time.

Since the fundamental work on RBMM
for functional programming~\cite{TofteTalpin97},
there have been several improvements and new developments
in that context~\cite{Aiken95better,Birkedal96from,Henglein01hmn}.
RBMM has also been adapted to other programming paradigms,
such as imperative programming~\cite{Gay98memory,Grossman02cyclone},
object-oriented programming~\cite{Cherem04,Chin04},
and logic programming~\cite{Makholm00,Makholm00master,MakholmSagonas02}.

The initial work on RBMM for logic programming languages
applied RBMM to Prolog.
However, the first attempt~\cite{Makholm00,Makholm00master}
was developed for a non-standard implementation of Prolog
which would require substantial changes
before it could be applied in any standard implementation.
The authors of \cite{MakholmSagonas02} fixed this problem by implementing RBMM
in the context of the standard technology for implementing Prolog,
the Warren Abstract Machine (WAM).
Nevertheless, this work mainly concentrated on
the runtime extensions needed to run Prolog programs with RBMM.
As its analysis algorithm, it used an adapted version
of a type-based region analysis originally developed
for the strongly typed functional language SML~\cite{Henglein01hmn}.
Since Prolog has no static type system
and more importantly no static mode system,
the region inference has to get the information it needs
from type and mode inferences, which often yield imprecise results.
Moreover, a Prolog implementation's lack of knowledge
about the determinism of a program's predicates
generally requires them to be treated as nondeterministic.
These limitations prevent the application of most of the optimizations
that would improve the performance of RBMM,
making it hard for it to become a practical alternative
to native runtime collectors in Prolog systems.
The logic programming language Mercury has none of these limitations;
the Mercury compiler knows the type of every variable
and the mode and determinism of every goal in the program.
This fact, the pure nature of Mercury (the absence of side-effects),
and the limited research on RBMM in logic programming
motivated us to investigate whether region-based memory management
could be developed and implemented efficiently for Mercury.

In this paper
we describe the first automated RBMM system for Mercury.
Given a Mercury program,
\begin{itemize}
\item
our system determines the set of regions the program should use;
\item
it decides, for each allocation site in the program,
which region the allocation should happen in;
\item
it inserts instructions into the program to create each region
just before it is first needed; and
\item
it inserts instructions into the program to remove each region
as soon as it is safe to do so.
\end{itemize}

The main contributions of our work are as follows.
\begin{enumerate}
\item
We develop the static program analyses needed
for generating region-annotated programs.
These include a region points-to analysis to divide Mercury terms into regions,
a liveness analysis that assigns lifetimes to the regions,
and a program transformation to annotate the
original programs with the derived region information.
\item
We prove several safety properties
for memory accesses and region operations in the resulting annotated programs.
\item
Our runtime support system
handles the interaction of RBMM with backtracking
correctly and without incurring excessive overheads.
\item
Our RBMM-enabled system achieves faster execution times
and much lower memory requirements
for most of our benchmark programs
than the standard Mercury system, which uses
the Boehm-Demers-Weiser garbage collector~\cite{Boehm88} for memory management.
The region system actually achieves
optimal memory consumption on some benchmarks.
\item
We make a detailed analysis of the RBMM behavior
of a selection of programs, including some well-known difficult cases.
This study reveals the impact of sharing on memory reuse in RBMM systems.
\end{enumerate}

A previous version of our region analysis and transformation
was published in \cite{Phan07iclp}.
In~\cite{Phan08ismm} we described the runtime support for RBMM.
They all have been reformulated, extended and/or refined in this paper.

The structure of the paper is as follows.
In Section~\ref{seCbg} we introduce Mercury
and the compiler's internal representation of Mercury programs.
Section~\ref{seCbgCmercuryRBMM} describes intuitively
how RBMM can be realized for Mercury,
and explains our decisions on how to support backtracking.
Section~\ref{seCregionCmodel} explains how we decide
which terms should be stored in which regions,
taking into account sharing among terms.
Based on this region model, we develop the static analyses of our system:
Sections~\ref{seCrpta},~\ref{seClra}, and~\ref{seCpt} contain respectively
our region points-to analysis, our region liveness analysis,
and our program transformation,
together with theorems about their correctness.
Section~\ref{seCsupportdet} shows
the basic extensions to the Mercury runtime system
needed to support RBMM in deterministic code,
while Section~\ref{seCsupportnondet}
describes the extensions needed to support backtracking (nondeterminism).
Section~\ref{seCexper} presents a detailed evaluation of our RBMM system,
as well as a discussion of the relation
between sharing and memory reuse in region-based systems.
We discuss related research in Section~\ref{seCrelatedwork},
present our ideas for future work in Section~\ref{seCfuture},
and conclude in Section~\ref{seCconclusion}.

\section{Background}
\label{seCbg}

\subsection{Mercury}
\label{seCbgCmercury}

Mercury is a pure logic programming language
intended for the creation of large, fast,
reliable programs~\cite{Somogyi96mercury}.
While the syntax of Mercury is based on the syntax of Prolog,
semantically the two languages are very different due to Mercury's purity,
its type, mode, determinism and module systems,
and its support for evaluable functions.
(Mercury treats functions
as predicates with the return value as an extra argument,
so in the rest of the paper we will talk only about predicates.)

Mercury has a strong Hindley-Milner type system very similar to Haskell's.
Some types are built into the language (e.g. \code{int}),
but users can also introduce new types using type definitions
such as the one in Example~\ref{example:typedeclaration}.
\begin{example} The declaration of the type \code{list\_int}.\\
\code{:- type list\_int ---> []; [int | list\_int]}.\\
This defines the type of lists of integers.
\hfill $\Box$
\label{example:typedeclaration}
\end{example}
Mercury programs are statically typed; the compiler knows the type of
every argument of every predicate (from declarations or inference)
and every local variable (from inference).

The mode system classifies each argument of each predicate
as either input or output;
there are exceptions, but they are not relevant to this paper.
If input, the argument passed by the caller must be a ground term.
If output, the argument passed by the caller must be a distinct free variable,
which the callee will instantiate to a ground term.
It is possible for a predicate to have more than one mode;
the usual example is \code{append}, which has two principal modes:
\code{append(in,in,out)} and \code{append(out,out,in)}.
We call each mode of a predicate a \emph{procedure}.
The Mercury compiler generates separate code for each procedure.

Each procedure has a determinism,
which puts limits on the number of its possible solutions.
Procedures with determinism \emph{det} succeed exactly once;
\emph{semidet} procedures succeed at most once;
\emph{multi} procedures succeed at least once;
while \emph{nondet} procedures may succeed any number of times.

\begin{example}
Figure~\ref{fig:running:normal} shows
the \emph{quicksort} program written in Mercury,
including declarations of the types, modes, and determinisms
for its two essential predicates, \code{qsort} and \code{split}.
We include the code of \code{main} for completeness,
but it is of no relevance to the topic of the paper.
The notation \code{!IO} represents two variables,
which in this case stand for the initial and final states of the world,
i.e.\ the state before the program writes out its result with \code{io.write},
and the state after.
(The \code{io.write} predicate is defined
in the \code{io} module of the Mercury standard library.)
\hfill $\Box$
\label{example:running:normal}
\end{example}
\begin{figure}
\scriptsize
\begin{Verbatim}[frame=single,framerule=0.2pt,framesep=3pt]
  main(!IO) :-                                    :- pred split(int, list_int, list_int,
      qsort([2, 3, 1], [], S),                        list_int).
      io.write(S, !IO).                           :- mode split(in, in, out, out) is det.
                                                  split(_, [], [], []).
  :- pred qsort(list_int, list_int, list_int).    split(X, [Le | Ls], L1, L2) :-
  :- mode qsort(in, in, out) is det.              ( if X >= Le then
  qsort([], A, A).                                    split(X, Ls, L11, L2),
  qsort([Le | Ls], A, S) :-                           L1 = [Le | L11]
      split(Le, Ls, L1, L2),                        else
      qsort(L2, A, S2),                               split(X, Ls, L1, L21),
      qsort(L1, [Le | S2], S).                        L2 = [Le | L21]
                                                  ).
\end{Verbatim}
\normalsize
\caption{The \emph{quicksort} program in Mercury.}
\label{fig:running:normal}
\end{figure}

We support a very large subset of Mercury:
unifications, first order calls,
conjunctions, disjunctions, switches,
if-then-elses, negations, and quantification.
The only parts we do not support are higher order calls
(including typeclass method calls),
calls to foreign language code,
and multi-module programs.
A complete description of Mercury can be found in~\cite{mercury_refman}.

\subsection{Mercury Code inside the Compiler}
\label{seCbgCinternal}

The compiler converts all predicate definitions into an internal form.
For our subset of Mercury,
this internal form is given by the following abstract syntax:
$$
\begin{array}{lcl}
\mbox{predicate}~P & : & p(x_1, \ldots, x_n)~\leftarrow~G\\
\mbox{goal}~G      & : & x = y ~|~ x = f(y_1, \ldots, y_n) ~|~
                            p(x_1, \ldots, x_n) ~|~ \\
                   &   & (G_1, \cdots, G_n) ~|~ (G_1 ; \ldots ; G_n) ~|~
                            not~G ~|~ \\
                   &   & (if~G_c~then~G_t~else~G_e) ~|~
                            some~[x_1,\ldots,x_n]~G \\
\end{array}
$$
We call the first three kinds of goals (unifications and calls)
atomic goals or just atoms.
The rest are called compound goals, in which
a sequence of goals separated by commas is a conjunction, while
a sequence of goals separated by semicolons is a disjunction.

As this implies, the Mercury compiler internally converts
any predicate definition with two or more clauses
into a single clause with an explicit disjunction.
The clauses themselves are transformed into \emph{superhomogeneous form},
in which each atom (including clause heads) must be of one of the forms
\code{p(X1,...,Xn)}, \code{Y = X}, or \code{Y = f(X1,...,Xn)},
where all of the \code{Xi} are distinct.

Inside the compiler, every goal (compound as well as atomic)
is annotated with mode and determinism information.
For unifications, we show the mode information by writing
\code{<=} for construction unifications,
\code{=>} for deconstruction unifications,
\code{==} for equality tests,
and \code{:=} for assignments.
The compiler reorders conjunctions as needed to ensure that
goals that consume the value of a variable
always come after the goal that produces its value.
We show the \emph{quicksort} program in this abstract syntax
in Figure~\ref{fig:running:superhomogeneous}.
\begin{figure}[htb]
\scriptsize
\begin{Verbatim}[frame=single,framerule=0.2pt,framesep=3pt]
         main(!IO) :-                          split(X, L, L1, L2) :-
         (1) L <= [2, 1, 3],                   (
         (2) A <= [],                          (1) L => [],
         (3) qsort(L, A, S),                   (2) L1 <= [],
         (4) io.write(S, !IO),                 (3) L2 <= []
                                               ;
         qsort(L, A, S) :-                     (4) L => [Le | Ls],
         (                                     (5) ( if X >= Le then
         (1) L => [],                          (6)     split(X, Ls, L11, L2),
         (2) S := A                            (7)     L1 <= [Le | L11]
         ;                                           else
         (3) L => [Le | Ls],                   (8)     split(X, Ls, L1, L21),
         (4) split(Le, Ls, L1, L2),            (9)     L2 <= [Le | L21]
         (5) qsort(L2, A, S2),                     )
         (6) A1 <= [Le | S2],                  ).
         (7) qsort(L1, A1, S)
         ).
\end{Verbatim}
\small
\caption{\emph{quicksort} program in superhomogeneous form.}
\normalsize
\label{fig:running:superhomogeneous}
\end{figure}
For readability, we have chosen meaningful names for some additional variables
that are added automatically by the Mercury compiler.
We also replace the sequence of unifications
needed to construct a single ground term
with a single goal.
For example, the list construction at (1) in \code{main}
in Figure~\ref{fig:running:superhomogeneous},
actually stands for
\begin{verbatim}
V_0 <= [],
V_1 <= 3, V_2 <= [V_1 | V_0],
V_3 <= 1, V_4 <= [V_3 | V_2],
V_5 <= 2, L <= [V_5 | V_4]
\end{verbatim}
These extra details are of no interest in this paper.

In the rest of the paper, we will ignore negation, since \code{not~$G$} can
be implemented as \code{if $G$ then fail else true},
where \code{fail} and \code{true} are two builtin goals,
with \code{fail} always failing and \code{true} always succeeding.
Note that in Mercury (unlike in Prolog),
the condition of an if-then-else is allowed to succeed several times.
Whether the condition of a particular if-then-else can do so
will be recorded in its determinism annotation,
and many parts of the compiler, including the RBMM implementation,
handle conditions of different determinisms differently.

Another situation in which determinism information is important
is existential quantification.
(Mercury also supports universal quantification,
but the compiler internally converts ${all}~[x_1,~\ldots,~x_n]~G$
to ${not}~{some}~[x_1,~\ldots,~x_n]~{not}~G$,
so we do not have to deal with it.)
If ${some}~[\ldots]~G$ quantifies away all the output variables of $G$,
then different solutions of $G$ would be indistinguishable,
so even if $G$ can have more than one solution, ${some}~[\ldots]~G$ will not.
We call such a quantification a \emph{commit},
and we handle commits differently from other quantifications.

\section{Overview of Region-Based Memory Management for Mercury}
\label{seCbgCmercuryRBMM}

We divide the task of realizing RBMM for Mercury into two parts:
(a) two static analyses and a program transformation,
which work entirely at compile time,
and (b) dynamic runtime support, which executes at runtime
code added to the program by the compiler at compile time.

The goal of the static analyses and transformation
is to annotate Mercury programs with information about regions.
An annotated program contains information about the regions
in which terms are constructed and when regions are created and freed.
To obtain this information, we first use
a region points-to analysis to detect the regions used by a program,
and then we compute the lifetimes of these regions
using a region liveness analysis.
The program transformation then uses these pieces of information
to convert the program into a region-annotated program.

The runtime support for RBMM has two main tasks.
First, it has to implement the necessary operations on regions:
the creation of regions, allocation into regions, and the removal of regions
(Section~\ref{seCsupportdet}).
Second, it has to provide support
for the interaction of backtracking with RBMM.
There are two main forms of interaction:
instant reclaiming and backward liveness
(Section~\ref{seCsupportnondet}).

The memory allocated by computations that have been backtracked over
will never be accessed again,
since backtracking effectively ``erases'' such computations.
To prevent memory leaks, this memory should be recovered
immediately when forward execution resumes again;
we call this \emph{instant reclaiming}.
This obviously has to be done at runtime,
so in our system,
the compiler inserts the code required to do this into the program
at both resume points
(points in the program where forward execution can resume after backtracking,
such as the starts of second and later disjuncts in a disjunction)
and at program points that establish resume points
(such as just before entry into a disjunction).

In logic programming languages, the presence of backtracking
requires the notion of liveness to be divided into two parts.
A variable, memory location or region
is \emph{forward live} at a program point
if it can be accessed during forward execution from that program point,
and it is \emph{backward live} at a program point
if it can be accessed during backward execution (i.e.\ after backtracking
to a choice point established \emph{before} that program point).
The two notions of liveness are independent:
all four combinations of forward and backward liveness and deadness
are possible.
Regions can be reclaimed only
when they are both forward dead and backward dead.

Our region liveness analysis takes into account \emph{only} forward liveness,
and we ensure safety with respect to backward liveness through runtime support.
Our reasons for why we handle backward liveness this way are that
\begin{itemize}
\item
handling it purely at compile time is not possible,
since runtime support will still be needed in some cases,
as we will point out in Section~\ref{seCfuture},
and a purely-runtime solution is simpler than a solution that mixes
compile time and runtime aspects; and
\item
we can implement a large part of this runtime support
using the machinery we need anyway for instant reclaiming.
\end{itemize}
However, handling backward liveness at least partially at compile time
may turn out to be more efficient,
which is why we intend to explore it in future work.

\subsection{Region Variables}

We use \emph{region variables} to refer to regions,
just as we use program variables to refer to values.
To allocate a new region,
we use the instruction \code{create(R)},
which creates a region and binds the region variable \code{R} to it.
To free a region we use the instruction \code{remove(R)},
which frees the memory of the region to which \code{R} is currently bound.
Our regions can and actually do live across procedure boundaries,
and thus we pass region variables as extra arguments to procedure calls.
Figure~\ref{fig:running:annotated} shows
the region-annotated \emph{quicksort} program after our region transformation.
Our source-to-source transform represents these instructions,
and the instructions we introduce later,
as calls to builtin predicates.
We describe the implementation of these predicates
in Section~\ref{seCsupportdet}.

\begin{figure}[htb]
\scriptsize
\begin{Verbatim}[frame=single,framerule=0.2pt,framesep=3pt]
   main(!IO) :-                              split(X, L@R1, L1@R3, L2@R4) :-
       create(R20), create(R21),             (
   (1) L <= [2, 1, 3] in R20,                (1) L => [],
       create(R22),                              remove(R1),
   (2) A <= [] in R22,                           create(R3),
   (3) qsort(L@R20, A@R22, S@R22),           (2) L1 <= [] in R3,
   (4) io.write(S, !IO),                         create(R4),
       remove(R21), remove(R22).             (3) L2 <= [] in R4
                                             ;
   qsort(L@R6, A@R8, S@R8) :-                (4) L => [Le | Ls],
   (                                         (5) ( if X >= Le then
   (1) L => [],                              (6)      split(X, Ls@R1, L11@R3, L2@R4),
       remove(R6),                           (7)      L1 <= [Le | L11] in R3
   (2) S := A                                      else
   ;                                         (8)      split(X, Ls@R1, L1@R3, L21@R4),
   (3) L => [Le | Ls],                       (9)      L2 <= [Le | L21] in R4
   (4) split(Le, Ls@R6, L1@R9, L2@R10),          )
   (5) qsort(L2@R10, A@R8, S2@R8),           ).
   (6) A1 <= [Le | S2] in R8,
   (7) qsort(L1@R9, A1@R8, S@R8)
   ).
\end{Verbatim}
\small
\caption{Region-annotated \emph{quicksort} program.}
\normalsize
\label{fig:running:annotated}
\end{figure}

In the region-annotated code,
we use the postfix \code{@Ri} to annotate both actual and formal arguments
with their region variables.
We also annotate each unification that constructs a new memory cell
with the region in which the cell will be allocated.
For example, in \code{main}, the skeleton of the list \code{L}
is in the region (bound to) \code{R20},
while that of the accumulator \code{A} is in \code{R22}.
The elements of the lists are in \code{R21} (but see below).
In the call to \code{qsort},
\code{R20} and \code{R22} are passed as actual region arguments,
corresponding to the formal arguments \code{R6} and \code{R8}
in the definition of \code{qsort}.
We do not need to pass the region of the elements because
\code{qsort} and \code{split} just read from it.
The region \code{R20} is passed to \code{qsort} from \code{main} and
is removed in the base case branch of \code{split}
in the call to \code{split} at (4) in \code{qsort}.
The two new lists \code{L1} and
\code{L2} are allocated in two separate regions referred to by \code{R9}
and \code{R10}.
These regions are created by the base case branch of \code{split},
and removed (indirectly) by the recursive calls to \code{qsort} at (5) and (7).
If \code{L1} and \code{L2} are empty lists,
the removals will happen in the base case branch of \code{qsort};
otherwise, they will happen in the base case branch of \code{split}.
The region \code{R22} of the resulting list is the region of the accumulator,
which is created in \code{main}.

\section{Region Modelling}
\label{seCregionCmodel}

\subsection{Storing Terms in Regions Based on Their Types}

As we want to distribute terms over different regions,
we first discuss the representation of terms
when the heap memory is divided into regions.

We assume that a term that does not fit into one word
will be represented by a pointer to a memory cell on the heap.
We also assume that a term that can be represented by a single memory word
does not need storage on the heap in its own right.
When those terms are on their own,
they will be stored in registers or in stack slots.
When they are arguments of a larger term,
they will be stored in a word on the heap,
but this word will be counted as
belonging to the memory cell representing the larger term.

Our assumptions are compatible with the implementation of Mercury
in the Melbourne Mercury Compiler (MMC).
The MMC knows the types of all variables,
and these types give us information about the storage size of terms.
Terms of primitive types such as int and char are stored in one word,
and the same is true of enumeration types
(types in which \emph{all} functors have arity zero).
The principal functor of a term that needs heap space is represented
by a \emph{possibly-tagged} pointer to a block of memory words on the heap.
The compiler knows all the functors in the type of the term.
It also knows that all words in the Mercury heap are aligned,
so pointers to them have two free bits on 32 bit machines,
and three free bits on 64 bit machines.
Therefore if a type has at most four function symbols
(eight on 64 bit machines),
the principal functor can be represented
by what Mercury calls a ``primary tag'' on the lowest bits of the pointer.
When a type has only one functor, even this is not needed.
When a type has more than four or eight functors
(on 32 and 64 bit machines respectively)
the compiler will use one primary tag value
to represent several function symbols,
and will use the first word of the pointed-to memory block
as a secondary tag to distinguish between them.
(The usual implementations of Prolog have a similar word
in \emph{every} heap cell other than those storing lists,
increasing their memory footprint.)

\begin{example} Consider the following types.\\
    \code{:- type elem ---> f; g(int); h(list\_int, int)}.\\
    \code{:- type list\_elem ---> []; [elem | list\_elem]}.\\
Figure~\ref{fig:termrep} shows MMC's representation
of the term \code{[f, g(1), h([1, 2], 2)]} bound to the variable \code{L},
which is of type \code{list\_elem}.
Boxes with slim border are locations on the stack or in registers,
while boxes with bold borders are locations on the heap.
Note the representation of the term \code{h([1, 2], 2)}
in the last element of the list:
we need a two-word block for \code{h}'s arguments,
but the functor itself is stored implicitly in the tagged pointer.
\hfill $\Box$
\label{example:termrep}
\end{example}
\begin{figure}[htp]
\centering
\includegraphics[scale=1]{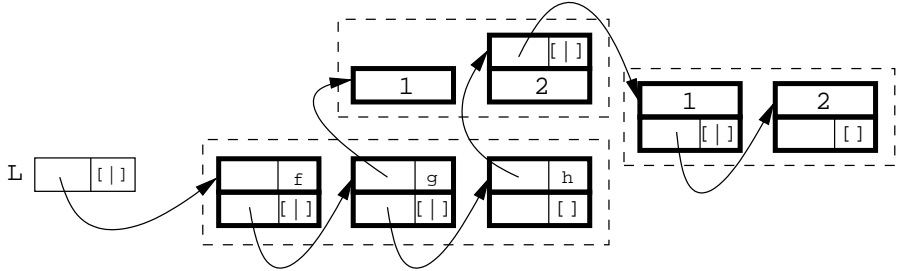}
\caption{Term representation of \code{L}=[f, g(1), h([1, 2], 2)].}
\label{fig:termrep}
\end{figure}

We now consider the storage of terms when the heap is split into regions.
The idea is to use different regions to store different parts of a term
so that we can reclaim the memory of a part by destroying its region
as soon as that part becomes dead.
Many programs (including \emph{quicksort})
create temporary lists
in which the elements have much longer lifetimes.
Therefore storing the elements and the list skeletons in different regions
will allow us to recover the memory of the list skeletons much earlier.
Generalizing from this, we divide each term into regions
based on the type of each of its subterms.
We will develop this idea in the next subsection.

In Figure~\ref{fig:termrep},
the regions used to store our example term are shown by the dashed lines.
We put the two-word memory blocks
making up the skeleton of the list \code{L} into one region
because they have the type \code{list\_elem}.
We also put \emph{all} the elements,
which have the type \code{elem},
into another region.
Finally, the first subterm of the third element,
which is of type \code{list\_int} rather than \code{list\_elem},
is stored in yet another region.

The representation of the list of integers here
seems inconsistent with what we said in Section~\ref{seCbgCmercuryRBMM},
where we have an extra, separate region for the integers.
The reason for this is because in this section we want to give a region model
as close as possible to the implementation of Mercury in the MMC,
in which integers do not need their own memory cells on the heap.
Here we have two different viewpoints:
a theoretical one that wants to treat all types the same way,
and a practical one that wants to accurately reflect
how the implementation handles values of each type.
For convenience,
we take the liberty of switching between the two viewpoints at will.
When talking about theoretical topics
such as static analyses and transformation for convenience
we generally assume that all types (including \code{int}) require heap storage;
when talking about the actual implementation,
we will assume that the implementation does not create regions
without having anything to put into them.
We will be more specific only if the context is not clear.

\subsection{Modelling Regions of a Type}

Our system needs a storage scheme
that specifies how the terms of a type are stored.
Consider a type \code{t} declared as follows.
\begin{verbatim}
:- type t ---> ...; f(t1,..., ti,..., tn); ...
\end{verbatim}
We associate a region variable $R^t$ with the type.
The block of memory words corresponding to a principal functor,
such as \code{f}, of a term of type \code{t}
is stored in the region bound to $R^t$.
In the rest of the paper we abbreviate this by simply saying that
a principal functor is stored in $R^t$.
The principal functor of an argument of \code{f} that has type \code{ti}
is stored in the region bound to $R^{ti}$,
which is associated to \code{ti}.
If a type \code{t} is recursive or mutually recursive,
we still use only one region variable $R^t$.
This implies that any term of a recursive type
is modelled by a finite number of regions.

We model the storage scheme using a type-based region graph,
$\TG(N,E)$ with $N$ being a set of nodes and $E$ being a set of directed edges.
A node stands for a region variable.
A directed edge from one node to another represents the fact that
the region bound to the region variable
represented by the source node of the edge
contains references into (points-to)
the region bound to the region variable
represented by the target node of the edge.
The reference relation represented by the edges
is actually defined by the type.

Consider the type-based region graph of the type \code{t}, $\TG_t$,
with the region variables $R^t$, $R^{t1}$, $R^{t2}$ and so on.
If $R^t$ is represented by the node $n$,
then for each node $m$ representing $R^{ti}$,
we have exactly one edge $(n, (f,i), m)$ with the label $(f,i)$.
We refer to $n$ as the \emph{principal node} of $\TG_t$.

\begin{example}
The type-based region graph for the type \code{list\_elem}
in Example~\ref{example:termrep}
is shown in Figure~\ref{fig:typebased:listelem}.
The \code{[|]} principal functor is stored in $R^{\mathit{list\_elem}}$.
It has two arguments, the first having the type \code{elem} and
the second having the same type \code{list\_elem}.
Thus we have two edges from $R^{\mathit{list\_elem}}$,
the first pointing to $R^{\mathit{elem}}$ where the principal functors of
\code{elem} (\code{g/1} and \code{h/2}) are stored,
and the second being a self-edge.
The edge labelled \code{(h,1)} is due to the first argument of the functor
\code{h/2}.
The reader may want to compare this type-based region graph
with Figure~\ref{fig:termrep},
which shows the memory representation of a term of this type.
\begin{figure}[htp]
\centering
\includegraphics{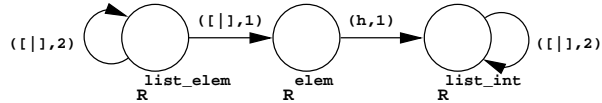}
\caption{The type-based region graph of the type \code{list\_elem}.}
\label{fig:typebased:listelem}
\end{figure}
\hfill $\Box$
\label{example:typebasedrptg}
\end{example}
\begin{example}
Consider the following types \code{t1} and \code{t2},
which are mutually recursive.
\begin{verbatim}
    :- type t1 ---> f(int, t2).
    :- type t2 ---> g(t1, int) ; h.
\end{verbatim}
The type-based region graph for these types
is shown in Figure~\ref{fig:typebased:mutually}.
\begin{figure}[htp]
\centering
\includegraphics{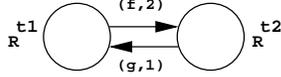}
\caption{Type-based region graph of mutually recursive types.}
\label{fig:typebased:mutually}
\end{figure}
\end{example}

\subsection{Region Points-To Graph}

Now that we have the region model for types,
our next goal is to model the memory used by a Mercury program
in terms of regions.
A program consists of a set of procedures,
each having its own set of program variables
that, at runtime, are instantiated with relevant terms.
Therefore we define the notion of a region points-to graph
that models the memory used by a set of variables.
The memory used by a procedure
is modelled by a region points-to graph for its variables.
Finally, the memory model for the whole program
is expressed through the region points-to graphs of its procedures.

In Mercury, variables are instantiated by unifications.
A construction unification \code{X <= f(..., Y, ...)}
allocates new memory for storing the functor \code{f}
(actually the block of memory words storing \code{f}'s arguments,
and, if the tag on the pointer to the block is not enough for this,
\code{f}'s identity),
and creates sharing between \code{X} and each \code{Y}.
In a deconstruction unification \code{X => f(..., Y, ...)}
or an assignment unification \code{Y := X},
\code{Y} is instantiated and shares a subterm or the whole term with \code{X},
respectively.
Hence the region points-to graphs should capture
the memory locations of the variables and the sharing among them.

A {\bf region points-to graph}, $G(N, E)$, for a set of variables $V$,
consists of a set of nodes $N$, representing region variables and
a set of directed edges, $E$, representing references between the regions
bound to these region variables.
The edges here serve exactly the same purpose as those in a $\TG$ graph.
However, each node $n$ in the region points-to graph
has an associated set of program variables, $\vars{n}$,
whose principal functors are stored in the region
that is bound to the region variable that is represented by $n$.
The $\varset$ sets of the various nodes
must represent a partition of the set of variables of interest
(e.g. the set of variables in a procedure):
each variable in the set must appear in the $\varset$ set of exactly one node.
(Note that the $\varset$ set of a node may be empty;
this can happen when a variable's value has some subterms
that the code in question does not access.)
We have $V = \hbox{\raise-1mm\hbox{$\textstyle
\bigcup \atop \scriptstyle {n \in N}$}} \vars{n}$.
The notation $n_X$ denotes the node where \code{X} $\in \vars{n_X}$
and we refer to $n_X$ as the \emph{location} of \code{X},
since this node represents the region where
the principal functor of the term that \code{X} is bound to is stored.
The function $\mathit{node}(n_X, (f, i))$
returns the node $m$ if $(n_X, (f,i), m) \in E$,
otherwise its result is undefined.

\emph{Sharing} is represented in a region points-to graph in two ways.
First, directed edges represent the sharing of subterms,
and second, a node whose $\varset$ set contains more than one variable
represents the fact that these variables may be bound to the same term.
An example of the latter is given by
the variables of an assignment unification:
they are bound to the same term and
therefore they should be in the $\varset$ set of the same node.
A region points-to graph represents sharing at the level of the regions.

\begin{definition}[Region-sharing in a region points-to graph]
\label{definition:region-sharing}
Two variables X and Y {\bf region-share} in a region points-to graph
if there exists a node that can be reached from both $n_X$ and $n_Y$.
\end{definition}

For convenience, we also say \emph{a node represents a region},
by which we mean the region to which
the region variable represented by the node is bound at runtime.
Then we can say \emph{a functor is stored in a node} meaning
that the functor (i.e.\ the memory block corresponding to it)
is stored in the region represented by the node.

For a procedure $p$,
we denote its region points-to graph by $G_p(N_p, E_p)$.
$G_p$ should represent the locations and sharing
among all the variables in $p$.
It is possible to form a region points-to graph for a procedure
exactly from the type-based region graphs of all of its variables
(whose types are known to the compiler).
Although this region points-to graph adequately models
the locations of the procedure's relevant terms,
it does not represent the sharing among them.
Actually, as we will see in Section~\ref{seCrpta},
we use that region points-to graph as the starting point
in our region points-to analysis of a procedure,
with the ultimate aim of producing a region points-to graph
that also represents all the \emph{possible} sharing
among the procedure's variables.

\begin{example}
Consider the following sequence of code to construct the term
that \code{L} in Example~\ref{example:termrep} is bound to.
The type of \code{K} is of no importance.
\begin{verbatim}
    ...,
    X <= [1, 2],
    Y := X,
    Z <= h(Y, 2),
    L <= [f, g(1), Z],
    K <= k(Z),
    ...
\end{verbatim}
The region points-to graph that represents
the memory manipulated by this sequence is shown in Figure~\ref{fig:sharing}.
\code{X} and \code{Y} are in the $\varset$ set of the same node
because the assignment makes \code{Y} point to the term
to which \code{X} is bound.
The direct sharing between \code{Z} and \code{Y},
and between \code{L} and \code{Z},
is represented by the edges between their corresponding nodes.
The indirect sharing between \code{L} and \code{Y} is modelled
by the fact that $n_Y$ is reachable from $n_L$ through the directed edges.
The sharing between \code{L} and \code{K} is represented by the fact that
$n_Z$ is reachable from both $n_L$ and $n_K$.
\begin{figure}[htp]
\centering
\includegraphics{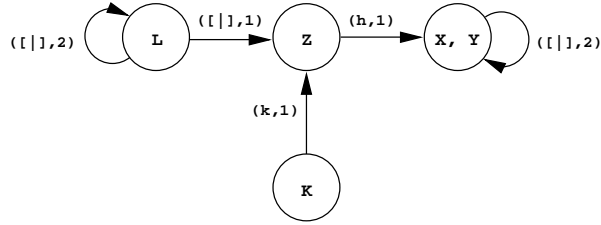}
\caption{Modelling of sharing information.}
\label{fig:sharing}
\end{figure}
\label{example:sharingrptg}
\hfill $\Box$
\end{example}

\section{Region Points-To Analysis}
\label{seCrpta}

The region points-to analysis aims at computing
for each procedure in a Mercury program
a region points-to graph that represents
the locations of its variables and the sharing among them.

The region points-to analysis is unification-based and flow-insensitive,
i.e.\ the execution order of the atomic goals in a procedure does not matter,
and consists of an intraprocedural analysis and an interprocedural analysis.
Both analyses make use of the unify operation
shown in Algorithm~\ref{algo:unify},
whose task is capture sharing between two nodes in a region points-to graph.
This algorithm should be invoked when the analyses learn
that two variables whose nodes are $n$ and $m$ respectively
can refer to the same storage;
it will update the points-to graph by unifying the two nodes,
i.e.\ merging them into one.
To ensure that
there is only one out-edge with a specific label from any given node,
unifying two nodes will cause
their corresponding child nodes to be unified as well,
unless they are the same node already.

\algsetup{indent=2em}
\begin{algorithm}
    \small
    \caption{$\mathit{unify}(n, m)$}
    \label{algo:unify}
\begin{algorithmic}
    \REQUIRE{$G(N,E)$, $n, m \in N$.}
    \ENSURE{$G(N,E)$ with $n$ representing the unified node.}

    \STATE $N = N \setminus \{m\}$

    \STATE $\vars{n} = \vars{n} \cup \vars{m}$

    \FORALL{$(m, (f, i), k) \in E$}
        \STATE $E = E \setminus \{(m,(f,i),k)\}$
        \IF{$ (n, (f, i), k) \not\in E$}
            \STATE $E = E \cup \{(n,(f,i),k)\}$
        \ENDIF
    \ENDFOR

    \FORALL{$(k, (f, i), m) \in E$}
        \STATE $E = E \setminus \{(k,(f,i),m)\}$
        \IF{$(k, (f, i), n) \not\in E$}
            \STATE $E = E \cup \{(k,(f,i),n)\}$
        \ENDIF
    \ENDFOR

    \FORALL{${l}, {l}' \in N$}
        \IF{$(n, (g, j), {l}) \in E \wedge (n, (g, j), {l}') \in E
        \wedge {l} \neq {l}'$}
            \STATE $\mathit{unify}({l}, {l}')$
        \ENDIF
    \ENDFOR

\end{algorithmic}
\normalsize
\end{algorithm}

We will describe the analyses in turn
with the assumption that we are analyzing a procedure $p$.

Recall that, when describing the static region analysis and transformation,
for convenience, we make the assumption that
\emph{all} terms are stored on the heap and therefore we need regions for them.
In a concrete implementation, such as ours inside the MMC
(Sections~\ref{seCsupportdet} and \ref{seCsupportnondet}),
if certain terms do not need heap storage,
their corresponding regions can just be ignored.

\subsection{Intraprocedural Analysis of a Procedure}

The intraprocedural analysis initializes $G_p$
and then captures the sharing created by the explicit unifications.
Its definition is in Algorithm~\ref{algo:intraproc}.
(See section \ref{seCbgCinternal} for the definition of superhomogeneous form.)
\begin{algorithm}
    \small
    \caption{$\mathit{intraproc}(p)$:
        intraprocedural analysis of a procedure $p$}
    \label{algo:intraproc}
\begin{algorithmic}
    \REQUIRE{$p$ is in superhomogeneous form.}
    \ENSURE{The sharing created by explicit unifications
        is represented in $G_p$.}

    \STATE $G_p = (\emptyset, \emptyset)$
    \FORALL{$\code{X} \in p$}
        \STATE $G_p = G_p \uplus \mathit{init\_rptg}(\code{X})$
    \ENDFOR

    \FORALL{$\mathit{unif} \in p$}
        \IF{$\mathit{unif} \equiv (\code{X:=Y})$}
            \STATE $\mathit{unify}(n_X, n_Y)$
        \ELSIF{$\mathit{unif} \equiv (\code{X = > f(\args[Y]{n})}\hspace{1mm}
        \mathit{or}\hspace{1mm}\code{X < = f(\args[Y]{n})})$}
            \FOR{$i=1$ to $n$}
                \STATE $\mathit{unify}(\mathit{node}(n_X, (f, i)), n_{Y_i}$)
            \ENDFOR
        \ENDIF
    \ENDFOR

\end{algorithmic}
\normalsize
\end{algorithm}

As we know the type of each variable in $p$,
we initialize $G_p$ by using
the $\TG$ graphs of the variables.
In Algorithm~\ref{algo:intraproc},
we use a function $\mathit{init\_rptg}(\code{X})$
that
\begin{itemize}
\item
    generates a region points-to graph for \code{X}
    from the type-based region graph of the type of \code{X},
    $\TG_{\mathit{type}(\code{X})}$,
\item
    sets the $\varset$ set of the node
    corresponding to the principal node in $\TG_{\mathit{type}(\code{X})}$
    to $\{\code{X}\}$
    and the $\varset$ set of all others nodes to the empty set,
\item
    and generates a fresh region variable
    for each node in the region points-to graph.
\end{itemize}

The intraprocedural analysis then adds to $G_p$
all the sharing created by the unifications in the procedure.
For assignment, construction and deconstruction unifications
we unify the nodes corresponding with the sharing created by them.
We ignore test unifications because they do not create any sharing.

\subsection{Interprocedural Analysis}

The interprocedural analysis, Algorithm~\ref{algo:interproc},
updates $G_p$ by integrating into it
the \emph{relevant} region-sharing information
from the region points-to graphs of the called procedures.

\begin{algorithm}
    \small
    \caption{$\mathit{interproc}(p)$:
    interprocedural analysis of a procedure $p$}
    \label{algo:interproc}
\begin{algorithmic}
    \REQUIRE{$p$ is in superhomogeneous form.}
    \ENSURE{The sharing created by procedure calls is represented in
    $G_p(N_p,E_p)$.}

    \REPEAT
        \FORALL{call sites in $p$}
            \STATE Assume that the call is $q(\args[Y]{n})$,
            with \args[X]{n} being the corresponding formal arguments,
            and that $G_q$ is available.

            \STATE
            \STATE \emph{\% Build an $\alpha$ relation.}
            \FOR{$k=1$ to $n$}
                \STATE $\alpha(n_{X_k}) = n_{Y_k}$
            \ENDFOR

            \STATE
            \STATE \emph{\% Ensure $\alpha$ is a function.}
            \FORALL{$X_i, X_j$}
                \IF{$\alpha(n_{X_i}) = n_{Y_i} \wedge
                    \alpha(n_{X_j}) = n_{Y_j}
                    \wedge n_{X_i} = n_{X_j}
                    \wedge n_{Y_i} \not= n_{Y_j}$}
                    \STATE  $\mathit{unify}(n_{Y_i}, n_{Y_j})$
                \ENDIF
            \ENDFOR

            \STATE
            \STATE \emph{\% Integrate sharing in $G_q$ into $G_p$.}
            \STATE In the graph $G_q$, do a depth-first traversal
                starting from each $n_{X_i}$,
                visiting each node only once
                and applying the rules P1 and P2 in
                Figure~\ref{fig:inter:rules} when applicable.
        \ENDFOR
    \UNTIL{There is no change in either $G_p$
    or in any of the $\alpha$ functions.}

\end{algorithmic}
\normalsize
\end{algorithm}

\sloppy
Consider a call $q(\args[Y]{n})$ in the body of $p$,
with the head of the called procedure being $q(\args[X]{n})$.
Any region-sharing among the $X_i$ in $G_q$
may not currently be present in $G_p$ as region-sharing among the $Y_i$.
The interprocedural analysis makes sure that
any such sharing in $G_q$ will be copied to $G_p$.
First, it builds the function $\alpha:~N_q \rightarrow N_p$ that
maps the nodes of the formal arguments ($X_i$'s)
to the nodes of the corresponding actual arguments ($Y_i$'s).
Then these nodes are the starting points for the
integration of the remaining region-sharing.
This is done by following the relevant edges in $G_q$
to extend the $\alpha$ function to all the relevant nodes in $G_q$ (rule P2)
and to unify the relevant nodes in $G_p$ (rule P1).

\begin{figure}
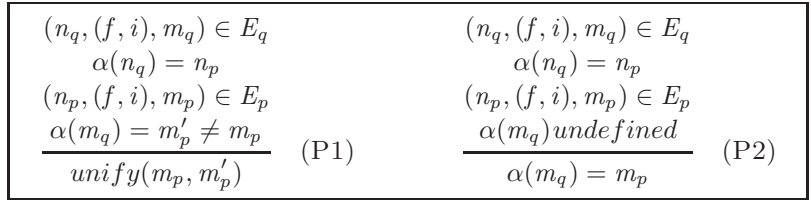

\centering
\framebox{
$
\begin{array}{c}
\inferrule
{
    (n_q, (f, i), m_q) \in E_q \\\\
    \alpha(n_q)=n_p \\\\
    (n_p, (f, i), m_p) \in E_p \\\\
    \alpha(m_q)=m'_p \neq m_p
}
{
    \mathit{unify}(m_p, m'_p)
}
\quad(\textsc{P1})
\qquad
\qquad
\inferrule
{
    (n_q, (f, i), m_q) \in E_q \\\\
    \alpha(n_q)=n_p \\\\
    (n_p, (f, i), m_p) \in E_p \\\\
    \alpha(m_q)~\mathit{undefined}
}
{
    \alpha(m_q) = m_p
}
\quad(\textsc{P2})
\end{array}
$
}
\small
\caption{Interprocedural analysis rules.}
\normalsize
\label{fig:inter:rules}
\end{figure}

For a whole program,
we start by performing the intraprocedural analysis for every procedure.
Since our interprocedural analysis propagates information only upwards,
from the graphs of callees to those of callers,
we compute the strongly connected components of the call-dependency graph
and analyze the components in bottom-up order.
Algorithm~\ref{algo:rpta} illustrates this approach.
\begin{algorithm}
    \small
    \caption{Region points-to analysis of a program}
    \label{algo:rpta}
\begin{algorithmic}
    \REQUIRE{A Mercury program $P$
    with its procedures in superhomogeneous form.}
    \ENSURE{Region points-to graphs for all procedures.}

    \FORALL{procedure $p$ in $P$}
    \STATE $\mathit{intraproc}(p)$
    \ENDFOR

    \STATE Compute the strongly connected components (SCCs)
    of $P$'s call-dependency graph.

    \FORALL{SCCs in bottom-up order}
        \REPEAT
            \FORALL{$p$ in SCC}
                \STATE $\mathit{interproc}(p)$
            \ENDFOR
        \UNTIL{we have reached a fixpoint}
    \ENDFOR
\end{algorithmic}
\normalsize
\end{algorithm}

The points-to graphs of the \code{split} and \code{qsort} procedures
in the \emph{quicksort} program in Example~\ref{example:running:normal}
are shown in Figure~\ref{fig:running:rptgraph}.
For \code{split}, the region points-to analysis detects that the two
sublists \code{L1} and \code{L2} can be in separate regions that are
different from the region of the input list \code{L}.
For \code{qsort}, the input list, the two temporary lists, and the resulting
list are all in different regions.
That the resulting list \code{S} is in the same region
as the accumulator and the temporary lists \code{S2} and \code{A1}
is reasonable because the result list is gradually built up from them.
\begin{figure}
    \centering
    \subfloat[\code{split}\label{fig:example:split}]
        {\includegraphics{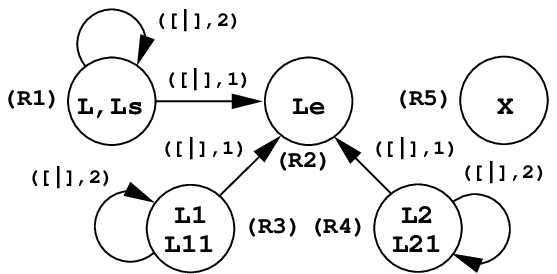}}
    \quad
    \subfloat[\code{qsort}\label{fig:example:qsort}]
        {\includegraphics{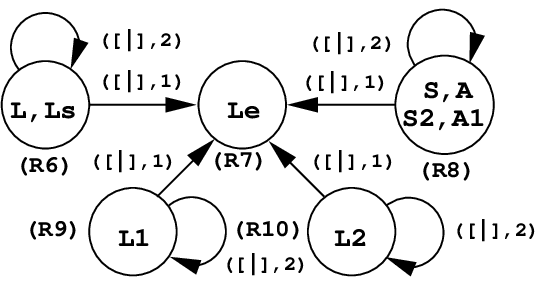}}
    \small
    \caption{The region points-to graphs of \code{split} and \code{qsort}.}
    \normalsize
    \label{fig:running:rptgraph}
\end{figure}

\subsection{Correctness of the Region Points-To Graphs}

We will prove that the region points-to analysis of a program terminates
and that the resulting region points-to graphs
for the procedures in the program are \emph{correct},
i.e.\ they represent all the locations of the terms
and the sharing among the terms.

\begin{theorem}
\label{theorem:termination}
The region points-to analysis of a program terminates.
\end{theorem}
\begin{proof}
An $\alpha$ function at a call site is a mapping
from a subset of the nodes in the callee's region points-to graph
to a subset of the nodes in the caller's region points-to graph.
Therefore if we can show that the sets of nodes are finite
then so is the $\alpha$ function.
This then implies that the termination
of the region points-to analysis solely depends on
the finiteness of the region points-to graphs.

For any procedure, the Algorithm~\ref{algo:intraproc}
starts with a region points-to graph
having a finite number of nodes and edges.
The analysis uses only the $\mathit{unify}$ operation
(Algorithm~\ref{algo:unify}) to change the graphs.
This always decreases the number of nodes
and never increases the number of edges.
Therefore the analysis must, at some point, terminate.
In the extreme case,
the final region points-to graph of a procedure
contains only one node and maybe some self-edges.
\end{proof}

\begin{theorem}
\label{theorem:loc_and_sharing}
The graphs that result from the region points-to analysis of a program
represent all the locations of the terms
that can possibly be constructed during the execution of the program,
and the possible sharing among the terms.
\end{theorem}

The theorem has two parts, one about locations and the other about sharing.
We prove each part separately.
\begin{proof}[Proof (Locations)]
During the execution of a program,
a variable can get bound to a compound term.
However, that compound term must be built step-by-step
using construction unifications.
In such a step, a construction unification allocates memory
to store only the principal functor that
the variable on its left-hand side is bound to.
Therefore to show that the graphs represent
all the locations of a compound term,
it suffices to show that the graphs represent the locations
of the variables in the left-hand sides of construction unifications.

Consider a procedure.
The region points-to analysis of the procedure starts with
the intraprocedural analysis (Algorithm~\ref{algo:intraproc}) that
assigns a set of nodes to each variable based on
the type-based region graph of the type of the variable.
These nodes represent the regions where
a term to which the variable is possibly bound is stored.
Moreover, the variable is assigned a location by the fact that
it is added to the $\varset$ set of the node
where the principal functor of the term it is bound to is stored.
During the analysis, this node may be removed from the graph
when it is unified with another node.
However, regardless of where this happens,
in the intraprocedural or in the interprocedural analysis,
the $\mathit{unify}$ operation ensures that
the remaining node now represents the location of the variable.
\end{proof}

Now, for the second part of Theorem~\ref{theorem:loc_and_sharing},
we will show that
all sharing between the terms is represented in the region points-to graphs.
For a procedure, the sharing among its variables is created
either by explicit unifications in the procedure,
or by unifications hidden inside the procedures it calls.
The lemma below deals with explicit unifications.

\begin{lemma}[Sharing created by explicit unifications]
\label{lemma:sharing_unifs}
    If a unification explicitly appears in a procedure,
    the sharing created by the unification
    is represented in the region points-to graph of the procedure.
\end{lemma}
\begin{proof}
Explicit unifications are handled by Algorithm~\ref{algo:intraproc},
the intraprocedural analysis.
Test unifications do not create sharing, so we can ignore them.
Consider an assignment unification.
Algorithm~\ref{algo:intraproc} unifies
the nodes of its left and right variables and keeps
these two variables in the $\varset$ set of the unified node.
This represents their sharing.

The only remaining form of unification in Mercury's superhomogeneous form
is $X=f(\ldots,X_i,\ldots)$.
When processing such a unification,
Algorithm~\ref{algo:intraproc} calls $\mathit{unify}(m, n_{X_i})$
where $m = \mathit{node}(n_X, (f, i))$.
This adds $X_i$ to $\vars{m}$.
After the unification, the edge $(n_X,(f,i),m)$,
which was already in the region points-to graph,
has become $(n_X,(f,i),n_{X_i})$.
This represents the sharing between $X$ and $X_i$.
\end{proof}

For procedure calls, we consider a procedure $p$ that invokes $q$.
As before, use $X_i$ to denote the formal parameters
and $Y_i$ to denote the actual parameters.
We call
$G^{sub}_p(N^{sub}_p,E^{sub}_p)$ the subgraph
of the region points-to graph of $p$ rooted at the nodes of the $Y_i$s and
$G^{sub}_q(N^{sub}_q,E^{sub}_q)$ the subgraph
of the region points-to graph of $q$ rooted at the nodes of the $X_i$s.

In order to prove that
all the region-sharing in $G^{sub}_q$ is also in $G^{sub}_p$,
we consider two arbitrary formal arguments $X_i$ and $X_j$ that share.
By Definition~\ref{definition:region-sharing},
this means that there exists a node in $G^{sub}_q$
that can be reached from both $n_{X_i}$ and $n_{X_j}$.
There are two cases.
Either $n_{X_i}=n_{X_j}$, which means that
the sharing between $X_i$ and $X_j$ is represented in $G^{sub}_q$
by them both being in the $\varset$ set of the same node,
or $n_{X_i} \not= n_{X_j}$,
which means that the sharing between them
is represented by some node being reachable from both of them.

The following lemma shows that
region-sharing of the first kind in $G^{sub}_q$
is also reflected in $G^{sub}_p$.
\begin{lemma}
\label{lemma:sharing_same_node}
The region-sharing between the formal arguments that are in the
$\varset$ set of a node $n_q \in N^{sub}_q$ is also in $G^{sub}_p$.
\end{lemma}
\begin{proof}
The interprocedural analysis (Algorithm~\ref{algo:interproc})
first builds an $\alpha$ relation
that represents the connections between $G^{sub}_q$ and $G^{sub}_p$.
The initial $\alpha$ relation connects
the nodes of the formal arguments with
the nodes of the corresponding actual arguments.
In this $\alpha$ relation, it is possible that a node in $G^{sub}_q$
whose $\varset$ set contains more than one formal argument
is connected to more than one node of the actual arguments in $G^{sub}_p$.
The region-sharing of such formal arguments
(represented by the fact that they are in the same $\varset$ set)
is brought into $G^{sub}_p$ when Algorithm~\ref{algo:interproc}
unifies all the nodes in $G^{sub}_p$
that any single node in $G^{sub}_q$ is related to.
This ensures that the actual arguments corresponding to the formal arguments
that are in the $\varset$ set of a node $n_q$ in $G^{sub}_q$
will be in the $\varset$ set of a single node $n_p$ in $G^{sub}_p$,
with $\alpha(n_q) = n_p$.
\end{proof}

For the region-sharing of the second kind,
we first introduce the following lemma.
\begin{lemma}
\label{lemma:sharing_diff_nodes}
If $n$ and $m$ are in $N^{sub}_q$ such that
$(n, (f,i), m) \in E^{sub}_q$
and $\alpha(n) \in N^{sub}_p$,
then $\alpha(m) \in N^{sub}_p$
and also $(\alpha(n), (f,i), \alpha(m)) \in E^{sub}_p$.
\end{lemma}
\begin{proof}
In a well-typed Mercury program,
an actual argument must have the same type
as the corresponding formal parameter.
Therefore if $(n, (f,i), m)$ is in $E^{\mathit{sub}}_q$,
then there must exist a node $k \in N^{\mathit{sub}}_p$ such that
$(\alpha(n), (f,i), k)$ is in $E^{\mathit{sub}}_p$.
If $\alpha(m) = k$, our proof is done.
If $\alpha(m) = m' \not= k$,
then Algorithm~\ref{algo:interproc}
applies rule P1 to unify $k$ and $m'$,
after which again we have $\alpha(m) = k$.
If $\alpha(m)$ is undefined,
the algorithm applies rule P2 to produce $\alpha(m) = k$.
\end{proof}

Lemma~\ref{lemma:sharing_diff_nodes} essentially shows that
the $\alpha$ function extends to all the nodes in $N^{sub}_q$
reachable from the formal parameters,
and that all the edges connecting these nodes in $E^{sub}_q$
have their counterparts in $G^{sub}_p$.

\begin{theorem}[Sharing created by procedure calls]
\label{theorem:sharing_calls}
All the region-sharing in $G^{sub}_q$ is also in $G^{sub}_p$.
\end{theorem}
\begin{proof}[Proof (Sharing created by procedure calls)]
The proof of Theorem~\ref{theorem:sharing_calls} follows from
Lemmas~\ref{lemma:sharing_same_node} and~\ref{lemma:sharing_diff_nodes}.
\end{proof}

Note that
in recursive procedures, where the caller and the callee are the same,
one invocation of interprocedural analysis (Algorithm~\ref{algo:interproc})
will not necessarily be sufficient to reflect all sharing
from $G^{sub}_q$ to $G^{sub}_p$,
since in that case the very act of updating $G^{sub}_p$
updates $G^{sub}_q$ as well.
This is why Algorithm~\ref{algo:rpta} does a fixpoint iteration.

Now we can continue with the proof of the sharing-among-terms part of
Theorem~\ref{theorem:loc_and_sharing}.
\begin{proof}[Proof (Sharing among terms)]
The proof of the second part of
Theorem~\ref{theorem:loc_and_sharing}
follows from Lemma~\ref{lemma:sharing_unifs} and
Theorem~\ref{theorem:sharing_calls},
which show that the sharing created by explicit unifications
as well as by procedure calls in a procedure is all represented in
the region points-to graph of the procedure.

When a procedure is recursive or mutually recursive,
it is possible that the region points-to graph of a called procedure
(recursive or mutually recursive) has not fully represented the sharing
among its formal arguments.
However, if a program ever creates sharing,
ultimately this creation must involve a unification.
Lemma~\ref{lemma:sharing_unifs} shows that
this sharing is represented in the region points-to graph of
the procedure containing the unification,
and Theorem~\ref{theorem:sharing_calls} shows that
the sharing will also be represented in
the region points-to graphs of any procedures that invoke the procedure.
\end{proof}

In the rest of the paper,
when we mention region points-to graphs,
we mean the ones obtained by the region points-to analysis of the program.

\subsection{Regions that a Procedure Allocates Into}
\label{seCallocCregions}

During the region points-to analysis of a procedure,
we can track the regions
that are \emph{possibly} allocated into in the procedure.
A construction unification is the only construct in Mercury
that allocates memory.
When processing a construction unification \code{X <= f(\ldots)}
we mark the node $n_X$ as \emph{allocated}.
When two nodes are unified, if one node is marked as allocated
then the unified node is also marked as allocated.
At a call site,
if a node $n$ reachable from a formal parameter in the callee
is marked as allocated,
and $\alpha(n) = m$,
then we mark $m$ in the caller as allocated as well.
We call the set of nodes in procedure $p$ marked in this way \allocation{p}.
In the \emph{quicksort} example of
Figure~\ref{fig:running:superhomogeneous}
and Figure~\ref{fig:running:rptgraph},
$\allocation{\code{split}} = \{\code{R3}, \code{R4}\}$, and
$\allocation{\code{qsort}} = \{\code{R8}, \code{R9}, \code{R10}\}$,

\section{Region Liveness Analysis}
\label{seClra}

After the region points-to analysis,
we know the region variables of each procedure
and how the program variables are distributed
over the regions to which these region variables are bound.

In this section, we construct a region liveness analysis
that approximates the lifetimes of the region variables,
i.e.\ their liveness,
to decide when a region needs to be created
and when it can safely be reclaimed.
We make a distinction between local liveness and global liveness.
Local liveness concerns the lifetime of the region variable
inside the procedure itself, namely when we consider the procedure alone.
Global liveness concerns liveness with respect to the whole program,
namely when we take into account the call sites that call the procedure.
We show how we compute local liveness in Section~\ref{seClivenessClocal},
while Section~\ref{seClivenessCglobal} shows how we compute global liveness.

\subsection{Technical Background}

A region variable being live means that
(a) it should be bound to a region, and
(b) that region may possibly be used in future (forward) execution.
During its lifetime, the region bound to a region variable
may be allocated into by procedures other than the one that created the region,
so we often need to pass region variables as arguments of procedures.

Consider a procedure $p$.
We associate a \emph{program point} with every atomic goal in the body of $p$.
An \emph{execution path} in $p$ is a sequence of program points,
such that at runtime the atomic goals associated with these program points
are executed in sequence.
We denote an execution path by \ep,
in which the ${atom}_i$'s are the atomic goals involved, and
the indexes $i$'s are a dense sequence
giving the order among the atomic goals in this execution path.
The function \pp{{atom}} returns the program point associated with ${atom}$.
We use the notions
\emph{before} and \emph{after} a program point.
Before a program point means
right before the associated atomic goal is going to be executed;
while after a program point means its atomic goal has just been completed.
The set of live region variables at a program point is computed via the set
of live variables at the program point.
We also use two functions,
$\inargs{\mathit{atom}}$ and $\outargs{\mathit{atom}}$,
that respectively return
the sets of input and output arguments of $\mathit{atom}$.
For specialized unifications they are defined in Table~\ref{table:lva:unif}.
If $\mathit{atom}$ is a procedure's head, they return formal parameters,
whereas if $\mathit{atom}$ is a call they return actual parameters.
Those sets can be computed from the mode information of Mercury procedures.
\begin{table}[htb]
  \caption{Input and output arguments of unifications.}
  \begin{tabular}{l|l|*{2}{c|}}
    \hline
    \hline
    & & $\mathit{in\_args}$ & $\mathit{out\_args}$\\
    \hline
    construction   & $X$ \code{<=} $f(X_1,\dots, X_n)$ &  $\{X_1,\dots,X_n\}$  & $\{X\}$             \\
    deconstruction & $X$ \code{=>} $f(X_1,\dots, X_n)$ &  $\{X\}$               & $\{X_1,\dots,X_n\}$ \\
    test           & $X$ \code{==} $Y$                 &  $\{X,Y\}$             & $\emptyset$         \\
    assign         & $X$ \code{:=} $Y$                 &  $\{Y\}$           & $\{X\}$             \\
    \hline
    \hline
  \end{tabular}
  \label{table:lva:unif}
\end{table}

\subsection{Live Region Variables at a Program Point}
\label{seClivenessClocal}

In this subsection we specify the analysis
that computes the local liveness of region variables in a procedure.
We express local liveness by the sets of region variables
that are live before and after every program point in a procedure.
The liveness of a region variable at a program point is determined
by the liveness of the variables that are stored in the corresponding region.

\noindent\textbf{Live variables.}
A variable is live \emph{before} a program point
if it has been instantiated before the point
and may be used in the goal associated with the program point or after it.
A variable is live \emph{after} a program point
if it has been instantiated before or at the point
and may be used after the point.

\begin{algorithm}
    \small
    \caption{$\mathit{lva}(p)$: live variable analysis of a procedure $p$.}
    \label{algo:lva}
\begin{algorithmic}
    \REQUIRE{$p$ in superhomogeneous form.}
    \ENSURE{The sets of live variables
    before ($\LVbefore$) and after ($\LVafter$) all program points in $p$.}

    \FORALL{program points $i$ in $p$}
        \STATE $\LVb{i} = \LVa{i} = \emptyset$
    \ENDFOR
    \FORALL{$\mathit{ep} \equiv \ep$ in $p$}
        \FOR{$j = n$ downto 1}
            \STATE $i = \pp{{atom}_j}$
            \IF{$j = n$}
                \STATE $\LVa{i} = \outargs{p}$
            \ELSE
                \STATE $\LVa{i} = \LVa{i} \cup \LVb{\pp{{atom}_{j+1}}}$
            \ENDIF
            \IF{$j = 1$}
                \STATE $\LVb{i} = \inargs{p}$
            \ELSE
                \STATE $\LVb{i} = (\LVa{i} \setminus \outargs{{atom}_j})
                \cup \inargs{{atom}_j}$
            \ENDIF
        \ENDFOR
    \ENDFOR

\end{algorithmic}
\normalsize
\end{algorithm}

The live variable analysis for a procedure $p$
is defined in Algorithm~\ref{algo:lva}.
It traverses each execution path ($\mathit{ep}$) backwards,
starting with the last program point,
computing sets of live variables along the way.
At each program point, we update its $\LVafter$ and $\LVbefore$ sets.
The $\LVafter$ of the last program point(s) is defined to be \outargs{p},
while the $\LVbefore$ of the first program point(s) will be \inargs{p}.
This assumes that every procedure uses all its arguments,
but since we run this analysis \emph{after} a Mercury compiler pass
that removes unused arguments, this is a justified assumption.

\noindent\textbf{Live region variables.}
A region variable is live before (after) a program point
if its node is reachable from a variable that is live before (after) the
program point.

The set of nodes that are reachable from a variable $X$ is defined as follows:
\begin{center}
    \small
    $\Reach{X} = \{n_X\} \cup \{m \mid
    \exists (n_X, \mathit{label}_0, n_1), \ldots,
    (n_{i-1}, \mathit{label}_{i-1}, n_i) \in E \land m = n_i$\}.
    \normalsize
\end{center}

The live region variable analysis of a procedure
is specified in Algorithm~\ref{algo:lra}.
This algorithm computes the sets of live region variables
before ($\LRbefore$) and after ($\LRafter$) each program point
as the unions of the $\mathit{Reach}$ sets of all variables
in the $\LVbefore$ and in $\LVafter$ sets of the program point, respectively.

\begin{algorithm}
    \small
    \caption{$\mathit{lra}(p)$:
        live region variable analysis of a procedure $p$ }
    \label{algo:lra}
\begin{algorithmic}
    \REQUIRE{$\LVbefore$ and $\LVafter$ of all program points in $p$.}
    \ENSURE{The sets of live region variables
    before ($\LRbefore$) and after ($\LRafter$) all program points in $p$.}

    \FORALL{program points $i$ in $p$}
        \STATE $\LRb{i} = \LRa{i} = \emptyset$
        \FORALL{$X \in \LVb{i}$}
            \STATE $\LRb{i} = \LRb{i} \cup \Reach{X}$
        \ENDFOR

        \FORALL{$X \in \LVa{i}$}
            \STATE $\LRa{i} = \LRa{i} \cup \Reach{X}$
        \ENDFOR
    \ENDFOR
\end{algorithmic}
\normalsize
\end{algorithm}

\subsection{Lifetime of Regions across Procedure Boundary}
\label{seClivenessCglobal}

Sometimes we have to pass region variables between procedures.
For a procedure, the region variables reachable from its arguments
are all candidates to be region arguments.
But as we will see later,
not all of them may actually \emph{need} to be arguments.
This subsection introduces an analysis that,
by looking at the calling contexts of a procedure in the whole program,
decides which region variables become live or become dead inside the procedure.
With this global liveness information,
we can give regions shorter lifetimes, achieving better memory reuse.

Consider a procedure $q$ that is called by some procedure $p$.
We define:
\begin{itemize}
    \item \bornR{q} is the set of region variables of $q$
    that are mapped (by the $\alpha$ function at the call site) to
    region variables of $p$ that definitely become live inside $q$,
    i.e.\ in the code of $q$ or in one of the procedures $q$ calls.
    \item \deadR{q} is the set of region variables of $q$
    that are mapped to region variables of $p$
    that definitely cease to be live (i.e.\ they become dead) inside $q$.
    \item \outlivedR{q} is the set of region variables of $q$
    that are mapped to region variables of $p$ that outlive the call to $q$.
    They are live before the call and are still live after the call.
\end{itemize}
The idea is that, in the transformed program,
the region variables in \bornR{q}
will get bound to a region inside $q$
and $q$ will return the bound region variable to $p$, while
the region variables corresponding to \deadR{q} are passed by $p$ to $q$
and have their regions safely removed during the call to $q$.
The alternative would be that
$p$ creates the regions corresponding to \bornR{q} just before the call to $q$,
and removes the regions corresponding to \deadR{q} right after the call.
With that approach, many regions would have a longer lifetime,
which is why we prefer to create regions as late as possible
and remove them as soon as possible.

For a procedure $q$, we initially set
$\bornR{q} = \outputR{q} \setminus \inputR{q}$ and
$\deadR{q} = \inputR{q} \setminus \outputR{q}$,
where \inputR{q} and \outputR{q} are
the sets of region variables reachable from the variables in
\inargs{q} and \outargs{q}, respectively.
This is an overestimate in which
all the region variables that contain input terms
but are not involved with output terms
are assumed to become dead in $q$,
while all the region variables where output terms are stored
but are not yet bound at the entry of $q$
are assumed to become live in $q$.
We use \localR{q} to denote
the set of the region variables that are local to $q$
(not reachable from input or output variables);
it is computed by $N_q \setminus (\inputR{q} \cup \outputR{q})$.
Initially, $\outlivedR{q} = \inputR{q} \cap \outputR{q}$.
It is clear that \localR{q}, \bornR{q}, \deadR{q}, and \outlivedR{q}
form a partition of $N_q$.

The calling contexts of a procedure influence
what it can do to its non-local region variables.
Therefore when analyzing a procedure $p$, the analysis
applies the rules in Figure~\ref{fig:lra:rules}
to any ${atom}$ in $p$ that is a call to $q$.
These rules update the $\deadRegs$ and $\bornRegs$ sets of $q$
according to the calling context.
Rule L1 requires
a region variable to be moved from \deadR{q} to \outlivedR{q}
if its region needs to be live in $p$ after the call to $q$.
Rule L2 is there to avoid the problems that would arise
if we let a region that is referred to by more than one region variable in $q$
be removed when one of those region variables becomes dead.
Either that region can still be referred to through the other region variables,
in which case we would have removed it too early,
or the other region variables are also in \deadR{q},
in which case the region would be removed again.
Repeated application of L2 will ensure that our system
never removes aliased regions during the call to $q$
through \emph{any} of the region variables referring to them.
Rule L3 is analogous to L1;
it moves a region variable from \bornR{q} to \outlivedR{q}
if it is already live before the call to $q$.
Rule L4 is analogous to L2 in the same way;
just as we do not want to remove a region twice,
we do not want to create it twice.
Rules L2 and L4 together ensure that
region variables that are involved in a region alias
never belong to either $\bornRegs$ or $\deadRegs$ sets.

\begin{figure}[tb]
\centering
\scriptsize
\framebox{
$
\begin{array}{c}
\inferrule
{
    r \in \LRb{\pp{atom}} \\\\
    r \in \LRa{\pp{atom}} \\\\
    r = \alpha(r') \\
    r' \in \deadR{q}
}
{
    \deadR{q} = \deadR{q} \setminus \{r'\} \\\\
    \outlivedR{q} = \outlivedR{q} \cup \{r'\}
}
\quad(\textsc{L1})
\qquad
\qquad
\inferrule
{
    \alpha(r') = r \\
    \alpha(r'') = r \\\\
    r' \not= r'' \\
    r' \in \deadR{q}
}
{
    \deadR{q} = \deadR{q} \setminus \{r'\} \\\\
    \outlivedR{q} = \outlivedR{q} \cup \{r'\}
}
\quad(\textsc{L2})
\\\\
\inferrule
{
    r \in \LRb{\pp{atom}} \\\\
    r = \alpha(r') \\
    r' \in \bornR{q}
}
{
    \bornR{q} = \bornR{q} \setminus \{r'\} \\\\
    \outlivedR{q} = \outlivedR{q} \cup \{r'\}
}
\quad(\textsc{L3})
\qquad
\qquad
\inferrule
{
    \alpha(r') = r \\
    \alpha(r'') = r \\\\
    r' \not= r'' \\
    r' \in \bornR{q}
}
{
    \bornR{q} = \bornR{q} \setminus \{r'\} \\\\
    \outlivedR{q} = \outlivedR{q} \cup \{r'\}
}
\quad(\textsc{L4})
\end{array}
$
}
\\
\small
The atomic goal $atom$ is a call to $q(\dots)$ at a program point.
\caption{Region liveness analysis rules.}
\normalsize
\label{fig:lra:rules}
\end{figure}

When there is a change to any of the sets of $q$,
$q$ must be analyzed to propagate the change to the procedures it calls.
Therefore, this analysis requires a fixpoint computation.
After a fixpoint is reached,
each procedure has exactly one $\bornRegs$ set and one $\deadRegs$ set,
and these will be suited for its \emph{most restrictive} calling context.
For calls in a less restrictive context,
some regions will be created or removed outside the call,
which will mean that
some regions will be created earlier than needed and/or
some other regions will be removed later than needed.
For call sites that are sufficiently heavily used,
we could avoid the inefficiency inherent in that
by creating a specialized copy of the callee
that exactly matches the caller's context,
but this could be fairly expensive,
since it may (and generally will) require
specialized copies of many of the specialized callee's descendants as well.

In the \emph{quicksort} program from Figure~\ref{fig:running:normal},
\code{split} has three execution paths:
$\langle(1),(2),(3)\rangle$, $\langle(4),(5)$, $(6),(7)\rangle$, and
$\langle(4),(8),(9)\rangle$,
while \code{qsort} has two paths:
$\langle(1),(2)\rangle$ and $\langle(3),(4),(5),(6),(7)\rangle$.
\footnote{For convenience, we use program points to describe execution paths.}
Note that the third execution path of \code{split} does not contain
the test at (5) because of the semantics of if-then-else.
The $\mathit{LV}$ and $\mathit{LR}$ sets of \code{split} are
in Table~\ref{table:lra:result}\subref{table:lra:split},
while the sets of \code{qsort}
are in Table~\ref{table:lra:result}\subref{table:lra:qsort}
(see also Figure~\ref{fig:running:superhomogeneous}
and Figure~\ref{fig:running:rptgraph}).
\begin{table}
\centering
\small
\caption{Live variable and live region variable sets
in the \emph{quicksort} program.}
\subfloat[\code{split}\label{table:lra:split}]{
\begin{oldtabular}{c|c|c}
    \hline
    \hline
    pp          & $\mathit{LV}$                           & $\mathit{LR}$                       \\
    \hline
    ($1_b$)     & $\{\code{X},\code{L}\}$                 & $\{\code{R5},\code{R1},\code{R2}\}$ \\
    ($1_a,2_b$) & $\{\}$                                  & $\{\}$                              \\
    ($2_a,3_b$) & $\{\code{L1}\}$                         & $\{\code{R3},\code{R2}\}$           \\
    ($3_a$)     & $\{\code{L1},\code{L2}\}$               & $\{\code{R3},\code{R2},\code{R4}\}$ \\
    ($4_b$)     & $\{\code{X},\code{L}\}$                 & $\{\code{R5},\code{R1},\code{R2}\}$ \\
    ($4_a,5_b$) & $\{\code{X},\code{Le},\code{Ls}\}$      & $\{\code{R5},\code{R2},\code{R1}\}$ \\
    ($5_a,6_b$) & $\{\code{X},\code{Le},\code{Ls}\}$      & $\{\code{R5},\code{R2},\code{R1}\}$ \\
    ($6_a,7_b$) & $\{\code{L2},\code{Le}, \code{L11}\}$   & $\{\code{R4},\code{R2},\code{R3}\}$ \\
    ($7_a$)     & $\{\code{L1},\code{L2}\}$               & $\{\code{R3},\code{R2},\code{R4}\}$ \\
    ($4_a,8_b$) & $\{\code{X},\code{Le},\code{Ls}\}$      & $\{\code{R5},\code{R2},\code{R1}\}$ \\
    ($8_a,9_b$) & $\{\code{L1},\code{Le},\code{L21}\}$    & $\{\code{R3},\code{R2},\code{R4}\}$ \\
    ($9_a$)     & $\{\code{L1},\code{L2}\}$               & $\{\code{R3},\code{R2},\code{R4}\}$ \\
    \hline
    \hline
\end{oldtabular}
}
\vspace{0.5cm}
\subfloat[\code{qsort}\label{table:lra:qsort}]{
\centering
\begin{oldtabular}{c|c|c}
    \hline
    \hline
    pp          & $\mathit{LV}$                                 & $\mathit{LR}$                                     \\
    \hline
    ($1_b$)     & $\{\code{L},\code{A}\}$                       & $\{\code{R6},\code{R7},\code{R8}\}$               \\
    ($1_a,2_b$) & $\{\code{A}\}$                                & $\{\code{R8},\code{R7}\}$                         \\
    ($2_a$)     & $\{\code{S}\}$                                & $\{\code{R8},\code{R7}\}$                         \\
    ($3_b$)     & $\{\code{L},\code{A}\}$                       & $\{\code{R6},\code{R7},\code{R8}\}$               \\
    ($3_a,4_b$) & $\{\code{A},\code{Le},\code{Ls}\}$            & $\{\code{R8},\code{R7},\code{R6}\}$               \\
    ($4_a,5_b$) & $\{\code{A},\code{Le},\code{L1},\code{L2}\}$  & $\{\code{R8},\code{R7},\code{R9}, \code{R10}\}$   \\
    ($5_a,6_b$) & $\{\code{Le},\code{L1},\code{S2}\}$           & $\{\code{R9},\code{R7},\code{R8}\}$               \\
    ($6_a,7_b$) & $\{\code{L1},\code{A1}\}$                     & $\{\code{R9},\code{R7},\code{R8}\}$               \\
    ($7_a$)     & $\{\code{S}\}$                                & $\{\code{R8},\code{R7}\}$                         \\
    \hline
    \hline
\end{oldtabular}
}
\normalsize
\label{table:lra:result}
\end{table}
In this example, the sets after one program point are always equal
to the corresponding sets before the next point in the execution path.
However, this is not true in all cases.
Consider the last program point before a disjunction.
The set of live variables after this point contains
the region variables that are live in any of the disjuncts;
in general, some of these variables
will be live in only some of the disjuncts, not all.

When computing the $\deadRegs$ and $\bornRegs$ sets of these procedures,
the initial partition is changed only once,
when \code{R5} is removed from $\deadR{\code{split}}$
by an application of rule L1 to the call to \code{split} inside \code{qsort}.
The final result is as in Table~\ref{table:lra:across}.
\begin{table}
    \centering
    \small
    \caption{Partition of the set of region variables.}
    \begin{tabular}{l|*{4}{c|}}
        \hline
        \hline
                          & $\localRegs$                  & $\bornRegs$               & $\deadRegs$     & $\outlivedRegs$     \\
        \hline
        $\code{split}$    & $\emptyset$                   & $\{\code{R3},\code{R4}\}$ & $\{\code{R1}\}$ & $\{\code{R2,R5}\}$  \\
        $\code{qsort}$    & $\{\code{R9},\code{R10}\} $   & $\emptyset$               & $\{\code{R6}\}$ & $\{\code{R7,R8}\}$  \\
        \hline
        \hline
    \end{tabular}
    \label{table:lra:across}
\end{table}

\subsection{Correctness}

Algorithm~\ref{algo:lra}, the algorithm that detects
live region variables locally at each program point
is an extension of live variable analysis,
which is a standard, well-known program analysis \cite{POPA}.
Theorem~\ref{theorem:loc_and_sharing} guarantees that
the locations of variables and their possible sharing
are represented in the region points-to graphs.
Therefore Algorithm~\ref{algo:lra} computes all the live region variables
by starting from the live variables
and collecting all the reachable region variables
using the region points-to graphs.

The analysis in Section \ref{seClivenessCglobal}
aims to compute a shortest possible lifetime for a region.
Its termination follows from the facts that
each procedure uses a finite set of region variables
(which guarantees that
the initial $\bornRegs$ and $\deadRegs$ sets are finite),
and that the analysis only ever reduces the sizes of these sets.
The rules in Figure~\ref{fig:lra:rules} enforce all the cases
where a caller of a procedure needs to restrict
what the callee can do to its region variables.
The eager application of the rules therefore ensures that
after a fixpoint has been reached,
the $\bornRegs$ and $\deadRegs$ sets obtained for a procedure
will respectively contain exactly the region variables
that the procedure will safely create and remove.

\section{Program Transformation}
\label{seCpt}

The purpose of the program transformation is
to annotate all the procedures in the program
with the information the code generator needs about regions.
For each procedure, the tasks of the transformation are:
\begin{itemize}
    \item extend the procedure definition with the formal region arguments;
    \item extend its procedure calls
        with the corresponding actual region arguments;
    \item annotate each construction unification with the region variable
        representing the region into which the new memory cell should be put;
    \item insert instructions to \code{create} and \code{remove} regions
        at suitable points.
\end{itemize}
The third task is straightforward
because the new cell is always put into the region associated with
the variable on the left hand side of the construction unification,
and the map from variables to the region variables representing their regions
is available after the region points-to analysis.

We elaborate the other tasks in the next three subsections.

\subsection{Region Arguments}

The region variables in $\bornRegs$ and $\deadRegs$ must be arguments
because their regions will be created and removed inside the procedure.
Besides these region variables,
we also need to pass as arguments
the region variables that are reachable from the input and output variables
\emph{and} are allocated into in the procedure.
This set of arguments, which we call $\allocRegs$,
is therefore computed by
$\allocRegs = (\inputRegs \cup \outputRegs) \cap \allocationRegs$
(Section~\ref{seCallocCregions}).
Note that $\allocRegs$ is not necessarily disjoint
with any of $\bornRegs$, $\deadRegs$ and $\outlivedRegs$.

So all in all, the set of formal region arguments of a procedure
is $\deadRegs \cup \bornRegs \cup \allocRegs$.
In the \emph{quicksort} program,
$\allocR{\code{split}} = \{\code{R1},\code{R2},\code{R3},\code{R4}\}
\cap \{\code{R3}, \code{R4}\} = \{\code{R3}, \code{R4}\}$,
$\allocR{\code{qsort}} = \{\code{R6}, \code{R8}\} \cap \{\code{R8}\} =
\{\code{R8}\}$,
and the region arguments are
$\{\code{R1}\} \cup \{\code{R3}, \code{R4}\} \cup \{\code{R3},
\code{R4}\} = \{\code{R1}, \code{R3}, \code{R4}\}$ for \code{split}
and $\{\code{R6}\} \cup \emptyset \cup \{\code{R8}\} =
\{\code{R6}, \code{R8}\}$ for \code{qsort}.

The actual region arguments of a procedure call are computed simply
by looking up the formal region arguments of the called procedure
and applying the $\alpha$ function of the call site.

\subsection{Insertion of \code{create} and \code{remove} Instructions}

Regions are created and removed
only by the \code{create} and \code{remove} instructions respectively.
When a region is created,
the region variable in the \code{create} instruction is bound to it.
Removing a region consists of
calling \code{remove} on the region variable bound to the region.
We implement \code{create} and \code{remove}
as builtin Mercury procedures.
Calls to other procedures may also create and remove regions,
but only if those procedures directly or indirectly invoke
\code{create} or \code{remove}.
Unifications can never either create or remove regions.

\subsubsection{Transformation Rules}

The transformation rules in Figure~\ref{fig:transformationrules}
make use of the local and global liveness of region variables
to introduce \code{create} and \code{remove} instructions
when necessary.

\begin{figure}
    \centering
    \scriptsize
    \framebox{
    $ \begin{array}{c}
    \inferrule
    {
        {atom} \equiv q(\ldots) \\\\
        r \in \LRa{\pp{{atom}}} \setminus \LRb{\pp{{atom}}} \\\\
        r \in \localR{p} \cup \bornR{p} \cup \deadR{p} \\\\
        r = \alpha(r') \rightarrow r' \not\in \bornR{q}
    }
    {
        {add}\hspace{1mm}``\code{create(r)}"\hspace{1mm}
        {before}\hspace{1mm}{{atom}}
    }
    \quad(\textsc{T1})
    \qquad
    \inferrule
    {
        {atom} \equiv X< =f(\ldots) \\\\
        r \in \LRa{\pp{{atom}}} \setminus \LRb{\pp{{atom}}} \\\\
        r \in \localR{p} \cup \bornR{p} \cup \deadR{p}
    }
    {
        {add}\hspace{1mm}``\code{create(r)}"\hspace{1mm}
        {before}\hspace{1mm}{{atom}}
    }
    \quad(\textsc{T2}) \\\\
    \inferrule
    {
        {atom} \equiv q(\ldots) \\\\
        r \in \LRb{\pp{{atom}}} \setminus \LRa{\pp{{atom}}} \\\\
        r \in \localR{p} \cup \deadR{p} \cup \bornR{p} \\\\
        r = \alpha(r') \rightarrow r' \not\in \deadR{q}
    }
    {
        {add}\hspace{1mm}``\code{remove(r)}"\hspace{1mm}
        {after}\hspace{1mm}{{atom}}
    }
    \quad(\textsc{T3})
    \qquad
    \inferrule
    {
        {atom} \equiv \mathit{unif} \\\\
        r \in \LRb{\pp{{atom}}} \setminus \LRa{\pp{{atom}}} \\\\
        r \in \localR{p} \cup \deadR{p} \cup \bornR{p}
    }
    {
        {add}\hspace{1mm}``\code{remove(r)}"\hspace{1mm}
        {after}\hspace{1mm}{{atom}}
    }
    \quad(\textsc{T4}) \\\\
    \inferrule
    {
        {atom}'\hspace{1mm}{is}\hspace{1mm}{next}\hspace{1mm}{to}\hspace{1mm}
        {atom}\hspace{1mm}{in}\hspace{1mm}{an}\hspace{1mm}
        {execution}\hspace{1mm}{path} \\\\
        r \in \LRa{\pp{{atom}}} \setminus \LRb{\pp{{atom}'}} \\\\
        r \in \localR{p} \cup \deadR{p} \cup \bornR{p}
    }
    {
        {add}\hspace{1mm}``\code{remove(r)}"\hspace{1mm}
        {before}\hspace{1mm}{{atom'}}
    }
    \quad(\textsc{T5})
    \qquad
    \inferrule
    {
        r \in \VR{\pp{{atom}}} \setminus \LRa{\pp{{atom}}} \\\\
        r \in \localR{p} \cup \deadR{p} \cup \bornR{p}
    }
    {
        {add}\hspace{1mm}``\code{remove(r)}"\hspace{1mm}
        {after}\hspace{1mm}{{atom}}
    }
    \quad(\textsc{T6})
\end{array} $ }
\small
\caption{Transformation rules.}
\normalsize
\label{fig:transformationrules}
\end{figure}

\noindent\textbf{Creation rules T1 and T2.}
As we will show in Section~\ref{seCptCcorrectness}
(Proposition~\ref{propo:one}),
a region variable will never become locally live \emph{between} atomic goals;
a region cannot be not live after a program point
but live before the immediately next program point in some execution path.
A region variable can become locally live only \emph{within} atomic goals.
Let this be the atomic goal ${atom}$ at program point $i$ in procedure $p$.
T1's first condition says that this rule covers the case
where ${atom}$ is a call, for example to $q$.
The second condition is true for a region $r$
that is not live before ${atom}$ but is live after ${atom}$.
The third condition checks whether $p$ itself is allowed to create the region.
It is intuitively clear that $p$ needs to create regions
bound to region variables in \bornR{p} and \localR{p}.
The reason why we also allow $p$ to create regions in $\deadR{p}$
is that it is OK for $p$ to remove the region bound to $r$
at some point before ${atom}$, if that is safe,
and then recreate $r$ right before ${atom}$.
The new region will be removed later because $r$ is in $\deadR{p}$.
Such deletion-followed-by-recreation
is not allowed for regions in $\outlivedR{p}$
because the caller needs their contents.
The fourth condition checks whether the call will create the region;
if it will, then $p$ itself need not do so.
Overall, if the third condition is false,
then $p$'s caller will have created the region;
if the third condition is true,
but the fourth condition is false,
then $q$ will create the region;
if both the third and fourth conditions are true,
then the instruction that T1 inserts before the call will create the region.

Rule T2 covers the case where a region becomes live in a unification.
The first condition looks only for construction unifications
because for all other kinds of unifications,
the second condition always fails
(see Proposition~\ref{propo:two}, Section~\ref{seCptCcorrectness}).
T2 is analogous to T1, the main difference being that
unifications can never create regions.

\noindent\textbf{Removal Rules T3, T4, and T5.}
Removal rule T3 is analogous to creation rule T1.
If a region variable locally ceases to be live during a call,
the situation described by the first and second conditions,
what happens is governed by the third and fourth conditions.
If the third condition is false,
then $p$'s caller or one of its ancestors will (eventually) remove the region;
if the third condition is true,
but the fourth condition is false,
then $q$ will remove the region;
if both the third and fourth conditions are true,
then the instruction that T3 inserts after the call will remove the region.
Note that it is OK for $p$ to remove a region in $\bornR{p}$,
a region it must have previously created;
since the region will be live at the end of $p$,
$p$ will later create it again,
and that is all that $p$'s caller expects.

Removal rule T4 is likewise analogous to creation rule T2,
but a region can become dead in any kind of unification,
not just constructions.

While a region cannot be not live after one program point
and then magically become live before an immediately following program point,
it \emph{is} possible for a region
to be live after one program point (${atom}$ in T5)
and dead before an immediately following program point (${atom'}$).
This can happen e.g.\ when the following program point
is the first goal of a disjunct in a disjunction or switch,
and the region is live in \emph{other} disjuncts of the disjunction or switch.
In that case, the region is live after ${atom}$
because it is live in some execution paths that do not include ${atom'}$.
In such cases, rule T5 removes the region before ${atom'}$,
provided as usual that $p$ is allowed to do so.

\noindent\textbf{Handling instantly-dead variables: rule T6.}
In some cases,
a variable may be instantiated at some point but then never used after that.
We call them instantly-dead variables.
In logic programming in general and in Mercury in particular,
they can be void or singleton variables.
A void variable's name starts with the underscore
(see e.g.\ the first clause of \code{split} in Figure~\ref{fig:running:normal})
to explicitly tell the compiler that we do not care about its value.
A singleton variable
is a variable that occurs exactly once in a clause
whose name does \emph{not} start with an underscore.
Singleton variables often represent mistakes,
so the Mercury compiler issues a warning for them;
programmers who believe the code to be correct
can avoid the warning by adding a leading underscore,
turning the singleton into a void variable.

Because it is useless to do a construction unification
that binds the new term to an instantly-dead variable,
we assume that such unifications are eliminated
before our region analysis and transformation;
the Mercury compiler has an optimization that does this.
However, this is not a full solution.
A procedure can return several output arguments,
and it may be that the caller ignores some
and pays attention only to the others.
The ignored arguments pose a problem for our analysis.
Being instantiated means that we need regions to store their terms,
and of course we want those regions to eventually be removed.
However, the fact that the ignored arguments are not used in the future
makes them \emph{never} live
according to our concept of live variables (Section~\ref{seClra}).
Therefore we may not rely on the change of their liveness from live to dead
(the basis of rules T3-T5)
to remove the regions storing their terms.
That is why we have rule T6,
which tries to remove regions reachable from void variables
right after the point where the void variables get instantiated.
We assume that at each program point $i$,
we have available the set of such instantly-dead variables, \VV{i}
($i$ is the point at which they get instantiated).
We then compute \VR{i}, the set of region variables that
are reachable from the variables, by
\hbox{\raise-1mm\hbox{$\textstyle
\bigcup \atop \scriptstyle {V \in \VV{i}}$}} \Reach{V}.
The basic idea of T6 is to remove the region
of a region variable reachable from an instantly-dead variable
right after the point where the variable gets instantiated,
provided of course that the region variable
is not reachable from any of the live variables after the point.

\noindent\textbf{Example of re-creation and re-removal.}
We illustrate (a) creating, removing and recreating a region on the one hand
and (b) removing, creating, reremoving a region on the other hand
using the two procedures in Figure~\ref{fig:recreate:normal}
and their region-annotated counterparts in Figure~\ref{fig:recreate:annotated}.
\begin{figure}[tb]
\scriptsize
\vspace{2mm}
\begin{Verbatim}[frame=single,framerule=0.2pt,framesep=3pt]
         % p(in, out).          % q(in, out).
         p(A, B) :-             q(X, Y) :-                length(L) = N :-
         (1) C <= [1],          (1) Z := length(X),       (
             ( if                   ( if                      L == [],
         (2)     A == 1         (2)     Z == 1                N := 0
               then                   then                ;
         (3)     B := C         (3)     V := X                L => [_ | T],
               else                   else                    N := length(T) + 1
         (4)     B <= [2]       (4)     V <= [1]          ).
             ).                     ),
                                (5) Y := Z + length(V).
\end{Verbatim}
\small
\caption{Effect of re-creation of regions.}
\label{fig:recreate:normal}
\normalsize
\end{figure}
For completeness,
we include the definition of the function \code{length},
which returns the number of elements of the input list,
though its code is not important in this case.
\begin{figure}[tb]
\scriptsize
\vspace{2mm}
\begin{Verbatim}[frame=single,framerule=0.2pt,framesep=3pt]
             p(A, B@R1) :-                       q(X@R2, Y) :-
                 create(R1),                     (1) Z := length(X),
             (1) C <= [1] in R1,                     ( if
                 ( if                            (2)     Z == 1
             (2)     A == 1                            then
                   then                          (3)     V := X
             (3)     B := C                            else
                   else                                  remove(R2),
                     remove(R1),                         create(R2),
                     create(R1),                 (4)     V <= [1] in R2
             (4)     B <= [2] in R1                  ),
                 ).                              (5) Y := Z + length(V),
                                                     remove(R2).
\end{Verbatim}
\small
\caption{Effect of re-creation of regions: region-annotated version.}
\label{fig:recreate:annotated}
\normalsize
\end{figure}
We also assume that there is no region for integers.
Therefore the focus is only on
the variables \code{B} and \code{C} in the procedure \code{p} and
\code{V} and \code{X} in \code{q},
which are of the type \code{list\_int}
(see Example~\ref{example:typedeclaration}).
Each pair of them is assigned to the same region variables,
\code{R1} in \code{p} and \code{R2} in \code{q}
due to the assignments at the program points (3) in both procedures.
\code{p} and \code{q} are unrelated;
we use them to demonstrate different situations.

Assume that \code{p} can create \code{R1},
i.e.\ no calling context forces it otherwise.
So \code{R1} is in \bornR{\code{p}}.
In Figure~\ref{fig:recreate:annotated},
the \code{create} instructions added for it before (1) and (4)
are due to the rule T2.
The \code{remove} instruction added before (4) is due to rule T5.
If execution reaches the else branch,
the \code{R1} that was live after (1) is no longer live before (4),
and we can reclaim the memory occupied by \code{[1]}
by removing this incarnation of \code{R1},
before creating a new incarnation of it
and putting \code{[2]} into it.

For \code{q}, assume that \code{R2} is in \deadR{\code{q}}.
\code{R2} is not live before the program point (4),
and the \code{remove} instruction there is added by rule T5.
As \code{R2} is live after (4),
T2 adds the \code{create} instruction there as well.
The \code{remove} instruction after (5) is added by rule T4.
If execution reaches the else branch,
we reclaim the memory of the input list \code{X} by removing \code{R2}
before recreating it to construct \code{V}.

In both cases, we need to make sure that
the two operations before program point (4) are done in the right order.
This is ensured by the following algorithm.

\subsection{Insertion Algorithm}

The insertion of the instructions
is specified by Algorithm~\ref{algo:insert_region_instr},
which says how the transformation rules in Figure~\ref{fig:transformationrules}
should be applied to the atomic goal at each program point.

\begin{algorithm}
    \small
    \caption{Insertion of region instructions in a procedure $p$.}
    \label{algo:insert_region_instr}
\begin{algorithmic}
    \REQUIRE{$p$ in superhomogeneous form;
    all points-to graphs and region liveness sets are available.}

    \FORALL{program points $i$ in $p$}
    \STATE ${atom} = \atomat{i}$
    \STATE apply rule T6 to ${atom}$
        \IF{${atom} \equiv \mathit{unif}$}
            \STATE apply rule T4 to ${atom}$
            \IF{${atom} \equiv X < = f(\dots)$}
                \STATE apply rule T2 to ${atom}$
            \ENDIF
        \ELSE
            \STATE apply rules T1 and T3 to ${atom}$
        \ENDIF
    \ENDFOR

    \FORALL{$\mathit{ep} \equiv \ep$ in $p$}
        \FOR{$j = 1$ to $n - 1$}
            \STATE apply rule T5 to ${atom}_j$, with ${atom}' \equiv {atom}_{j+1}$
        \ENDFOR
    \ENDFOR

\end{algorithmic}
\normalsize
\end{algorithm}

Each program point is associated with three sets of region instructions:
a set of \code{remove} instructions added before it,
a set of \code{create} instructions added before it, and
a set of \code{remove} instructions added after it.
The instructions in the first set will be executed
before the instructions in the second set.
In Section~\ref{seCptCcorrectness},
we will prove the correctness of this choice
not just in our examples but also in the general case.

The first loop in Algorithm~\ref{algo:insert_region_instr}
applies all the transformation rules except T5
to the atomic goals at all the program points in a procedure.
We use the function \atomat{i} to refer to
the atomic goal at program point $i$.
While rule T6 can be applied to any atomic goal,
T4 needs to be tried only when the atom at a program point is a unification,
T2 only when the atom is a construction unification, and
T1 and T3 only when the atom is a procedure call.
The second loop follows every execution path to try rule T5,
which needs to consult information at two consecutive program points
at the same time.

The result of the program transformation
of the \emph{quicksort} program in Example~\ref{example:running:normal}
was shown in Figure~\ref{fig:running:annotated}.
The additions of the \code{remove} instructions
after the first program points in both \code{qsort} and \code{split}
result from the applications of T4.
The two \code{create} instructions in \code{split} were added by T2.

\subsection{Correctness of Region-Annotated Programs}
\label{seCptCcorrectness}

Region-annotating a program does not change its computational behavior;
it changes only the locations of terms in memory.
We therefore restrict our attention to the correctness of memory accesses,
i.e.\ the safety of read and write accesses to terms.
Before arguing about this safety,
we prove a theorem about the bindings of live region variables.

\begin{theorem}
\label{theorem:boundness}
    Consider a procedure $p$ in a program $P$.
    We call $P'$ the region-annotated program
    that is produced by applying the analyses and transformation in
    Sections~\ref{seCrpta}, \ref{seClra}, and \ref{seCpt} to $P$,
    in which $p'$ is the region-annotated version of $p$.
    If a region variable is live before (after) a program point $i$ in $p'$,
    then in $p'$ it is bound to a region before (after) $i$.
\end{theorem}
To prove Theorem~\ref{theorem:boundness}, we formulate several propositions.
\begin{proposition}
\label{propo:one}
    If program point $i$ is right before program point
    $j$ in some execution path of a procedure, then
    (i) $\LVb{j} \subseteq \LVa{i}$ and (ii) $\LRb{j} \subseteq \LRa{i}$.
\end{proposition}
\begin{proof}
    (i) follows directly from Algorithm~\ref{algo:lva}.
    (ii) follows from (i) and Algorithm~\ref{algo:lra}.
\end{proof}

\begin{proposition}
\label{propo:two}
    When the atomic goal at program point $i$ is a unification,
    we have the following two properties.
    If it is a construction unification, then $\LRb{i} \subseteq \LRa{i}$.
    If $\LRb{i} \subset \LRa{i}$ (strict subset),
    then the unification is a construction.
\end{proposition}
\begin{proof}
    Consider a construction unification
    of the form $X$ \code{<=} $f(X_1, \dots, X_n)$.
    By definition (Algorithm~\ref{algo:lva})
    $\LVb{i} = \LVa{i} \setminus \{X\} \cup \{X_1, \dots, X_n\}$.
    So we can compute $\LRb{i} =
    \bigcup_{{  V \in \LVa{i}, V \not= X}}    \Reach{V} ~ \cup~
    \bigcup_{j=1}^{n} \Reach{X_j}$.
     We can also write $\LRa{i} =
    \bigcup_{{  V \in \LVa{i}, V \not= X}} \Reach{V} ~ \cup ~
    \Reach{X}$.
    Algorithm~\ref{algo:intraproc} ensures that
    the edges from $n_X$ to $n_{X_j}$ are in the region points-to graph,
    therefore
    $\Reach{X} \supseteq \bigcup_{j=1}^{n}  \Reach{X_j}$.
    So $\LRb{i} \subseteq \LRa{i}$.

    To prove the second property we will show that
    if the unification is not a construction unification,
    then $\LRb{i} \supseteq \LRa{i}$.

    Consider an assignment unification of the form $X := Y$.
    From Algorithm~\ref{algo:intraproc} we have that
    $X$ and $Y$ are in the same node in the region points-to graph,
    therefore $\mathit{Reach}(X) = \mathit{Reach}(Y)$.
    By definition $\LVb{i} = (\LVa{i} \setminus \{X\}) \cup \{Y\}$, so
    $\LRb{i} = \bigcup_{{  V \in \LVa{i}, V \not= X}} \Reach{V}
   ~ \cup ~ \Reach{Y}$.
    We can write
    $\LRa{i} = \bigcup_{ V \in \LVa{i},V \not= X} \Reach{V}
    ~ \cup ~ \Reach{X}$
    and therefore $\LRb{i} = \LRa{i}$.

    Consider a test unification of the form $X == Y$.
    In this case, $\LVb{i} = \LVa{i} \cup \{X, Y\}$
    so obviously $\LRb{i} \supseteq \LRa{i}$.

    Consider a deconstruction unification of the form
    $X = > f(X_1, \dots, X_n)$.
    Here $\LVb{i} = (\LVa{i} \setminus \{X_1,\dots, X_n\}) \cup \{X\}$,
    and we have $\LRb{i} =
    \bigcup_{ V \in \LVa{i} \setminus \{X_1,\dots,X_n\}} \Reach{V}
    ~ \cup  ~ \Reach{X}$.
    We can write $\LRa{i} =
    \bigcup_{V \in \LVa{i} \setminus \{X_1,\dots,X_n\}} \Reach{V}
    ~ \cup ~ \bigcup_{j=1}^n \Reach{X_j}$.
    We have shown that
    $\mathit{Reach}(X) \supseteq (\bigcup_{j=1}^n \Reach{X_j}$).
    Therefore $\LRb{i} \supseteq \LRa{i}$.
\end{proof}

\begin{proposition}
\label{propo:three}
    If the atomic goal at program point $i$ is a unification
    and there exists a region variable $R$ such that
    $R \not\in \LRb{i}$ and $R \in \LRa{i}$,
    then $\LRb{i} \subset \LRa{i}$ (strict subset).
\end{proposition}
\begin{proof}
    The existence of a region variable $R$ such that
    $R \not\in \LRb{i}$ and $R \in \LRa{i}$ means that the unification
    cannot be an assignment, a test, or a deconstruction,
    because in each of those cases $\LRb{i} \supseteq \LRa{i}$
    (proof of Proposition~\ref{propo:two}).

    If the unification is a construction,
    then $\LRb{i} \subseteq \LRa{i}$ (Proposition~\ref{propo:two}).
    This implies that if there exists an $R$
    such that $R \not\in \LRb{i}$ and $R \in \LRa{i}$,
    then $\LRb{i} \subset \LRa{i}$.
\end{proof}

Now we can give the proof for Theorem~\ref{theorem:boundness}.

\begin{proof}[Proof of Theorem~\ref{theorem:boundness}]
\textbf{Hypothesis}:
Assume that Theorem~\ref{theorem:boundness} is true globally at all the
points that are reached before the (local) program point $i$ in $p$
in an execution of the program.

Consider a region variable $R$.

If $R$ belongs to \outlivedR{p},
then according to the Hypothesis
it is bound to a region at the call to $p$.
Since our transformation does not add \create{R} or \remove{R} to $p$
and none of the procedures called by $p$ creates or removes $R$,
it is bound to the same region at all points in $p$,
certainly including the points where it is live.

Consider the other case in which $R$ belongs
to one of $\localRegs$, $\bornRegs$, or $\deadRegs$.

\begin{description}
\item[Case 1.]
    Consider a region variable $R$ that is live before $i$,
    i.e.\ $R \in \LRb{i}$.

    \begin{itemize}
    \item
        When $i$ is the first program point,
        $R$ must be reachable from a variable in \inargs{p}
        (Algorithms~\ref{algo:lva} and~\ref{algo:lra}).
        In the context of a caller of $p$,
        the region variable of the caller that $R$ is mapped to
        is live before the call.
        By the Hypothesis we have that
        it is bound to a region before the call
        and therefore $R$ is bound to the region at the entry to $p$.
        The transformation rule T5,
        which adds a ${remove}$ instruction
        before a program point,
        is not applicable to the first program point,
        since it has no predecessor.
        Therefore no ${remove}$ instruction is added before $i$,
        meaning that $R$ is bound to a region before $i$.

    \item
        If $i$ is not the first program point,
        then $R$ is in \LRa{h}
        where $h$ is the program point right before $i$
        in an execution path (Proposition~\ref{propo:one}).
        According to our hypothesis,
        $R$ is bound to a region after $h$.
        Again, the rule T5 is not applicable
        because $R$ is in both \LRa{h} and \LRb{i},
        and therefore $R$ is bound before $i$.
    \end{itemize}

\item[Case 2.]
    Consider a region variable $R$ that is live after $i$,
    i.e.\ $R \in \LRa{i}$.
    Assume that ${atom}$ is the atomic goal at $i$.

    \begin{enumerate}
    \item
        Consider the case in which $R$ is not in \LRb{i}.

        \begin{itemize}
        \item
            If ${atom}$ is a unification,
            from Proposition~\ref{propo:three}
            we have that $\LRb{i} \subset \LRa{i}$
            and then from Proposition~\ref{propo:two}
            it must be a construction unification.

            Rule T1 adds a \create{R} instruction before ${atom}$,
            which means that $R$ is bound to a region before ${atom}$.
            Recall that we assume that
            the set of ${create}$ instructions
            are executed right before ${atom}$,
            after the execution of the set of ${remove}$ instructions, if any.
            Therefore $R$ is bound before ${atom}$.
            Since construction unifications never remove regions,
            and we never insert ${remove}$ instructions after them,
            $R$ must still be bound to the region after ${atom}$.

        \item
            Consider the case in which
            ${atom}$ is a procedure call to $q$.
            If $R$ is mapped to a region variable in \bornR{q},
            the region variable is live
            after any last program point of $q$.
            By the Hypothesis we can say that
            the region variable is bound to a region
            at the exit of $q$.
            So $R$ is bound to that region after the call.

            Otherwise, rule T1 will add a \create{R} before ${atom}$,
            which means that $R$ is bound to a region before ${atom}$
            (again no ${remove}$ instruction can be executed
            in between \create{R} and ${atom}$).

            Because $R$ is not live before the call,
            it is not reachable from any actual input arguments
            of the call to $q$.
            Therefore it is not mapped to a region variable of $q$
            that belongs to \deadR{q}.
            So we have that $R$ is not mapped to any region variables of $q$
            that are in any of \bornR{q} or \deadR{q},
            and \localR{q} contains only region variables local to $q$,
            $R$ must be mapped to a region variable in \outlivedR{q},
            which means that $R$ is not removed in $q$.
        \end{itemize}

        In both subcases above,
        the rules T3, T4 and T6 will not be applicable
        because $R$ is in \LRa{i}.
        Therefore no \remove{R} is added after ${atom}$.
        So we can conclude that $R$ is bound to a region after ${atom}$.

    \item
        Consider the case in which $R$ is in \LRb{i}.
        We showed in Case 1 that $R$ is bound to a region before $i$.

        If ${atom}$ is a unification it does not remove $R$.
        If ${atom}$ is a call to $q$,
        because $R$ is in both \LRa{i} and \LRb{i},
        $R$ cannot be mapped to a region variable
        in either \deadR{q} or in \bornR{q} (Rules L1 and L3).
        So ${atom}$ does not remove it.

        Again, no \remove{R} is added after ${atom}$ because $R$ is in \LRa{i}.

        Therefore we conclude that
        $R$ is bound to the same region after ${atom}$.
    \end{enumerate}
\end{description}
\end{proof}

\begin{theorem}
In region-annotated programs,
allocations of memory, and the associated memory write accesses, are safe.
\end{theorem}

\begin{proof}
An allocation of memory involves a construction unification.
From Theorem~\ref{theorem:loc_and_sharing},
we know the region variable corresponding to
the variable on the left hand side of the construction unification,
whose region is where the memory cell being constructed is stored.
We say that the construction unification is safe
if that region variable is bound to a region before this unification.

Consider the program point associated with the construction unification.
If the left hand side variable were instantly dead,
the unification would have been optimized away before region analysis,
so we know it is live after this unification.
This means that its region variable
must also be live after this point (Algorithm~\ref{algo:lra}).
By Theorem~\ref{theorem:boundness},
the region variable is bound to a region after the program point.
Since the construction unification does not create regions,
the region must have been created before the construction
and is available at the construction.
\end{proof}

\begin{theorem}
When a variable appears as an input argument
to an atomic goal at a program point,
we say that the variable is read at that point.
In region-annotated programs,
when a variable is read at a program point,
the term it is bound to is available.
\end{theorem}

\begin{proof}
When a variable is read at a program point,
the mode analysis pass of the Mercury compiler ensures
that it has been instantiated before that point.
From Theorem~\ref{theorem:loc_and_sharing},
we know the region variables in whose regions
the terms that the variable may be bound to are stored;
they are the region variables reachable from the variable.

Because the variable is read at that point,
we consider it a live variable before that point,
and therefore the region variables reachable from it
are also live before the point
(Algorithms~\ref{algo:lva} and \ref{algo:lra}).

Consider a variable $X$ that is read at a program point $i$ in a procedure $p$.
$X$ is bound in $p$
either because it is an input argument of $p$,
or because it is the output argument of some atomic goal in $p$.
Consider some execution path of $p$.
In the first case, $X$ is live before the first program point of the path.
Because it is bound by $p$'s caller, the Mercury mode system ensures that
$X$ cannot be an output of any atomic goal in $p$.
So according to Algorithm~\ref{algo:lva}, we have that
$X$ is live in the scope from before the first program point up to before $i$.
Similarly in the second case, we have that
$X$ is live in the scope from after its producing atomic goal up to before $i$.
This means that all the region variables reachable from $X$
are live during the same scope.
Therefore none of them get removed during the scope,
since rules T3, T4, T5, and T6 are not applicable,
and no procedure calls in the scope remove any of them due to rule L1.

So the term that $X$ is bound to is available at $i$
and the read at $i$ is safe.
\end{proof}

\section{Runtime Support for Regions During Forward Execution}
\label{seCsupportdet}

We now describe the runtime support needed
to execute region-annotated programs.
In this section,
we cover the support needed for forward execution,
while in the next section we will look at
the support needed for backward execution, i.e. backtracking.
The latter is much more extensive,
partly because our analyses in Section~\ref{seClra}
determine liveness only with respect to forward execution.

Let us look at the lifespan of a region during forward execution.
A region comes into existence
with the execution of a \code{create(R)} instruction
that assigns memory to the region and
binds the region variable \code{R} to a so-called {\em region handle},
which refers to the assigned memory.
From then on, terms are \emph{allocated} into the region
by construction unifications annotated with \code{R}.
When the memory referred to by the region handle bound to \code{R}
is no longer needed,
the program will end the lifetime of \code{R}
by executing \code{remove(R)}, which reclaims that memory.

This aspect of our implementation is generally similar to
the ``standard'' RBMM implementations for SML and Prolog,
which are described in detail in~\cite{Makholm00,Makholm00master}.
In our system, a region is a singly-linked list of fixed-size region pages.
Each region page has a \emph{data area},
an array of words that can be used to store program data,
and a pointer to the next region page to form the singly-linked list.
The \emph{handle} of the region,
which is how the rest of the system refers to it,
is the address of the \emph{region header}.
Besides some other fields that we will introduce later,
the header structure includes a \emph{region size record}:
a pointer to the newest region page,
and a pointer to the next available word in the newest region page.
Since region pages have a fixed size,
these two values implicitly also specify
the amount of free space in the newest region page.
As is usual in RBMM systems, we store each region header
at the start of the data area of its region's \emph{first} region page.
\footnote{
Storing the headers separately from the region pages
would require the system that now keeps track of which region pages are free
to also keep a separate free list for header-sized blocks.
This would cause fragmentation that would not occur
with the standard header-in-first-region-page design.
}
Figure~\ref{fig:regionstructure} shows a region with two region pages;
the shaded areas represent memory allocated to user data.

\begin{figure}[htp]
    \centering
    \includegraphics{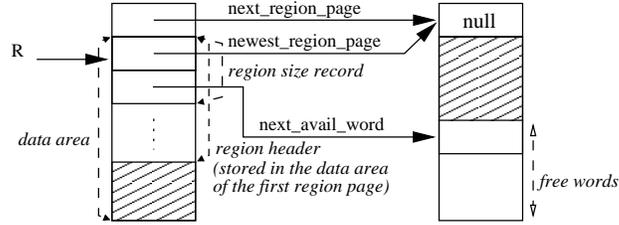}
    \small
    \caption{The data structure of a region \code{R}.}
    \normalsize
    \label{fig:regionstructure}
\end{figure}

There is no bound on the sizes of regions.
When a \emph{region is created} it will contain only one region page,
but it can be extended by adding more region pages when necessary.
The program maintains a global list of free region pages.
If the free list runs out,
the program requests a big chunk of memory from the operating system,
divides it into region pages, and adds them to the free list.
When a region needs to be extended,
we take a region page from the free list
and add it to the region as its new last region page,
and then update the region's size record.
When a \emph{region is reclaimed},
we return all its region pages to the free list.
An \emph{allocation into a region} always happens in its newest region page
simply by incrementing the pointer to the next available word.
When the amount of free memory in this region page
is not enough for the allocation,
we extend the region before allocating.

The advantage of this implementation is that
the basic region management actions are bounded in time;
even freeing all the region pages in a region can be done in constant time
(we can destructively append the region's list of pages
to the free list in constant time
because we maintain pointers to the tails as well as the heads of the lists).
Disadvantages are that
there is no natural size for the region pages \cite{Tofte04Retro},
and that if the remaining space of a region page
is not enough for an allocation,
that space will be wasted when a new region page is added.

Like most RBMM systems,
we do not do garbage collection inside regions.

\section{Runtime Support for Backtracking}
\label{seCsupportnondet}

Backtracking introduces two issues that need to be handled:
reclaiming the memory allocated by the computations backtracked over, and
ensuring that regions are reclaimed only when they are dead
with respect to both forward and backward execution.
The first issue obviously has to be handled at runtime.
For our initial implementation,
we have chosen to deal with the second issue, backward liveness,
in the runtime system too.
We expect this to give us the insights we will need later
to redesign the program analysis in Section~\ref{seClra}
to handle backward liveness both safely and precisely.
Moreover, our current system can serve as a reference for that work.

In Mercury, disjunctions are the main source of backtracking
because they provide alternatives.
However, backtracking is also possible in if-then-elses,
since they are just a special kind of disjunction:
$({if}~C~{then}~T~{else}~E)$
is semantically equivalent to
$(C, T ; {not}~{some}~[\cdots]~C, E)$.
Operationally, Mercury will try $C$.
If $C$ succeeds, Mercury executes $T$;
if $C$ fails, it executes $E$ as if $C$ had never been tried.
The handling of commit (Section~\ref{seCbgCinternal})
is related to the handling of backtracking
because committing to a solution may prune
some alternatives of relevant disjunctions.
Therefore, we need to provide runtime support for backtracking
in the context of these three language constructs.

The region-annotated program in Figure~\ref{fig:running2:annotated}
illustrates our two tasks.

\begin{figure}[htb]
\scriptsize
\vspace{2mm}
\begin{Verbatim}[frame=single,framerule=0.2pt,framesep=3pt]
   main(!IO) :-                            :- pred p(list_int, list_int, list_int, list_int).
       create(R1),                         :- mode p(in, in, out, out) is semidet.
   (1) X <= [1, 3, -1, 3] in R1,           p(X@R4, U@R5, V@R5, Y@R6) :-
       create(R2),                         (1) X => [H | T],
   (2) A <= [-2] in R2,                        ( if
       ( if                                        remove(R4),
           create(R3),                     (2)     H < 0
   (3)     p(X@R1, A@R2, B@R2, Y@R3)             then
         then                              (3)     Y  <= [H] in R6,
   (4)     io.write(B, !IO),               (4)     ( if    member(H, U)
           remove(R2),                     (5)       then  V := U
   (5)     io.write(Y, !IO),               (6)       else  V <= [H | U] in R5
           remove(R3)                              )
         else                                    else
   (6)     io.write(X, !IO),               (7)     p(T@R4, U@R5, V@R5, Y1@R6),
           remove(R1),                     (8)     ( if length(V) > length(Y1)
   (7)     io.write(A, !IO),               (9)       then fail
           remove(R2)                      (10)      else Y <= [H | Y1] in R6
       ).                                          )
                                               ).

   % mode(in, in), semidet                 % mode(in) = out, det.
   member(X, L) :-                         length(L) = N :-
       L => [H | T],                           (
       (                                           L == [], N := 0
           H == X                              ;
       ;                                           L => [_ | T],
           member(X, T)                            N := length(T) + 1
       ).                                      ).
\end{Verbatim}
\small
\caption{Illustrating the interaction of regions and backtracking.}
\label{fig:running2:annotated}
\normalsize
\end{figure}

We constructed this program,
which unfortunately has no \emph{intuitive} meaning,
to illustrate the interaction between regions and backtracking;
we will use it as our running example when describing the runtime support.
(We could find no equally useful real code of manageable size.
Also, we include the definitions of \code{member} and \code{length}
only for completeness;
their behavior is of no importance in this example.)
Regarding the lifetime of the regions,
\code{main} creates \code{R1} and \code{R2}
before the constructions of the lists \code{X} and \code{A}.
\code{main} creates \code{R3} before the call to \code{p} at (3),
and \code{p} will use this region to store the skeleton of \code{Y}.
All the \code{remove} instructions for regions are added
after the last \emph{forward} uses of the terms stored in them.
\code{member} and \code{length} only read their input variables,
so they need no region arguments.
For $p$, $\deadR{p} = \{R4\}$, $\bornR{p} = \emptyset$,
$\outlivedR{p} = \{R5, R6\}$, and $\allocation{p} = \{R5, R6\}$.

\begin{description}
\item[Task 1:]
    Preventing the reclamation of backward live regions.
    The condition of the if-then-else in \code{main}
    is the call to the semidet procedure \code{p}.
    The RBMM transformation marks the region \code{R1} for removal in the call
    because it is forward dead (it is not used in the then part)
    even though it is backward live (it \emph{is} used in the else part).
    We must make sure that \code{R1} is not actually removed
    while it is backward live.
    In this case, that means we need to delay
    the reclamation of \code{R1} until we reach the then part,
    since it is not safe to destroy \code{R1} if the condition fails.
    We therefore distinguish \emph{reclaiming} a region,
    which makes the memory of the region available for other uses
    and thus potentially destroys its contents,
    from the operation of \emph{removing} a region,
    which causes the region to be reclaimed only when it is safe to do so.
    Basically, a region is removed when it is forward dead,
    and it is reclaimed when it is both forward and backward dead.

\item[Task 2:]
    Reclaiming the memory used by backtracked-over computations.
    The call to \code{p} has two output arguments, \code{B} and \code{Y}.
    \code{main} tells \code{p} to put any cells for \code{B} in \code{R2},
    and creates \code{R3} so that \code{p} can put \code{Y} into it.
    If the condition succeeds, we must leave both regions alone.
    If the condition fails,
    we should restore \code{R2} to its size before the condition,
    and we should reclaim \code{R3} in its entirety.
\end{description}

We now define several runtime concepts
that we will use in the rest of the paper.

\noindent\textbf{Old vs new regions.}
A region is \emph{old} with respect to a point
during the execution of a program if it was created before that point,
otherwise it is \emph{new} with respect to that point.
We also refer to old regions as the \emph{existing regions}.
To allow efficient checks whether a region is old or new,
we maintain a global \emph{region sequence number} counter (starting at one)
and include a \code{sequence\_number} field in region headers.
When we create a region,
we timestamp it by setting its \code{sequence\_number} from the global counter,
and increment the counter.
When execution reaches a point in the program that sets up later backtracking,
such as the entry point of a disjunction,
we save the current sequence number.
Then all the regions which are created before that point,
i.e.\ the old regions with respect to the point,
will have their sequence numbers smaller than the saved value;
the regions which are created after that point,
i.e.\ the new regions with respect to the point,
will have their sequence numbers greater than or equal to the saved value.
When the program backtracks to that point, we can use the saved value
to check whether a region has been created before or after the point.
In the context of RBMM,
the memory that we want to reclaim at a resumption point
will be new allocations into existing regions,
and new regions in their entirety
(since they have been created
by the computation we have just backtracked over).

\noindent\textbf{Region list.}
To do instant reclaiming of new regions,
knowing the sequence numbers of the new regions is not enough;
we also need to \emph{reach} them.
We therefore link all the live regions into a doubly-linked \emph{region list}
(using two additional pointers in the region header).
We maintain a global pointer to the head of the list,
which will be the newest live region.
When a region is created, we add it to the head of the region list;
when a region is reclaimed, we remove it from the list.
We maintain the invariant that
the region list is ordered by regions' creation time, newest first.
To reclaim new regions,
we can traverse the region list from its head
and reclaim each region until we meet an old one.

\noindent\textbf{Region size snapshots.}
To do instant reclaiming of new allocations into an existing region,
we need the old size of the region.
When we need to remember the size of a region at a point,
we can save its region size record at that point.

\noindent\textbf{Protection.}
We will prevent the destruction of backward live regions
by \emph{protecting} them so that
when a removal happens to the region during forward execution,
the removal will be ignored.

\noindent\textbf{Changes to live regions by a goal.}
When providing support for backtracking,
sometimes we want to know about the changes which may be caused
by a goal to the set of regions the goal may refer to.
This means we need to know about
any new regions the goal creates, any live regions the goal removes,
and any live regions in which the goal performs allocations.
We refer to these sets of regions as the goal's
\emph{created}, \emph{removed}, and \emph{allocated} sets, respectively.
We have computed several sets of region variables for procedures,
such as $\inputRegs$, $\bornRegs$, $\deadRegs$, and $\allocationRegs$.
The \emph{created}, \emph{removed}, and \emph{allocated} sets of goals
can be computed from these in a fairly straightforward manner,
as shown by the following paragraphs.

\noindent\textbf{Changes to live regions by a goal: creation.}
Only \code{create} instructions and procedure calls may create regions.
A \code{create} instruction always creates the region in its argument.
A procedure call will create the regions
that are the actual region arguments corresponding to the formal arguments
in the ${bornR}$ set of the called procedure.
For a compound goal,
its created set is the set of all regions created inside it,
either directly or through a procedure call,
even if the region is also removed later,
because at compile time we cannot know
whether a removed region is actually reclaimed.

\noindent\textbf{Changes to live regions by a goal: removal.}
We can similarly use
\code{remove} instructions and the $\deadRegs$ sets of procedures
to compute the removed set of each goal.
Some of these regions may be removed, created and removed again.
Since we only care about the old regions which are removed inside a goal,
we exclude regions created inside the goal (i.e.\ the goal’s created set)
from its removed set.

\noindent\textbf{Changes to live regions by a goal: allocation.}
A region is allocated into by construction unifications and by procedure calls.
A construction unification will allocate
into the region with which it is annotated.
A procedure call possibly allocates into the regions of region variables
that are mapped to by those in the procedure's $\allocationRegs$ set.
Because we are only interested in allocations in old regions
(allocations into new regions being reclaimed by reclaiming the whole region),
we restrict the allocated set
to the regions in $\inputRegs \cap \allocationRegs$.

\noindent\textbf{Changes to live regions by a goal: an example.}
Take the condition of the if-then-else in the procedure \code{p}
in Figure~\ref{fig:running2:annotated} as an example goal.
We say that the region \code{R4} is removed in the condition
because \code{R4} is live before the condition
and \code{remove(R4)} has been added to the condition.
Or take the condition of the if-then-else in \code{main}.
We say region \code{R3} is created in the condition
because \code{create(R3)} has been inserted into the condition,
while region \code{R1} is removed in the condition
because it is live before the condition
and is removed in the call to \code{p}.
We have $\allocation{p} = \{R5, R6\}$,
but while \code{R5} is an input argument of \code{p}, \code{R6} is not,
so the only \emph{old} region \code{p} allocates into is \code{R5}.
So the allocation set of the condition in \code{main} is \code{R2},
since ${\code{R2}} = \alpha{({\code{R5})}}$.

We provide the runtime support for backtracking for a program
by generating extra supporting code at the right places to achieve our goals.
In the next three subsections we will describe in detail
the support for disjunctions, if-then-elses, and commits.

\subsection{Support for Disjunctions}
\label{seCsupportnondetCdisjunction}

The Mercury compiler supports only one search strategy:
depth-first search with chronological backtracking,
so that the disjuncts of each disjunction are tried in order.
Given a disjunction \code{(g1; \ldots; gi; \ldots; gn)},
we refer to \code{g1} as the first disjunct,
to the \code{gi}s for all $1 < i < n$ as middle disjuncts,
and to \code{gn} as the last disjunct of the disjunction.
We will also use ``later disjunct'' to refer to any \code{gi} for $i > 1$.

A disjunction can have any determinism.
The most general determinism is of course nondet,
but if one of the disjuncts always has at least one solution,
then the disjunction as a whole does too,
so a disjunction can also be multi.
And if the disjunction has no outputs
(which happens frequently for disjunctions in the conditions of if-then-elses),
then the disjunction as a whole cannot have more than one solution,
which means that it will be either det or semidet,
depending on whether it has an always-succeeding disjunct.
(Typical programs do not contain det disjunctions,
since they are equivalent to \code{true}.)

For our purposes, the important distinction is between
nondet and multi disjunctions on the one hand,
in which backtracking may reach a later disjunct
from code executed outside the disjunction,
\emph{after} the success of a previous disjunct,
and semidet and det disjunctions on the other hand,
in which backtracking to a later disjunct
is possible only from code \emph{within} an earlier disjunct.\footnote{
Semidet code in Mercury never does deep backtracking;
it only ever does local, shallow backtracking.
Semidet procedures return a success/failure indication,
which is then tested by the caller.
An arm of a semidet disjunction \emph{can} call nondet code,
but only if that nondet code is wrapped in a commit (see later);
the commit will convert any deep backtracks done by the code inside it
to shallow backtracking for the code outside it.
}
Since we do not care about the minimum number of solutions of each disjunction,
our support treats multi disjunctions the same as nondet ones
and det disjunctions the same as semidet ones.
In the following,
we will therefore talk only about nondet and semidet disjunctions.
We consider nondet disjunctions first, since they are more general.

Figure~\ref{fig:supportdisj} shows in pseudo-code form
the supporting code we add to a nondet disjunction.
We insert code at the following points:
(d1) which is the start of the first disjunct,
(d2) which represents the start of every middle disjunct,
and (d3) which is the start of the last disjunct.
These code fragments communicate
using shared data in what we call a \emph{disj frame}.
Each entry to a disjunction creates a new disj frame.
Since multiple nested disjunctions can be active at the same time,
we link these frames together to form the \emph{disj stack}
(this is possible due to chronological backtracking).
The disj stack is not a separate stack;
we reserve space for its frames
in the usual stacks used by the Mercury language implementation.
We maintain a global pointer to the top disj frame on the disj stack.

\begin{figure}[tb]
\scriptsize
\begin{Verbatim}[frame=single,framerule=0.2pt,framesep=3pt]
...,
( (d1): start of the disjunction and also of the first disjunct
        (a) push a disj frame
        (b) save the global region sequence number
        (c) save region size records and their number
    g1
; ...
; (d2): start of a middle disjunct
        (a) do instant reclaiming of new regions
        (b) do instant reclaiming of allocations in old regions
    gi
; ...
; (d3): start of the last disjunct
        (a) do instant reclaiming of new regions
        (b) do instant reclaiming of allocations in old regions
        (c) pop the disj frame
    gn
), ...
\end{Verbatim}
\small
\caption{RBMM runtime support for nondet disjunctions.}
\label{fig:supportdisj}
\normalsize
\end{figure}

A disj frame has a fixed part and a nonfixed part.
In Figure~\ref{fig:disjframe}, the fixed part is the 4-slot box
separated by a thick line from the nonfixed part.
The four slots in the fixed part are:
\begin{itemize}
\item
The \code{prev\_disj\_frame} slot holds the pointer to the previous disj frame,
or null if there is none.
\item
The \code{saved\_seq\_num} slot holds
the value of the global region sequence number
at the time when the disjunction was entered.
\item
The \code{num\_prot\_region} field
gives the number of regions which are protected by a semidet disjunction
(which we will discuss later).
For a nondet disjunction, this slot will contain zero.
\item
The \code{num\_size\_rec} field gives
the number of region size records saved in the nonfixed part.
\end{itemize}

\begin{figure}[htp]
\begin{center}
\expandafter\ifx\csname graph\endcsname\relax
   \csname newbox\expandafter\endcsname\csname graph\endcsname
\fi
\ifx\graphtemp\undefined
  \csname newdimen\endcsname\graphtemp
\fi
\expandafter\setbox\csname graph\endcsname
 =\vtop{\vskip 0pt\hbox{%
    \special{pn 8}%
    \special{pa 0 150}%
    \special{pa 4500 150}%
    \special{pa 4500 0}%
    \special{pa 0 0}%
    \special{pa 0 150}%
    \special{fp}%
    \special{pa 0 300}%
    \special{pa 4500 300}%
    \special{pa 4500 150}%
    \special{pa 0 150}%
    \special{pa 0 300}%
    \special{fp}%
    \special{pa 0 450}%
    \special{pa 4500 450}%
    \special{pa 4500 300}%
    \special{pa 0 300}%
    \special{pa 0 450}%
    \special{fp}%
    \special{pa 0 600}%
    \special{pa 4500 600}%
    \special{pa 4500 450}%
    \special{pa 0 450}%
    \special{pa 0 600}%
    \special{fp}%
    \special{pa 0 750}%
    \special{pa 4500 750}%
    \special{pa 4500 600}%
    \special{pa 0 600}%
    \special{pa 0 750}%
    \special{fp}%
    \special{pa 0 900}%
    \special{pa 4500 900}%
    \special{pa 4500 750}%
    \special{pa 0 750}%
    \special{pa 0 900}%
    \special{fp}%
    \special{pa 0 1050}%
    \special{pa 4500 1050}%
    \special{pa 4500 900}%
    \special{pa 0 900}%
    \special{pa 0 1050}%
    \special{fp}%
    \special{pa 0 1200}%
    \special{pa 4500 1200}%
    \special{pa 4500 1050}%
    \special{pa 0 1050}%
    \special{pa 0 1200}%
    \special{fp}%
    \special{pa 0 1350}%
    \special{pa 4500 1350}%
    \special{pa 4500 1200}%
    \special{pa 0 1200}%
    \special{pa 0 1350}%
    \special{fp}%
    \special{pa 0 580}%
    \special{pa 4500 580}%
    \special{fp}%
    \special{pa 0 590}%
    \special{pa 4500 590}%
    \special{fp}%
    \graphtemp=.5ex
    \advance\graphtemp by 0.075in
    \rlap{\kern 0.050in\lower\graphtemp\hbox to 0pt{prev\_disj\_frame\hss}}%
    \graphtemp=.5ex
    \advance\graphtemp by 0.225in
    \rlap{\kern 0.050in\lower\graphtemp\hbox to 0pt{saved\_seq\_num\hss}}%
    \graphtemp=.5ex
    \advance\graphtemp by 0.375in
    \rlap{\kern 0.050in\lower\graphtemp\hbox to 0pt{num\_prot\_region\hss}}%
    \graphtemp=.5ex
    \advance\graphtemp by 0.525in
    \rlap{\kern 0.050in\lower\graphtemp\hbox to 0pt{num\_size\_rec\hss}}%
    \graphtemp=.5ex
    \advance\graphtemp by 0.675in
    \rlap{\kern 0.050in\lower\graphtemp\hbox to 0pt{prot\_region\_id\hss}}%
    \graphtemp=.5ex
    \advance\graphtemp by 0.825in
    \rlap{\kern 0.050in\lower\graphtemp\hbox to 0pt{...\hss}}%
    \graphtemp=.5ex
    \advance\graphtemp by 0.975in
    \rlap{\kern 0.050in\lower\graphtemp\hbox to 0pt{snapshot\_region\_id\hss}}%
    \graphtemp=.5ex
    \advance\graphtemp by 1.125in
    \rlap{\kern 0.050in\lower\graphtemp\hbox to 0pt{snapshot\_size\_record\hss}}%
    \graphtemp=.5ex
    \advance\graphtemp by 1.275in
    \rlap{\kern 0.050in\lower\graphtemp\hbox to 0pt{...\hss}}%
    \graphtemp=.5ex
    \advance\graphtemp by 0.075in
    \rlap{\kern 1.500in\lower\graphtemp\hbox to 0pt{\emph{(previous disj frame)}\hss}}%
    \graphtemp=.5ex
    \advance\graphtemp by 0.225in
    \rlap{\kern 1.500in\lower\graphtemp\hbox to 0pt{\emph{(saved sequence number)}\hss}}%
    \graphtemp=.5ex
    \advance\graphtemp by 0.375in
    \rlap{\kern 1.500in\lower\graphtemp\hbox to 0pt{\emph{(number of protected regions)}\hss}}%
    \graphtemp=.5ex
    \advance\graphtemp by 0.525in
    \rlap{\kern 1.500in\lower\graphtemp\hbox to 0pt{\emph{(number of saved region snapshots)}\hss}}%
    \graphtemp=.5ex
    \advance\graphtemp by 0.675in
    \rlap{\kern 1.500in\lower\graphtemp\hbox to 0pt{\emph{(handle of a protected region)}\hss}}%
    \graphtemp=.5ex
    \advance\graphtemp by 0.975in
    \rlap{\kern 1.500in\lower\graphtemp\hbox to 0pt{\emph{(handle of a region in a snapshot)}\hss}}%
    \graphtemp=.5ex
    \advance\graphtemp by 1.125in
    \rlap{\kern 1.500in\lower\graphtemp\hbox to 0pt{\emph{(snapshot size record of that region)}\hss}}%
    \hbox{\vrule depth1.350in width0pt height 0pt}%
    \kern 4.500in
  }%
}%
\centerline{\box\graph}
\end{center}
\caption{The structure of a disj frame.}
\label{fig:disjframe}
\end{figure}

\noindent\textbf{Disj-protecting backward live regions.}
Consider a region which was created before the execution of a disjunction.
Assume that this region is removed during forward execution,
either by the code of a disjunct,
or after the success of that disjunct
by code following and outside the disjunction,
but that this region is backward live with respect to
a later disjunct of the disjunction.
In this case, we need to make sure that
if the region is removed during forward execution,
it will not be actually reclaimed.
Of course, the instruction that removes the region may not be reached
because forward execution may fail before it gets there.
But in general, we have to assume that
the \code{remove} instruction \emph{will} be executed,
and that if the region may be needed after backtracking,
we will need to prevent it from being reclaimed
during the forward execution.
We achieve this by \emph{disj-protecting} such regions as follows.
At the start of the disjunction, at (d1),
we push a disj frame on the disj stack
and save the current global sequence number into
the \code{saved\_seq\_num} slot of the disj frame.
A region is disj-protected by a disj frame if its sequence number
is smaller than the sequence number saved in that disj frame.
The \code{remove} instruction will only reclaim a region
if the region is not disj-protected.
Due to chronological backtracking,
the order of the frames on the disj stack
always corresponds to the order of the creation of those frames.
Together with the fact that
the global region sequence number is monotonically increasing,
this implies that if a region is protected by a disj frame,
it is also protected by all the later frames on the disj stack.
This invariant means that to check if a region is disj-protected or not,
we only need to check if it is protected by the top disj frame.

The program will no longer backtrack into a disjunction
after starting the execution of its last disjunct.
This means that no regions need to be protected any more by this disjunction.
Therefore, at the start of the last disjunct, at (d3),
we disj-unprotect them by popping the disj frame.
The regions which had previously been protected only by this disj frame
will be reclaimed when execution reaches their \code{remove} instructions.

\noindent\textbf{Instant reclaiming of new regions.}
When the program backtracks to a later disjunct,
we want to reclaim all the regions that have been created
during the computation that has just been backtracked over,
i.e.\ all the regions that were created after entry to the disjunction.
At (d1), we saved the global sequence number in the disj frame.
Therefore at the start of a later disjunct of the disjunction,
at (d2) and (d3),
we just need to traverse the region list,
and reclaim all the regions we see until we encounter a region
whose sequence number indicates that it was created before the disj frame.

\noindent\textbf{Instant reclaiming of new allocations in old regions.}
When arriving at a later disjunct,
we want to restore all the regions that existed before the disjunction
to the sizes they had when entering the disjunction,
recovering any memory that has been allocated in them.
For each old region,
we need to save the region's size record
in the nonfixed part of the disjunction's disj frame at (d1),
so that we can restore the region's size at (d2) and (d3).
We need three slots for each region:
one for the region handle so that we know
to which region the saved record belongs,
and the other two for the record itself (see Figure~\ref{fig:disjframe}).
To be able to loop through the saved records
and restore the regions at (d2) and (d3),
we store the number of saved records in the fixed \code{num\_size\_rec} slot.
The first saved record can be located by taking the address of the frame,
and adding both the size of the fixed part
and the number of slots for protected regions
(which is zero for nondet disjunctions).

The set of regions that existed before the disjunction
and that may be allocated into by code following the disjunction
is not available to the compiler.
In theory, we could implement a global analysis to make it available,
but such an analysis would be very complicated,
especially for multi-module programs.
Even if such an analysis existed,
we would still have a big problem,
which is that the number of regions in this set is not bounded,
and in many cases the set would contain
tens, hundreds or even thousands of regions.
Saving and then restoring the sizes of that many regions
can take a significant amount of both memory and time.
We do not want this overhead to outweigh the benefits of instant reclaiming.

In our implementation,
we have chosen to save and restore the sizes of only the regions
that are \emph{locally} forward live at the start of the disjunction;
this means the regions that are forward live before the disjunction
\emph{and} whose region variables are visible at that point.
(This information is readily available inside the Mercury compiler.)
This means that we do not recover memory in regions
that are forward live before the disjunction
but whose identity was not passed to the current procedure,
and are visible only in its ancestors.
Since nondet disjunctions are quite rare in most Mercury programs
(most programs that do serious searching tend to program their own searches
instead of relying on chronological backtracking),
we do not expect this to be too much of a problem.
We will see below that
we do not miss memory recovery opportunities for semidet disjunctions.

We save and restore the sizes of \emph{all} regions
that are locally forward live at the start of the disjunction
(the number of these regions governs
how much space we reserve for the nonfixed part of the disj frame).
We save and restore the sizes even of regions
that are never allocated into before backtracking,
since (in the absence of the analysis mentioned above)
we do not know which ones of those are.
This may lead to some unnecessary saving and restoring,
but in typical programs,
the number of regions whose size we save and restore at a disjunction
is usually relatively small,
and in that case the memory or runtime cost
of these unnecessary saves and restores is negligible.
In some cases, however, the cost can be significant,
and an optimization that eliminates saves/restores
with a poor cost/benefit ratio would be useful.
Such an optimization would probably need access
to profiling information about region reclamation.
We do not yet generate such information.

\noindent\textbf{Specialized treatment of semidet disjunctions.}
Because at most one disjunct of a semidet disjunction may succeed,
when one of its disjuncts is reached,
it means that all the previous disjuncts have failed and
that therefore (more importantly for us)
execution has not passed outside the disjunction's scope.
Therefore, we only need to provide runtime support for a semidet disjunction
if in its scope there is some change
with respect to the set of existing regions.
This basically means that
the runtime support for nondet disjunctions described above
will only be applied to semidet disjunctions
whose created, removed and allocated sets are not all empty.
In our practical experience with Mercury,
most semidet disjunctions contain only tests,
and rarely make changes to the heap.
Therefore the support we describe below
is needed only by a relatively small fraction of semidet disjunctions.

For a semidet disjunction, the Mercury compiler generates code such that
when one of its non-last disjunct succeeds,
the execution will commit to it and not go back to try any later disjuncts.
This means the code we add at (d3) may not be reached
after the success of a non-last disjunct, causing two problems.
First, the disj frame will not be popped.
Second, the regions which are removed by this disjunction
but are protected against reclamation while later disjuncts exist
will not be
first unprotected at the start of the execution of the last disjunct
and then reclaimed in the body of the last disjunct,
as in the case of nondet disjunctions.
Our solution is to do these two tasks at the end of any non-last disjuncts,
i.e.\ after their success at (e1) and (e2)
as in Figure~\ref{fig:supportsemidisj}.
\begin{figure}[htb]
\scriptsize
\begin{Verbatim}[frame=single,framerule=0.2pt,framesep=3pt]
...,
( (d1): start of the disjunction and of the first disjunct
        (a) push a disj frame
        (b) save the global region sequence number
        (c) save region size records and their number
        (d) save protected regions and their number
    g1
  (e1): end of the first disjunct
        (a) reclaim protected regions
        (b) pop the disj frame
; ...
; (d2): start of a middle disjunct
        (a) do instant reclaiming of new regions
        (b) do instant reclaiming of allocations in old regions
    gi
  (e2): end of a middle disjunct
        (a) reclaim protected regions
        (b) pop the disj frame
; ...
; (d3): start of the last disjunct
        (a) do instant reclaiming of new regions
        (b) do instant reclaiming of allocations in old regions
        (c) pop the disj frame
    gn
), ...
\end{Verbatim}
\small
\caption{RBMM runtime support for semidet disjunction.}
\label{fig:supportsemidisj}
\normalsize
\end{figure}

To solve the first problem, we pop the frame at (e1.b) and (e2.b).
To solve the second problem,
at (d1) we loop through the regions in the disjunction's removed set.
If a region is already protected,
we do not want it to be reclaimed in the disjunction
and its \code{remove} instructions inside the disjunction
will be ineffective anyway,
so we do not need to do anything.
If a region is not already protected,
we save its handle in the nonfixed part of the disj frame.
At the end, we store the number of region handles we saved
in the frame's \code{num\_prot\_region} slot.
The code at (e1.a) and (e2.a) will loop through the saved handles,
and reclaim all the saved regions
(they were logically removed during the disjunct,
but the protection of this disjunction
prevented their \code{remove} instructions from actually reclaiming them.)

At (d1.c), we save the sizes
of only the regions in the disjunction's allocated set.
Since execution cannot leave a semidet disjunction,
we do not miss any memory recovery opportunities
by restricting ourselves to these regions.

\subsubsection{Disjunctions: Summary}
To summarize Section \ref{seCsupportnondetCdisjunction},
we review how we handle Tasks 1 and 2 for disjunctions;
first nondet disjunctions, and then semidet disjunctions.

We prevent the reclamation of backward live regions (Task 1)
by disj-protecting all regions whose sequence number indicates
they were created before the disjunction was entered.
The protection of such regions
starts at the beginning of the first disjunct (d1.a and d1.b),
and ends at the beginning of the last disjunct (d3.c).
Such regions are no longer protected by this disjunction
during the execution of the last disjunct,
so that if they are removed, they can be reclaimed.

Task 2, the reclaiming of memory, consists of two parts.
Instant reclaiming of new regions happens
at the beginning of every nonfirst disjunct (at d2.a and d3.a);
the new regions are identified as such by their sequence numbers.
Instant reclaiming of new allocations in old regions
also happens at the beginning of every nonfirst disjunct (at d2.b and d3.b).
To allow us to restore each old region to its state before the disjunction,
each disj frame contains a list
of the old regions that are allocated into during the disjunction,
together with the sizes of these regions
at the start of the disjunction (d1.c).

Task 1 needs extra support in the case of semidet disjunctions.
The disj frames of such disjunctions have a list of the disj-protected regions,
namely the regions in the removed list of the disjunction
which are disj-protected only by this disj frame (set at d1.d).
We use this list to explicitly reclaim these regions
if a nonlast disjunct succeeds (e1 and e2).

\subsection{Support for If-then-elses}
\label{seCsupportnondetCite}

The condition of an if-then-else (ite) can be either semidet or nondet.
In most Mercury programs, the overwhelming majority are semidet,
and this is the case we will look at first.
Such if-then-elses share some properties with semidet disjunctions.
If the condition succeeds, the execution will never enter the else part,
and if the condition fails,
the failure must have occurred in the scope of the condition.

Like disjunctions, if-then-elses
need to protect regions from being reclaimed while backward live.
But in the case of if-then-elses,
we can restrict our attention to regions removed in the condition
(i.e., in the condition's removed set),
since this is the only part of the code in which
the if-then-else itself can make a region backward live.
When execution reaches the start of the then part,
backtracking to the else part is no longer possible, which means that
any regions that have been marked for removal in the condition
have to be reclaimed for real, unless they are protected by a surrounding scope.

Also, if-then-elses, like disjunctions, should do instant reclaiming
of memory allocated by backtracked-over computations.
In the case of if-then-elses,
this means that at the start of the else part,
we should recover any memory allocated by the condition.

In general, we only need to provide support
for changes to regions which occur inside the condition.
This is good, because the conditions of if-then-elses are often very simple,
containing only one or a few tests.
Conditions whose created, removed and allocated sets are all empty
are therefore fairly common.
For such if-then-elses, the mechanisms we describe below are unnecessary,
and so we optimize them away.
If at least one these three sets is not empty,
we add code at the starts of the condition, the then part, and the else part,
i.e., at points (i1), (i2), and (i3) in Figure~\ref{fig:supportite}.
\begin{figure}[htb]
\scriptsize
\begin{Verbatim}[frame=single,framerule=0.2pt,framesep=3pt]
( if
    (i1): start of the condition
        (a) push an ite frame
        (b) save the protected regions and their number
        (c) save size records and their number
    ...
  then
    (i2): start of the then part
        (a) reclaim the ite-protected regions
        (b) pop the ite frame
    ...
  else
    (i3): start of the else part
        (a) unprotect the ite-protected regions
        (b) do instant reclaiming of new regions
        (c) do instant reclaiming of allocations in old regions
        (d) pop the ite frame
    ...
)
\end{Verbatim}
\small
\caption{RBMM runtime support for if-then-else with semidet condition.}
\label{fig:supportite}
\normalsize
\end{figure}

For each if-then-else, we use a data structure called an \emph{ite frame}
to store the information used for its runtime support.
As with disj frames,
we embed ite frames in the ordinary stacks used by the Mercury implementation,
and link them together into the \emph{ite stack},
with a global variable pointing to its top.
The structure of an ite frame is exactly analogous to that of a disj frame,
the only difference being that
the first slot of the fixed part, \code{prev\_ite\_frame},
holds a pointer to the previous ite frame, or null if there is none.

\noindent\textbf{Ite-protecting backward live regions.}
Since the compiler knows the regions in the removed set of the condition
(in our example in Figure \ref{fig:running2:annotated},
\code{R1} is such a region),
we will stop them from being reclaimed by \emph{ite-protecting} them
at the entry to the if-then-else.
To allow us to ite-protect regions,
we add to the region header a pointer field, \code{ite\_protected},
which is set to null when a region is created.
A region is ite-protected if its \code{ite\_protected} field is not null.
The \code{remove} instruction will now only reclaim a region if its
\code{ite\_protected} field is null and it is not disj-protected.
(We do not use the same protection mechanism as in the case of disjunctions.
We will explain the reason for this
when we describe how we handle if-then-elses with nondet conditions.)
Before entering the condition, i.e.\ at (i1), we push an ite frame,
and then iterate over the to-be-protected regions.
If one of these regions is already protected
by a surrounding disjunction or if-then-else,
we ignore it.
Otherwise, we protect it
by setting its \code{ite\_protected} field, which must be currently null,
to point to the ite frame.
For such a protected region,
we add its handle to a \code{region\_id} slot
in the nonfixed part of the ite frame.
Then we also put the final number of regions we protect in this way
into the frame's \code{num\_prot\_region} slot.
We do this so that
we can loop over all the regions protected by this ite frame in two places:
at the start of the then part (i2.a), where we reclaim all these regions
(giving delayed effect to the \code{remove} instructions in the condition),
and at the start of the else part (i3.a), where we undo their protection
by resetting their \code{ite\_protected} fields to null.

\noindent\textbf{Instant reclaiming.}
When the condition fails, we want to reclaim
both the new regions created inside it
and any new allocations into old regions.
In our example in Figure~\ref{fig:running2:annotated}
we want to reclaim all of \code{R3} and some of \code{R2}.

To reclaim new regions,
at (i1.a) we save the current sequence number
into the new frame's \code{saved\_seq\_num} slot,
and at (i3.b), we add code that traverses the region list
and reclaims all the regions until it meets an old region.

To reclaim new allocations into an old region,
at (i1.c) we save its size record into the nonfixed part of the ite frame.
Although it is reasonable to do this
for the regions in the allocated set of the condition,
it would be wasteful to reclaim new allocations into the regions
which will be reclaimed right at the start of the else part.
Unfortunately, while the compiler knows which old regions
have \code{remove} instructions at the start of the else part,
it does not know which of these will actually reclaim their regions,
since it does not know which regions are protected by surrounding code.
We handle this uncertainty as follows.
We generate code at (i1.c) for every old region which is live at that point.
For those that are not removed at the start of the else branch,
this code always saves their size records unconditionally.
For those that are removed at the start of the else branch,
this code checks whether they are protected \emph{before} this if-then-else,
and saves their size records only if they are.
This is an optimization
because the test to see if a region is protected
takes less time than
saving its size record, and restoring it if the condition fails.
We record the number of size records
we saved in the \code{num\_size\_record} slot,
so that code at (i3.c) can restore them all.

The final action of the support code
for an if-then-else with a semidet condition
is to pop the ite frame at either (i2.b) or (i3.d).

\noindent\textbf{If-then-else with nondet condition.}
Unlike Prolog,
Mercury allows the condition of an if-then-else to have more than one solution.
If the condition is nondet,
then execution can backtrack into the condition
from the then part or later code.
This poses two problems we need to solve.

First,
since the condition can succeed more than once,
the code we add at the start of the then part (i2)
can also be executed more than once.
Because we need the ite frame every one of these times,
we cannot let the code pop it at (i2.b);
we must keep it until after the last time it may be used,
i.e., after the last success of the condition.
We arrange for this to happen
by modifying the way the code generator handles the failure of the condition.

Normally, the code generator arranges for failures of the condition
\emph{before} the condition succeeds for the first time
to cause a branch to the start of the else part,
while a failure of the condition \emph{after} it has succeeded
represents a failure of the if-then-else as a whole,
and will be handled accordingly,
in whatever way the surrounding context demands.
For example, if the if-then-else is one disjunct of a disjunction,
its failure will cause execution to resume at the start of the next disjunct.
We call the place to branch to on failure of the whole if-then-else
the \emph{failure continuation}.

We modified the code generator so that if the nondet condition
needs support for region operations,
i.e., it has a nonempty created set, removed set or allocated set,
we branch to the failure continuation
only after we execute code to pop the ite frame,
the same code that for semidet conditions we would execute at (i2.b).

Second, the condition being nondet means
that it must include, directly or indirectly, a nondet disjunction
(since this is the only Mercury construct that can introduce nondeterminism).
Therefore we must ensure that the supporting code fragments we generate
for the if-then-else and the disjunction inside it
do not step on each other's toes.

Our support for if-then-elses with semidet conditions provides ite-protection
for regions in the condition's removed set
that are not yet protected before the if-then-else.
For such a region in a nondet condition, there are two cases.
The first case is when the region is removed
before the first nondet disjunction inside the condition.
That means that when the \code{remove} instruction is executed,
the region is ite-protected but not disj-protected.
The \code{remove} instruction will (correctly) not reclaim it.
Later on, the region will be reclaimed
when the condition succeeds for the first time by
the supporting code added at (i2).
Because the program may backtrack into the condition and
may reach the then part again,
when the region is reclaimed at (i2.a),
we need to nullify its entry in the ite frame
so that it will not be wrongly reclaimed again
the next time execution reaches (i2.a).
This explains our saving of the pointer to the ite frame
in the \code{ite\_protected} field in the region header of a protected region.

In the second case, the region is removed after the start
of the first disjunction in the condition,
either in the disjunction itself or at some point after it.
In an execution containing a non-last disjunct,
when the \code{remove} instruction is encountered
the region is not reclaimed because it is both ite- and disj-protected.
We need to ensure that if the condition succeeds
and execution reaches the then part,
the region should not be reclaimed at (i2) because
it may be needed when execution backtracks into the condition.
We therefore put different code at (i2.a) if the condition is nondet;
this code will reclaim a region only if it is not currently disj-protected
(Figure~\ref{fig:supportite:nondet}).
The region will remain both ite- and disj-protected until
the execution enters the last disjunct,
at that time it will lose its disj-protection
(Section \ref{seCsupportnondetCdisjunction}).
When the \code{remove} instruction in the condition is executed after this,
it will not reclaim the region because it is still ite-protected,
but the code at (i2.a) will reclaim it.

\begin{figure}[tb]
\scriptsize
\begin{Verbatim}[frame=single,framerule=0.2pt,framesep=3pt]
for each saved region_id
    if region_id != null && !is_disj_protected(region_id)
        reclaim the region
        region_id  = null
\end{Verbatim}
\small
\caption{Code at (i2.a) for if-then-else with nondet condition.}
\label{fig:supportite:nondet}
\normalsize
\end{figure}

When the nondet condition fails, in both cases above,
the region is only ite-protected, not disj-protected.
It is because in the first case, the region is never disj-protected and
in the second case, the failure happens only after all the disjuncts of the
nondet code have been tried and failed,
and the region has been disj-unprotected at the start of the last disjunct.
This situation is exactly the same as when a semidet condition fails.
Therefore the code at (i3)
is exactly the same for nondet conditions as for semidet conditions.

\subsubsection{If-then-elses: Summary}
To summarize Section \ref{seCsupportnondetCite},
we review how we handle Tasks 1 and 2 for if-then-elses;
first if-then-elses with semidet conditions,
and then those with nondet conditions.

We prevent the reclamation of backward live regions (Task 1)
by ite-protecting any regions that are removed in the condition,
but are backward live, and are not protected by any other mechanism.
The mechanism we use for ite-protection takes the form
of \code{ite\_protected} fields in region headers:
if this field is not null, the region is ite-protected.
At the beginning of the condition (i1.a and i1.b),
we set this field to point to the ite frame of the if-then-else
for all the regions that meet the conditions listed above.
If the condition succeeds, then execution enters the then part,
and the code at (i2.a) reclaims these regions
(since backtracking to the else case is no longer possible,
and the regions are therefore no longer backward live).
If the condition fails, code at (i3.a) unprotects these regions.

Task 2 consists of two parts.
Instant reclaiming of new regions
happens at the beginning of the else part (at i3.b);
as with disjunctions,
new regions are identified as such by their sequence numbers.
Instant reclaiming of new allocations in old regions
also happens at the beginning of the else part (at i3.c).
To allow us to restore the size of the old regions,
each ite frame contains a list of old live regions,
together with their sizes at the start of the if-the-else (set at i1.c).

We need extra support for nondet conditions.
The reclaiming at the beginning of the then part
has to be done only if the region is not disj-protected
by a disjunction inside the condition.
The code that executes this reclaiming
executes once for each success of the condition.
A region may be unprotected by disjunctions inside the condition
for more than one of these executions, yet it must be reclaimed only once.
This is why after we reclaim a region
whose protection by a nondet if-then-else has just expired,
we remove it from the list of regions protected by that if-then-else.

\subsection{Support for Commit}
\label{seCsupportnondetCcommit}
When the goal inside a commit succeeds for the first time,
we commit to that solution by
discarding the inner goal's outstanding alternatives.
We call the point in the code where this happens the \emph{commit point}.
If the inner goal is nondet (rather than multi), it may also fail.
When it fails, the compiler's failure-handling mechanism
causes execution to pass through a \emph{failure point}
before the program resumes forward execution
at the resumption point of the next surrounding goal.
The failure point is there to allow the execution of some cleanup code.
We add code to support region operations at two or three points
in Figure~\ref{fig:supportcommit}:
the entry point of the commit (c1),
the commit point (c2),
and the failure point (c3).
If the inside goal has determinism multi,
there is no (c3) to modify as execution would never reach there.

\begin{figure}
\scriptsize
\begin{Verbatim}[frame=single,framerule=0.2pt,framesep=3pt]
some [...]
  (c1): entry to the commit
    (a) push a commit frame
    (b) save the sequence number
    (c) save the pointer to the top disj frame
    (d) save the pointer to the top ite frame
    (e) save the to-be-reclaimed old regions and their number
  ( the inner goal )

  (c2): commit point
    (a) reclaim the saved old regions
    (b) reclaim the new regions
    (c) restore the state of the disj stack
    (d) restore the state of the ite stack
    (e) pop the commit frame

  (c3): failure point
    (a) restore status of the saved regions
    (b) pop the commit frame
\end{Verbatim}
\small
\caption{RBMM runtime support for commit.}
\label{fig:supportcommit}
\normalsize
\end{figure}

Consider a region that is in the removed set of a commit goal.
If it is already protected by a disjunction or if-then-else
when execution arrives at (c1),
then the region should not be reclaimed by any code inside the commit,
and the mechanisms we have described so far are sufficient to ensure this.
If the region is not already protected at (c1),
then the region should be reclaimed
before execution reaches (c2).
Ensuring this needs a new mechanism
because the goal inside a commit will contain, directly or indirectly,
at least one disjunction that can succeed more than once
(if it did not, it would have at most one solution,
and there would be no commit operation),
and the runtime support for this disjunction will protect the region
from being reclaimed during the execution of its non-last disjuncts.
On the other hand, we cannot simply insert code at (c2)
to reclaim the region,
since it \emph{can} already be reclaimed
by its \code{remove} instruction in the execution of the last disjunct
before reaching (c2).
We do not need to worry about the case when regions are protected
only by semidet disjunctions
or by if-then-elses with semidet conditions inside a commit,
since these constructs, if they occur, protect regions only temporarily,
and ensure that any regions that are removed inside them
and are not protected when execution enters them
will be reclaimed before execution exits them.
If-then-elses with nondet conditions cannot protect regions either,
though the nondet disjunctions inside their conditions can.

As before, our solution involves a new embedded stack, the \emph{commit stack}.
We push a new commit frame at (c1), and fill in its fixed fields,
which we will discuss shortly.
Following this will be the code that,
for each region in the removed set of the commit goal,
checks whether the region is already protected.
If it is, that region is left alone.
If it is not, we add the handle of the region
to the commit frame's nonfixed part,
and record the address where this handle is stored in the commit frame
in the region's own header, in a new field called \code{commit\_slot}.
This way, when a region that should be reclaimed inside the commit
actually survives to (c2) due to the protection of an inner disjunction,
code at (c2) can iterate through all the region handles in the commit frame
and reclaim those regions.
However, we cannot do this for regions
that are actually reclaimed inside the commit
(whose remove instructions were executed in the last disjuncts).
That is why, when we reclaim a region,
we check whether its header's \code{commit\_slot} field is null.
If not, then it will contain
the address of a pointer to the region header,
an address that will be in a commit frame,
and the reclaim operation will replace that pointer in the commit frame
with a null.
Making the loop at (c2.a) ignore such nulled-out region handle pointers
ensures that each region recorded in the commit frame's list
is reclaimed exactly once, and that this will happen as soon as possible.

If the goal inside the commit fails, we need to undo
the update of the saved regions' \code{commit\_slot} fields,
so at (c3.a) we reset them all to their original values.
To make this possible, we save each original value in the commit frame
next to the pointer to the region header from which it is taken.
This effectively chains together
all the entries referring to a given region in the commit stack.
The reclaim operation will set to null not just the first slot in this chain,
but all of them.

This mechanism is sufficient to correctly handle
any old regions that are in the commit goal's removed set.
To handle any new regions (regions created inside the commit)
that are also removed inside the commit,
we record the current region sequence number in the commit frame at (c1).
When a new region is removed in the commit,
if it is not protected, it is reclaimed.
If it is protected, we mark it so that at the commit point we can reclaim it.
We add a field \code{destroy\_at\_commit} to the region header,
and we augment the \code{remove} instruction again so that
when a protected, new region is removed in a commit,
the \code{remove} instruction
will set the region's \code{destroy\_at\_commit} field to true
(it is always initialized to false).
At the (c2.b) part of the commit point,
we traverse the region list until meeting an old region,
and reclaim the new regions whose \code{destroy\_at\_commit} field is true.

We do not need to worry about instant reclaiming
of new regions in the created set
and of new allocations into regions in the allocated set of the commit,
since that will be done by the goals surrounding the commit.

At the commit point, the Mercury execution algorithm
throws away all the remaining alternatives of the goal inside the commit.
To reflect this, at (c2) we need to restore the embedded disj stack
to the state it had at (c1).
This is why at (c1.c),
we save the current disj stack pointer in a fixed slot in the new commit frame,
and at (c2.c), we restore the disj stack pointer from there.
The regions protected by the disj frames thrown away by this action
will be exactly the ones removed by the code at (c2.b).

In some rare cases, the thrown-away disj frames
will be from disjunctions inside if-then-elses with nondet conditions.
Such if-then-elses cannot protect any regions
in any code outside their conditions,
but we do still need to ensure that
we leave the embedded ite stack in the same state as we found it.
This is why at (c1.d) and (c2.d), we save and restore its stack pointer.
(The ite frames of if-then-elses with semidet conditions
will have been popped by the time we get to c2,
but the ite frames of if-then-elses with nondet conditions may still be there.)

The layout of commit frames is shown in Figure~\ref{fig:commitframe},
with the fixed and nonfixed parts are separated by a thick line.

\begin{figure}[htp]
\begin{center}
\expandafter\ifx\csname graph\endcsname\relax
   \csname newbox\expandafter\endcsname\csname graph\endcsname
\fi
\ifx\graphtemp\undefined
  \csname newdimen\endcsname\graphtemp
\fi
\expandafter\setbox\csname graph\endcsname
 =\vtop{\vskip 0pt\hbox{%
    \special{pn 8}%
    \special{pa 0 150}%
    \special{pa 4500 150}%
    \special{pa 4500 0}%
    \special{pa 0 0}%
    \special{pa 0 150}%
    \special{fp}%
    \special{pa 0 300}%
    \special{pa 4500 300}%
    \special{pa 4500 150}%
    \special{pa 0 150}%
    \special{pa 0 300}%
    \special{fp}%
    \special{pa 0 450}%
    \special{pa 4500 450}%
    \special{pa 4500 300}%
    \special{pa 0 300}%
    \special{pa 0 450}%
    \special{fp}%
    \special{pa 0 600}%
    \special{pa 4500 600}%
    \special{pa 4500 450}%
    \special{pa 0 450}%
    \special{pa 0 600}%
    \special{fp}%
    \special{pa 0 750}%
    \special{pa 4500 750}%
    \special{pa 4500 600}%
    \special{pa 0 600}%
    \special{pa 0 750}%
    \special{fp}%
    \special{pa 0 900}%
    \special{pa 4500 900}%
    \special{pa 4500 750}%
    \special{pa 0 750}%
    \special{pa 0 900}%
    \special{fp}%
    \special{pa 0 1050}%
    \special{pa 4500 1050}%
    \special{pa 4500 900}%
    \special{pa 0 900}%
    \special{pa 0 1050}%
    \special{fp}%
    \special{pa 0 1200}%
    \special{pa 4500 1200}%
    \special{pa 4500 1050}%
    \special{pa 0 1050}%
    \special{pa 0 1200}%
    \special{fp}%
    \special{pa 0 730}%
    \special{pa 4500 730}%
    \special{fp}%
    \special{pa 0 740}%
    \special{pa 4500 740}%
    \special{fp}%
    \graphtemp=.5ex
    \advance\graphtemp by 0.075in
    \rlap{\kern 0.050in\lower\graphtemp\hbox to 0pt{prev\_commit\_frame\hss}}%
    \graphtemp=.5ex
    \advance\graphtemp by 0.225in
    \rlap{\kern 0.050in\lower\graphtemp\hbox to 0pt{saved\_seq\_num\hss}}%
    \graphtemp=.5ex
    \advance\graphtemp by 0.375in
    \rlap{\kern 0.050in\lower\graphtemp\hbox to 0pt{saved\_disj\_sp\hss}}%
    \graphtemp=.5ex
    \advance\graphtemp by 0.525in
    \rlap{\kern 0.050in\lower\graphtemp\hbox to 0pt{saved\_ite\_sp\hss}}%
    \graphtemp=.5ex
    \advance\graphtemp by 0.675in
    \rlap{\kern 0.050in\lower\graphtemp\hbox to 0pt{num\_saved\_regions\hss}}%
    \graphtemp=.5ex
    \advance\graphtemp by 0.825in
    \rlap{\kern 0.050in\lower\graphtemp\hbox to 0pt{region\_id\hss}}%
    \graphtemp=.5ex
    \advance\graphtemp by 0.975in
    \rlap{\kern 0.050in\lower\graphtemp\hbox to 0pt{prev\_commit\_slot\hss}}%
    \graphtemp=.5ex
    \advance\graphtemp by 1.125in
    \rlap{\kern 0.050in\lower\graphtemp\hbox to 0pt{...\hss}}%
    \graphtemp=.5ex
    \advance\graphtemp by 0.075in
    \rlap{\kern 1.500in\lower\graphtemp\hbox to 0pt{\emph{(previous commitframe)}\hss}}%
    \graphtemp=.5ex
    \advance\graphtemp by 0.225in
    \rlap{\kern 1.500in\lower\graphtemp\hbox to 0pt{\emph{(saved sequence number)}\hss}}%
    \graphtemp=.5ex
    \advance\graphtemp by 0.375in
    \rlap{\kern 1.500in\lower\graphtemp\hbox to 0pt{\emph{(saved disj stack pointer)}\hss}}%
    \graphtemp=.5ex
    \advance\graphtemp by 0.525in
    \rlap{\kern 1.500in\lower\graphtemp\hbox to 0pt{\emph{(saved if-then-else stack pointer)}\hss}}%
    \graphtemp=.5ex
    \advance\graphtemp by 0.675in
    \rlap{\kern 1.500in\lower\graphtemp\hbox to 0pt{\emph{(number of saved regions)}\hss}}%
    \graphtemp=.5ex
    \advance\graphtemp by 0.825in
    \rlap{\kern 1.500in\lower\graphtemp\hbox to 0pt{\emph{(handle of a saved region)}\hss}}%
    \graphtemp=.5ex
    \advance\graphtemp by 0.975in
    \rlap{\kern 1.500in\lower\graphtemp\hbox to 0pt{\emph{(original commit slot of the saved region)}\hss}}%
    \hbox{\vrule depth1.200in width0pt height 0pt}%
    \kern 4.500in
  }%
}%
\centerline{\box\graph}
\end{center}
\caption{The structure of a commit frame.}
\label{fig:commitframe}
\end{figure}

The meaning of the first two fields should be clear.
The third and fourth fields contain
the values of the disj and ite stack pointers respectively
at the time when the commit was entered.
The last fixed field gives the number
of region handles and saved \code{commit\_slot} fields
actually stored by the code at (c1.d) in the nonfixed part.

\subsubsection{Commits: Summary}
To summarize Section \ref{seCsupportnondetCcommit},
we review how we handle commits.

A commit does not need to protect any regions against reclamation,
as it does not make any regions backward live.
When the commit goal succeeds,
it cuts away any backtrack points set up inside it,
so we need to take away all the protections
associated with those backtrack points,
and if this leaves a region (old or new) unprotected,
we need to reclaim it.

We keep in each commit frame
a list of the old regions (existing before the commit)
that may be subject to such reclamation.
We store this list at (c1.e),
and we reclaim the regions in it at (c2.a),
provided they have not been reclaimed within the commit goal itself,
by code executing within or after a last disjunct.
We set the \code{commit\_slot} of each of these regions' headers
to point to their entry in the commit frame;
if and when the region is reclaimed within the commit goal,
we delete this entry to prevent double reclamation.

Since commits may be nested,
a given to-be-reclaimed region may be listed in several commit frames.
We keep its entries in these frames in a chain,
and when a region is actually reclaimed,
we delete its entries in \emph{all} these frames.

To reclaim new regions,
we store a snapshot of the sequence region number
in the commit frame at (c1.b).
When the commit goal succeeds, we reclaim all regions younger than this
whose \code{destroy\_at\_commit} field has been set to true
by a \code{remove} instruction.

If the commit goal fails,
all the protections set up by any disjunctions or if-then-elses inside it
must have expired already, so we need do no more than
simply restore the commit stack to its original state.

\subsection{Compatibility with Tabling}
\label{seCtabling}

Mercury supports three forms of tabling:
loop checking (which detects the simplest form of infinite loops,
and aborts the program if found),
memoization (caching of results),
and minimal model tabling.

The mechanisms we have discussed in this section so far
are compatible with loop checking
because the only two changes loop checking makes to the flow of execution
are to force the execution of some table lookups,
which have no effect on our data structures,
and (maybe) to abort the program,
in which case what our mechanisms do does not matter.

Our mechanisms are also compatible with automatic caching
for det and semidet procedures.
This tabling method surrounds the body of the tabled procedure
with code that checks whether
a call with the current argument list has been seen before.
If it has not been seen,
it computes the answer and records it.
For det procedures, the answer consists of the values of the output arguments;
for semidet procedures, it includes the success/failure indication as well.
If this call has been seen before,
the transformed procedure just returns the recorded answer.
Neither the extra code executed at the starts and ends of new calls
nor the table lookup executed for previously-seen calls
interfere with any of our mechanisms.

Automatic caching for nondet and multi procedures is a more complex case,
because the code that adds answers to a table adds one answer at a time,
and only when execution is about to backtrack out of a new call
does the tabling system know that its set of answers is complete.
The Mercury system handles the interaction of
tabled nondet/multi procedures with commits,
just as it handles handles the interaction of
nondet/multi procedures using RBMM with commits,
but it does not handle the interaction of
tabled nondet/multi procedures using RBMM with commits.
There is no reason why it could not do so,
we just have not implemented it yet,
mainly because memoization is not as useful
for nondet and multi procedures as minimal model tabling.

The current implementation of minimal model tabling in the Mercury system
works by saving segments of the Mercury stacks and restoring them later,
possibly several times \cite{mercurytabling}.
This makes minimal model tabling fundamentally incompatible
with the mechanisms we have presented earlier in this section.

\section{Experimental Evaluation}
\label{seCexper}

\subsection{The Experimental Systems}
We have implemented the region analysis and transformation shown in
Sections~\ref{seCrpta},~\ref{seClra}, and~\ref{seCpt},
as well as the runtime support describe in
Sections~\ref{seCsupportdet} and~\ref{seCsupportnondet}
by incorporating them in the Melbourne Mercury compiler.
The runtime support is currently available
in the backend that generates low-level C code.

We use three variants of our RBMM system in our experiments.
The first one, \emph{rbmm1}, is
similar to the RBMM system in~\cite{Phan08ismm}
in which we do not track which regions that are allocated into.
In rbmm1, while the region operations (Section~\ref{seCsupportdet})
are implemented as C functions,
the runtime support for backtracking (Section~\ref{seCsupportnondet})
is implemented using C macros.
The functionality of the second system, \emph{rbmm2},
is exactly the same as rbmm1,
however we consistently implement the whole runtime support in functions.
The third system, \emph{rbmm3}, also uses only functions in the runtime system,
but differs from rbmm2 in that
it \emph{does} track which regions are allocated into
(using the algorithms in Section~\ref{seCallocCregions}),
which allows us to restrict the set of old regions
for which we take size snapshots for later reclaiming
(see Section~\ref{seCsupportnondet})
to just the regions for which this may have an effect.
We chose these three versions to evaluate
because comparing rbmm1 and rbmm2 tells us
which implementation technology is better,
while comparing rbmm2 and rbmm3 can reveal
the impact of tracing which regions are allocated into and which are not.
We also compare these RBMM variants with a Mercury compiler
that is identical in all aspects except that instead of RBMM,
it uses the Boehm garbage collector~\cite{Boehm88},
which is Mercury's standard garbage collector.
We call this system \emph{boehm}.

For all three RBMM systems,
we use a region page size of 2,048 words,
of which 2,047 are available to store program data.
When needed, we request blocks of 100 region pages from the OS.
The three systems use the same regions
and create and remove them in exactly the same places.
However, they do differ in other aspects,
such as compilation time, size of object files, and runtime performance.

Next, we will present the benchmarks and give the results of our experiments,
and then we will discuss the RBMM behavior of the benchmarks in more detail.
The experiments were performed on a Dell Optiplex 760 PC with
a 2.83 GHz Core 2 Quad Q9550 CPU, 8 GB of RAM,
running Ubuntu Linux, with the kernel version being 2.6.24-25-server SMP.
The Mercury programs were compiled to C with
the 3 December 2009 release-of-the-day of the Mercury system
(with different options for the different variants).
This and other releases-of-the-day are available on the Mercury web site.
The resulting C files were compiled to executables by gcc 3.4.4.
Every time we report was derived
by running the program eight times,
discarding the lowest and highest times,
and averaging the rest.

\subsection{The Benchmark Programs}

In our experiments, we used a set of relatively small benchmark programs.
We selected the benchmarks carefully;
they are actually more like a collection of case studies
that illustrate the strong and weak points of RBMM.
While we would have liked to test our system
with bigger, more realistic programs,
we are currently not able to do so
because the region analysis and transformation do not yet support
higher order code, foreign language code and multi-module programs.

\begin{table}[htb]
  \centering
  \small
  \caption{Information about the benchmarks.}
  \begin{tabular}{l|*{2}{r|}*{3}{c|}}
    \hline
    \hline
    \multirow{2}{*}{}
    & \multicolumn{1}{c|}{\multirow{2}{*}{\textbf{\# Predicates}}}
    & \multicolumn{1}{c|}{\multirow{2}{*}{\textbf{\# LOC}}}
    & \multicolumn{1}{c|}{\multirow{2}{*}{\textbf{if-then-else}}}
    & \multicolumn{2}{c|}{\textbf{disjunction}} \\
    & & &
    & \multicolumn{1}{c}{\textbf{semidet}}
    & \multicolumn{1}{c|}{\textbf{nondet}} \\
    \hline
    dna      & 16 & 251 & x &   &   \\
    isort    &  6 & 101 & x &   &   \\
    nrev     &  5 &  72 & x &   &   \\
    primes   &  8 &  93 & x &   &   \\
    qsort    &  6 &  92 & x &   &   \\
    \hline
    bigcatch & 12 & 159 & x &   &   \\
    boyer    & 17 & 372 & x &   &   \\
    bsolver  & 41 & 805 & x & x &   \\
    crypt    & 15 & 219 & x &   & x \\
    filrev   & 12 & 154 & x &   &   \\
    life     & 18 & 338 & x &   &   \\
    healthy  & 24 & 485 & x &   & x \\
    queens   &  9 & 128 & x &   & x \\
    sudoku   & 22 & 441 & x &   & x \\
    \hline
    rdna     & 17 & 262 & x &   &   \\
    risort   &  7 & 111 & x &   &   \\
    rlife    & 19 & 343 & x &   &   \\
    rqueens  & 10 & 138 & x &   & x \\
    \hline
    \hline
  \end{tabular}
  \normalsize
  \label{table:experiment:benchmarks}
\end{table}

The benchmark programs in Table~\ref{table:experiment:benchmarks}
are divided into three groups.
The first group contains benchmarks
that do not need any runtime support for backtracking.
The benchmarks in the second group do need such support.
The third group consists of manually modified versions of benchmarks
that illustrate how programs can be made more region-friendly
(hence the ``r'' as prefix on their names).

The programs in the first group contain only det code,
and maybe some if-then-elses with semidet conitions
whose created, removed and allocated sets are all empty.
\bench{dna} computes similarities between gene sequences,
\bench{isort} implements insertion sort on a list of 10000 integers,
\bench{nrev} reverses a list of 5000 integers,
\bench{primes} finds all the primes less than 20000, and
\bench{qsort} sorts a list of 100000 integers.

The programs in the second group need runtime support
for if-then-elses and/or disjunctions.
\bench{bigcatch} and \bench{filrev} are
Mercury versions of programs used in~\cite{Aspinall08usageaspect}.
They manipulate lists of lists of integers
and introduce sharing between the input, the temporary data and the output
and as such they also present difficult cases for RBMM.
\bench{bsolver} is a simple solver for systems of
binary linear equations and inequations over integers;
\bench{boyer} is a toy theorem prover;
\bench{crypt} finds the unique answer to a cryptoarithmetic puzzle;
\bench{life} implements the Game of Life
(known to be a difficult case for RBMM);
\bench{healthy} is a nondeterministic variant of \bench{life}
that searches for a generation that after a certain number of reproductions (8)
still has a number of live cells that is higher than a threshold (80);
\bench{queens} solves the 12-queens problem
by first generating permutations and then checking;
\bench{sudoku} finds the solution for a sudoku puzzle
by doing propagation on finite domains.

The programs \bench{rlife} and \bench{rdna}
are versions of \bench{life} and \bench{dna}
that have been manually made region-friendly
by copying some data instead of letting it be shared.
\bench{rqueens} is a modified form of \bench{queens};
its \code{delete} predicate (called by \code{permute})
copies the list remaining after a deletion.
Similarly, \bench{risort} copies the remaining list
when inserting an element into a sorted list.
We will come back to this group of programs
when discussing the benchmarks in detail.

\subsection{Experimental Results}

\subsubsection{Compilation Times and Object File Sizes}

We first compare the three RBMM systems and the Boehm system
with respect to their compilation times and the sizes of their object files
(the text sections).
The results are given in Table~\ref{table:experiment:ct_os}, which contains
two sets of columns, for compilation time and object file size respectively.
The first four columns in each group report results
for each of our four system variants, rbmm1/2/3 and boehm,
while the fifth column is computed by (rbmm3 - boehm)/boehm * 100.

\begin{table}[ht]
  \centering
  \scriptsize
  \caption{Compilation time and object file size.}
  \begin{tabular}{l|*{4}{r|}r||*{5}{r|}}
    \hline
    \hline
    \multirow{3}{*}{}
    & \multicolumn{5}{c||}{\textbf{Compilation time (s)}}
    & \multicolumn{5}{c|}{\textbf{Object file size (bytes)}} \\
    & \multicolumn{1}{c|}{\multirow{2}{*}{\textbf{boehm}}}
    & \multicolumn{3}{c|}{\textbf{rbmm}}
    & \multicolumn{1}{c||}{\textbf{r3/b}}
    & \multicolumn{1}{c|}{\multirow{2}{*}{\textbf{boehm}}}
    & \multicolumn{3}{c|}{\textbf{rbmm}}
    & \multicolumn{1}{c|}{\textbf{r3/b}}\\
    &
    & \multicolumn{1}{c}{\textbf{1}}
    & \multicolumn{1}{c}{\textbf{2}}
    & \multicolumn{1}{c|}{\textbf{3}}
    & \multicolumn{1}{c||}{\textbf{(\%)}}
    &
    & \multicolumn{1}{c}{\textbf{1}}
    & \multicolumn{1}{c}{\textbf{2}}
    & \multicolumn{1}{c|}{\textbf{3}}
    & \multicolumn{1}{c|}{\textbf{(\%)}}\\
    \hline
    dna      & 0.51 & 0.66 & 0.60 & 0.60 & 18 &  4,782 &  6,670 &  6,366 &  6,142 & 28 \\
    isort    & 0.41 & 0.47 & 0.43 & 0.45 & 10 &  1,048 &  1,800 &  1,512 &  1,512 & 44 \\
    nrev     & 0.38 & 0.43 & 0.42 & 0.43 & 13 &    976 &  1,728 &  1,408 &  1,408 & 44 \\
    primes   & 0.39 & 0.44 & 0.43 & 0.42 &  8 &  1,026 &  1,712 &  1,408 &  1,408 & 37 \\
    qsort    & 0.41 & 0.47 & 0.45 & 0.47 & 15 &  1,209 &  2,088 &  1,768 &  1,768 & 46 \\
    \hline
    bigcatch & 0.45 & 0.45 & 0.49 & 0.42 & -7 &  1,601 &  3,569 &  2,657 &  2,241 & 40 \\
    boyer    & 0.78 & 1.23 & 1.20 & 1.18 & 51 & 13,748 & 21,509 & 17,716 & 16,165 & 18 \\
    bsolver  & 0.97 & 1.37 & 1.35 & 1.25 & 29 & 16,034 & 26,227 & 22,867 & 18,931 & 18 \\
    crypt    & 0.57 & 0.67 & 0.68 & 0.58 &  2 &  5,656 &  9,808 &  7,184 &  7,136 & 26 \\
    filrev   & 0.40 & 0.47 & 0.48 & 0.48 & 20 &  1,650 &  3,105 &  2,561 &  2,401 & 46 \\
    life     & 0.56 & 0.70 & 0.67 & 0.67 & 20 &  5,564 &  9,771 &  8,123 &  7,147 & 28 \\
    healthy  & 0.61 & 0.95 & 0.77 & 0.78 & 28 &  7,906 & 16,610 & 11,988 & 10,498 & 33 \\
    queens   & 0.42 & 0.48 & 0.46 & 0.47 & 12 &  1,880 &  3,619 &  2,595 &  2,563 & 36 \\
    sudoku   & 0.65 & 0.87 & 0.87 & 0.85 & 31 &  7,685 & 11,989 & 11,077 & 10,213 & 33 \\
    \hline
    rdna     & 0.55 & 0.62 & 0.59 & 0.61 & 11 &  4,831 &  6,815 &  6,511 &  6,287 & 30 \\
    risort   & 0.40 & 0.43 & 0.46 & 0.44 & 10 &  1,194 &  2,040 &  1,752 &  1,752 & 47 \\
    rlife    & 0.55 & 0.70 & 0.71 & 0.66 & 20 &  5,741 & 10,284 &  8,652 &  7,628 & 33 \\
    rqueens  & 0.43 & 0.50 & 0.43 & 0.49 & 14 &  2,155 &  3,941 &  2,933 &  2,901 & 35 \\
    \hline
    \hline
  \end{tabular}
  \normalsize
  \label{table:experiment:ct_os}
\end{table}

Compilation times for most benchmarks
are so short that we get significant fluctuations due to clock granularity;
times in the table that differ only by a couple of tenths of seconds
are effectively indistinguishable in practice.
That said,
compilation is always somewhat slower for the RBMM systems
than when targeting the Boehm collector,
which is not surprising, given the analysis we have to do.
However, the cost of including RBMM is reasonable;
the average slowdown for rbmm3 is 17\%,
and it is only a bit higher for rbmm1 and rbmm2.
Compilation with the function-based systems
is usually faster than compilation for the partly macro-based rbmm1
because the runtime support functions in rbmm2 and rbmm3
are compiled just once (when the runtime system itself is built)
while in rbmm1 the macros containing their functionality
are expanded and compiled several times
during the compilation of each benchmark.
Compared to rbmm2,
tracing and making use of the allocated regions in rbmm3
sometimes helps to reduce the compilation time,
but the effect is quite small.
This is because the overhead of tracking is rather small,
and having information about allocated regions
allows the compiler to do less work:
it does not need to pass as many region arguments in calls,
and it can skip adding some runtime support code.

The object files of the RBMM systems are, as expected,
larger than those of the Boehm system.
The use of macros in rbmm1 can double the size compared to boehm,
as shown by \bench{bigcatch} and \bench{healthy},
with average increase being 74\%.
Replacing macros with calls reduces the overhead significantly;
the object size ratio between rbmm2 and boehm
ranges from 27\% to 66\%, averaging 43\%.
Rbmm3 yields even smaller object files,
since keeping track of allocated-into regions
allows the compiler to reduce
the number of region arguments passed
and the amount of support code generated;
the object size ratio between rbmm3 and boehm
ranges from 18\% to 47\%, averaging only 35\%.
This shows that for larger programs, rbmm3 is likely to be preferable.

\subsubsection{Memory Usage}

We measured the memory consumption of the regions for the RBMM systems.
Note also that the runtime support consumes some memory
as will be discussed later.
Here we focus on the storage of program data.
The results in Table~\ref{table:experiment:memory}
are the same in all three RBMM systems.
For each benchmark,
we give the total number of regions created during its execution,
and the maximum number of regions coexisting during its run.
We also include the total number of words allocated
and the maximum number of words that coexist.
\emph{SLR} is the Size of the Largest Region and
\emph{S (\%)} is the saving, calculated by 1 - Max words/Total words.

\begin{table}[tb]
  \centering
  \small
  \caption{Memory use in rbmm systems.}
  \begin{tabular}{l|*{6}{r|}}
    \hline
    \hline
    \multirow{2}{*}{} & \multicolumn{2}{c|}{\textbf{Regions}}
    & \multicolumn{2}{c|}{\textbf{Words used}}
    & \multicolumn{1}{c|}{\multirow{2}{*}{\textbf{SLR}}}
    & \multicolumn{1}{c|}{\multirow{2}{*}{\textbf{S (\%)}}} \\
    & \multicolumn{1}{c}{\textbf{Total}}
    & \multicolumn{1}{c|}{\textbf{Max}}
    & \multicolumn{1}{c}{\textbf{Total}}
    & \multicolumn{1}{c|}{\textbf{Max}} & & \\
    \hline
    dna      & 2,082,006 &   8 &  18,926,797 &  4,590,797 & 4,096,000 & 75.7 \\
    isort    &         3 &   1 &  67,029,222 & 67,009,222 &67,009,222 &  0.0 \\
    nrev     &     5,003 &   2 &  25,015,000 &     10,000 &    10,000 & 99.9 \\
    primes   &     2,265 &   1 &   5,221,386 &     39,998 &    39,998 & 99.2 \\
    qsort    &   200,003 &  21 &   5,865,744 &    200,000 &   200,000 & 96.6 \\
    \hline
    bigcatch &         3 &   2 &  25,015,000 & 25,015,000 &25,005,000 &  0.0 \\
    boyer    &         5 &   3 &     143,561 &    143,561 &   143,505 &  0.0 \\
    bsolver  &        78 &   7 &   2,914,444 &  2,911,528 & 2,908,442 &  0.1 \\
    crypt    &       417 &   3 &       3,442 &         94 &        64 & 97.3 \\
    filrev   &         6 &   3 &  25,023,004 & 25,019,000 &25,009,000 &  0.0 \\
    life     &    50,304 & 102 &     894,336 &      8,208 &     6,486 & 99.1 \\
    healthy  & 3,917,124 &  82 &  62,639,310 &      2,794 &     2,054 & 99.9 \\
    queens   & 4,545,703 &   2 & 121,453,230 &        114 &        90 & 99.9 \\
    sudoku   &     6,651 &  88 &      84,080 &     16,678 &    10,916 & 80.1 \\
    \hline
    rdna     & 2,083,006 &   9 &  18,930,797 &    501,752 &   428,733 & 97.3 \\
    risort   &   373,214 &   1 & 289,968,666 &      2,000 &     2,000 & 99.9 \\
    rlife    &    50,356 & 102 &     894,594 &      2,056 &     1,722 & 99.8 \\
    rqueens  &23,080,416 &  13 & 142,047,288 &        156 &        24 & 99.9 \\
    \hline
    \hline
  \end{tabular}
  \normalsize
  \label{table:experiment:memory}
\end{table}

RBMM achieves optimum memory management
in \bench{nrev}, in \bench{primes}, and in \bench{qsort}.
For the nondeterministic programs \bench{crypt}, \bench{healthy},
\bench{queens}, and \bench{sudoku}, the memory savings are also high.
The impact of instant reclaiming on memory reuse
differs among these programs
(Table~\ref{table:experiment:impact_rts}):
in \bench{crypt} and \bench{queens},
instant reclaiming collects most of the words,
while in \bench{healthy} it collects only a small fraction
and it reclaims none at all in \bench{sudoku}.

\begin{table}[htb]
  \centering
  \small
  \caption{Words reclaimed by runtime support.
  (Other words are reclaimed by \code{remove} instructions.)
  Only programs with some nontrivial numbers are shown.}
  \scriptsize
  \begin{tabular}{l|*{8}{r|}}
    \hline
    \hline
    \multirow{2}{*}{}
    & \multicolumn{2}{c|}{\textbf{New allocations}}
    & \multicolumn{2}{c|}{\textbf{New regions}}
    & \multicolumn{2}{c|}{\textbf{Start of then}}
    & \multicolumn{2}{c|}{\textbf{Commit point}} \\
    & \multicolumn{1}{c}{\textbf{Words}}
    & \multicolumn{1}{c|}{\textbf{\%}}
    & \multicolumn{1}{c}{\textbf{Words}}
    & \multicolumn{1}{c|}{\textbf{\%}}
    & \multicolumn{1}{c}{\textbf{Words}}
    & \multicolumn{1}{c|}{\textbf{\%}}
    & \multicolumn{1}{c}{\textbf{Words}}
    & \multicolumn{1}{c|}{\textbf{\%}} \\
    \hline
    bigcatch  &          0 &  0.00 &           0 &  0.00 & 10,000 & 0.04 &     0 & 0.00 \\
    crypt     &          0 &  0.00 &       3,270 & 95.00 &      0 & 0.00 &     6 & 0.17 \\
    queens    & 12,356,378 & 10.17 & 109,096,776 & 89.83 &     52 & 0.00 &     0 & 0.00 \\
    rqueens   &          0 &  0.00 & 133,809,696 & 94.20 &      0 & 0.00 &   132 & 0.00 \\
    healthy   &     81,862 &  0.13 &       3,314 &  0.01 &      0 & 0.00 &     0 & 0.00 \\
    sudoku    &          0 &  0.00 &           0 &  0.00 &      0 & 0.00 & 6,480 & 7.71 \\
    \hline
    \hline
  \end{tabular}
  \normalsize
  \label{table:experiment:impact_rts}
\end{table}

For cases such as \bench{isort}, \bench{bigcatch},
\bench{bsolver} and \bench{filrev},
we see that most of the memory goes to the biggest region.
Typically, this biggest region contains some garbage data,
but as it also holds some live data it cannot be reclaimed.

The boehm version of our system uses
the Boehm-Demers-Weiser garbage collector~\cite{Boehm88} for memory management.
In our experiments,
we just use the default configuration of this collector
as it is in the Mercury compiler distribution.
It is a stop-the-world, sequential mark-and-sweep collector
that uses 1024-word pages.
It starts with a heap of 64k words and
heuristically carries out collections of garbage or expands the heap on demand.

Data about memory use in the boehm system
is shown in Table~\ref{table:experiment:boehmdata}.
The second column (\# gc) shows the numbers of times the collector is run while
the third column (\# expans) tells the numbers of expansions of the heap.
The maximal sizes of the heap in kB and words
are shown in the next two columns, respectively.
The maximal numbers of words used and the numbers of words requested
(i.e., 2048 x the number of region pages requested)
in the rbmm systems are shown in the last two columns for reference purpose.

\begin{table}
  \centering
  \scriptsize
  \caption{Memory use in one iteration.}
  \begin{tabular}{|l|*{6}{r|}}
    \hline
    \hline
    \multirow{2}{*}{} & \multicolumn{1}{c|}{\textbf{\# gc}}
    & \multicolumn{1}{c|}{\textbf{\# expans}}
    & \multicolumn{2}{c|}{\multirow{1}{*}{\textbf{boehm max size}}}
    & \multicolumn{2}{c|}{\textbf{rbmm}} \\
    & & & \multicolumn{1}{c|}{\textbf{kB}}
    & \multicolumn{1}{c|}{\textbf{words}}
    & \multicolumn{1}{c|}{\textbf{max words}}
    & \multicolumn{1}{c|}{\textbf{words requested}}\\
    \hline
    dna                &  7 &  4 &  30,524 &  7,814,144 &  4,590,797 &  4,710,400\\
    isort              & 20 &  4 &  30,524 &  7,814,144 & 67,009,222 & 67,174,400\\
    nrev               &  9 &  4 &  30,524 &  7,814,144 &     10,000 &    204,800\\
    primes             &  3 &  4 &  30,524 &  7,814,144 &     39,998 &    204,800\\
    qsort              &  3 &  4 &  30,524 &  7,814,144 &    200,000 &    409,600\\
    \hline
    bigcatch           &  5 & 10 & 119,804 & 30,669,824 & 25,015,000 & 25,190,400\\
    boyer              &  2 &  2 &  17,168 &  4,395,008 &    143,561 &    204,800\\
    bsolver            &  2 &  4 &  30,524 &  7,814,144 &  2,911,528 &  3,072,000\\
    crypt              &  1 &  2 &  17,168 &  4,395,008 &         94 &    204,800\\
    filrev             &  5 & 10 & 119,804 & 30,669,824 & 25,019,000 & 25,190,400\\
    life               &  2 &  3 &  22,892 &  5,860,352 &      8,208 &    409,600\\
    healthy            & 19 &  4 &  30,524 &  7,814,144 &      2,794 &    204,800\\
    queens             & 36 &  4 &  30,524 &  7,814,144 &        114 &    204,800\\
    sudoku             &  1 &  2 &  17,168 &  4,395,008 &     16,678 &    204,800\\
    \hline
    rdna               &  7 &  4 &  30,524 &  7,814,144 &    501,752 &    614,400\\
    risort             & 83 &  4 &  30,524 &  7,814,144 &      2,000 &    204,800\\
    rlife              &  2 &  3 &  22,892 &  5,860,352 &      2,056 &    409,600\\
    rqueens            & 42 &  4 &  30,524 &  7,814,144 &        156 &    204,800\\
    \hline
    \hline
  \end{tabular}
  \scriptsize
  \label{table:experiment:boehmdata}
  \normalsize
\end{table}

The numbers show that, in almost all of the benchmarks,
the RBMM systems can work within spaces
that are smaller than those requested by the Boehm collector.
RBMM systems often need to request only the minimum,
which in our system is 100 * 2048 words.
The worst case for RBMM is \bench{isort}
in which RBMM is not able to reuse memory efficiently.
The boehm system can work with
only a bit more than one tenth the memory in this case.

\subsubsection{Runtime Performance}
\label{seCevalCruntime}

We also studied the runtime performance of our benchmark programs
because this is probably the most important criterion
for the practicality of RBMM.
To control the uncertainty involved in measuring small times,
we ran each program many times in a loop.
Each benchmark has a row in Table~\ref{table:experiment:runtime}
that gives the number of iterations,
the actual execution times with boehm (boehm)
the boehm system's gc time (gc),
and the boehm system's runtime \emph{minus} the gc time (nogc),
and then the runtime with the three RBMM systems
(all in seconds, all for user mode only).
Each row also includes
the number of collections executed by the Boehm collector,
and the savings achieved by using our preferred RBMM system, rbmm3,
instead of the boehm system.
The savings are given by 1 - rbmm3 runtime / boehm runtime.

\begin{table}[tb]
  \centering
  \small
  \caption{Runtime performance result.}
  \scriptsize
  \begin{tabular}{l|*{9}{r|}}
    \hline
    \hline
    \multirow{2}{*}{}
    & \multicolumn{1}{c|}{\multirow{2}{*}{\textbf{\# Iter}}}
    & \multicolumn{4}{c|}{\textbf{boehm runtime}}
    & \multicolumn{3}{c|}{\textbf{RBMM runtime}}
    & \multicolumn{1}{c|}{\textbf{Saving}}\\
    &
    & \multicolumn{1}{c}{\textbf{boehm}}
    & \multicolumn{1}{c}{\textbf{gc}}
    & \multicolumn{1}{c}{\textbf{nogc}}
    & \multicolumn{1}{c|}{\textbf{\# gcs}}
    & \multicolumn{1}{c}{\textbf{rbmm1}}
    & \multicolumn{1}{c}{\textbf{rbmm2}}
    & \multicolumn{1}{c|}{\textbf{rbmm3}}
    & \multicolumn{1}{c|}{\textbf{rbmm3}}\\
    \hline
    dna      &     100 & 25.27 &  8.80 & 16.47 &  549 & 20.81 & 20.20 & 21.19 & 16.1\% \\
    isort    &      60 & 53.47 & 17.90 & 35.57 & 1141 & 21.43 & 21.45 & 21.66 & 59.5\% \\
    nrev     &     160 & 50.09 & 17.58 & 32.51 & 1134 & 20.39 & 20.39 & 21.12 & 57.8\% \\
    primes   &     400 & 40.94 &  9.46 & 31.48 &  597 & 24.86 & 24.51 & 24.62 & 39.9\% \\
    qsort    &     400 & 41.41 & 12.65 & 28.76 &  701 & 20.62 & 20.45 & 21.15 & 48.9\% \\
    \hline
    bigcatch &      30 & 28.31 &  5.70 & 22.61 &   20 & 20.39 & 20.90 & 20.38 & 28.0\% \\
    boyer    &   8,000 & 25.69 &  5.60 & 20.09 &  357 & 22.59 & 34.34 & 34.83 &-35.6\% \\
    bsolver  &    1500 & 55.00 & 19.44 & 35.56 & 1242 & 22.92 & 23.05 & 22.91 & 58.3\% \\
    crypt    & 300,000 & 21.19 &  4.53 & 16.66 &  293 & 18.85 & 20.84 & 20.70 &  2.3\% \\
    filrev   &      50 & 38.40 & 11.03 & 27.37 &   54 & 24.09 & 24.00 & 23.85 & 37.9\% \\
    life     &     700 & 27.18 &  2.77 & 24.41 &  179 & 26.16 & 31.41 & 23.71 & 12.8\% \\
    healthy  &      30 & 37.65 &  8.34 & 29.31 &  533 & 41.63 & 61.12 & 29.62 & 21.3\% \\
    queens   &      15 & 32.90 &  7.97 & 24.93 &  517 & 22.34 & 29.60 & 30.05 &  8.7\% \\
    sudoku   &  20,000 & 23.02 &  6.45 & 16.58 &  413 & 17.65 & 17.69 & 17.57 & 23.7\% \\
    \hline
    rdna     &     120 & 30.41 & 10.52 & 19.89 &  657 & 24.38 & 25.59 & 23.66 & 22.2\% \\
    risort   &      25 & 89.81 & 31.84 & 57.89 & 2051 & 35.28 & 35.56 & 35.62 & 60.3\% \\
    rlife    &     700 & 27.02 &  2.74 & 24.28 &  179 & 26.04 & 31.23 & 23.54 & 12.9\% \\
    rqueens  &      15 & 35.65 &  9.57 & 26.08 &  604 & 43.09 & 50.24 & 48.95 &-37.3\% \\
    \hline
    \hline
  \end{tabular}
  \label{table:experiment:runtime}
  \normalsize
\end{table}

The rbmm3 system gets clearly better runtimes than the boehm system
for 15 out of our 18 benchmark programs,
including both deterministic and nondeterministic programs.
The speedups range from around 8\% to more than 60\%.
(We do not count the 2.3\% speedup as ``clearly better''.)
The overall average speedup,
even including the two programs with slowdowns, is about 24\%.
We get this promising result because with RBMM,
we avoid the burden of runtime garbage collection,
and because the overhead of supporting regions is reasonably modest.
Moreover, the runtimes of 10 of these 15 programs
are smaller than the corresponding runtimes in the boehm system
even excluding garbage collection times,
which strongly suggests that RBMM also improves data locality.
In \bench{bigcatch} and \bench{filrev}, two difficult cases for RBMM,
their memory-use pattern actually has even more adverse effects
on the operation of the Boehm collector.
These programs all build very large lists that are live data
before producing any garbage,
so during their initial phase,
the traversal of the memory allocated so far
by the collector's marking pass is almost completely a wasted effort.

\begin{table}[tb]
  \centering
  \small
  \caption{Frame statistics in rbmm1 and rbmm2 systems.}
  \tiny
  \begin{tabular}{l@{\hspace{0.03cm}}|*{11}{@{\hspace{0.03cm}}r@{\hspace{0.03cm}}|}}
    \hline
    \hline
    \multirow{2}{*}{}
    & \multicolumn{5}{@{\hspace{0.03cm}}c@{\hspace{0.03cm}}|}{\textbf{Disj frames}}
    & \multicolumn{6}{@{\hspace{0.03cm}}c@{\hspace{0.03cm}}|}{\textbf{Ite frames}} \\
    & \multicolumn{1}{@{\hspace{0.03cm}}c@{\hspace{0.03cm}}}{Total}
    & \multicolumn{1}{@{\hspace{0.03cm}}c@{\hspace{0.03cm}}}{M}
    & \multicolumn{1}{@{\hspace{0.03cm}}c@{\hspace{0.03cm}}}{\# Words}
    & \multicolumn{1}{@{\hspace{0.03cm}}c@{\hspace{0.03cm}}}{Mw}
    & \multicolumn{1}{@{\hspace{0.03cm}}c@{\hspace{0.03cm}}|}{Sr}
    & \multicolumn{1}{@{\hspace{0.03cm}}c@{\hspace{0.03cm}}}{Total}
    & \multicolumn{1}{@{\hspace{0.03cm}}c@{\hspace{0.03cm}}}{M}
    & \multicolumn{1}{@{\hspace{0.03cm}}c@{\hspace{0.03cm}}}{\# Words}
    & \multicolumn{1}{@{\hspace{0.03cm}}c@{\hspace{0.03cm}}}{Mw}
    & \multicolumn{1}{@{\hspace{0.03cm}}c@{\hspace{0.03cm}}}{Sr}
    & \multicolumn{1}{@{\hspace{0.03cm}}c@{\hspace{0.03cm}}|}{P} \\
    \hline
    bigcatch     & 0          &  0  & 0          &  0  &  0          &   5,046      &   1 &  35,504      &  11 & 5,091      & 47 \\
    boyer        & 0          &  0  & 0          &  0  &  0          &   38,629     &   2 &  271,469
  &  14 & 38,984     &  1 \\
    bsolver      & 90         &  1  & 1,170      &  13 &  270        &   244        &   1 &  4,031       &  19 & 1,018      &  1 \\
    crypt        & 55         &  4  & 220        &  16 &  0          &   2          &   1 &  9           &  5  & 0          &  1 \\
    filrev       & 0          &  0  & 0          &  0  &  0          &   5,001      &   1 &  50,005      &  10 & 10,000     &  1 \\
    life         & 0          &  0  & 0          &  0  &  0          &   177,789    &   1 &  1,777,885   &  10 & 355,576    &  1 \\
    healthy      & 2,431      &  9  & 24,304     &  84 &  4,860      &   17,449,110 &   2 &  174,491,089 &  14 &  34,898,216 &  1 \\
    queens       & 12,356,498 &  12 & 86,495,486 &  84 &  12,356,498 &   2          &   1 &  10          &  5  & 0          &  2 \\
    sudoku       & 81         &  81 & 810        &  810&  162        &   2          &   1 &  21          &  16 & 4          &  1 \\
    \hline
    rlife        & 0          &  0  & 0          &  0  &  0          &   177,789    &   1 &  1,777,885   &  10 & 355,576    &  1 \\
    rqueens      & 12,356,498 &  12 & 86,495,486 &  84 &  12,356,498 &   2          &   1 &  10          &  5  & 0          &  2 \\
    \hline
    \hline
  \end{tabular}
  \label{table:experiment:rtsupportrbmm12}
  \normalsize
\end{table}

\begin{table}[tb]
  \centering
  \small
  \caption{Frame statistics in rbmm3 system.}
  \tiny
  \begin{tabular}{l@{\hspace{0.03cm}}|*{11}{@{\hspace{0.03cm}}r@{\hspace{0.03cm}}|}}
    \hline
    \hline
    \multirow{2}{*}{}
    & \multicolumn{5}{@{\hspace{0.03cm}}c@{\hspace{0.03cm}}|}{\textbf{Disj frames}}
    & \multicolumn{6}{@{\hspace{0.03cm}}c@{\hspace{0.03cm}}|}{\textbf{Ite frames}} \\
    & \multicolumn{1}{@{\hspace{0.03cm}}c@{\hspace{0.03cm}}}{Total}
    & \multicolumn{1}{@{\hspace{0.03cm}}c@{\hspace{0.03cm}}}{M}
    & \multicolumn{1}{@{\hspace{0.03cm}}c@{\hspace{0.03cm}}}{\# Words}
    & \multicolumn{1}{@{\hspace{0.03cm}}c@{\hspace{0.03cm}}}{Mw}
    & \multicolumn{1}{@{\hspace{0.03cm}}c@{\hspace{0.03cm}}|}{Sr}
    & \multicolumn{1}{@{\hspace{0.03cm}}c@{\hspace{0.03cm}}}{Total}
    & \multicolumn{1}{@{\hspace{0.03cm}}c@{\hspace{0.03cm}}}{M}
    & \multicolumn{1}{@{\hspace{0.03cm}}c@{\hspace{0.03cm}}}{\# Words}
    & \multicolumn{1}{@{\hspace{0.03cm}}c@{\hspace{0.03cm}}}{Mw}
    & \multicolumn{1}{@{\hspace{0.03cm}}c@{\hspace{0.03cm}}}{Sr}
    & \multicolumn{1}{@{\hspace{0.03cm}}c@{\hspace{0.03cm}}|}{P} \\
    \hline
    bigcatch  & 0          & 0  & 0          &  0   &  0          &  47      & 1 &  235     & 5   &  0       &47\\
    boyer     & 0          & 0  & 0          &  0   &  0          &  38,272  & 1 &  267,899 & 7   & 38,270   & 1\\
    bsolver   & 0          & 0  & 0          &  0   &  0          &  18      & 1 &  175     & 10  &  34      & 1\\
    crypt     & 55         & 4  & 220        &  16  &  0          &  2       & 1 &  9       & 5   &  0       & 1\\
    filrev    & 0          & 0  & 0          &  0   &  0          &  1       & 1 &  5       & 5   &  0       & 1\\
    life      & 0          & 0  & 0          &  0   &  0          &  1       & 1 &  5       & 5   &  0       & 1\\
    healthy   & 2,431      & 9  & 24,304     &  84  &  4,860      &  2       & 1 &  9       & 5   &  0       & 1\\
    queens    & 12,356,498 & 12 & 86,495,486 &  84  &  12,356,498 &  2       & 1 &  10      & 5   &  0       & 2\\
    sudoku    & 81         & 81 & 324        &  324 &  0          &  2       & 1 &  15      & 10  &  2       & 1\\
    \hline
    rlife     & 0          & 0  & 0          &  0   &  0          &  1       & 1 &  5       & 5   &  0       & 1\\
    rqueens   & 12,356,498 & 12 & 86,495,486 &  84  & 12,356,498  &  2       & 1 &  10      & 5   &  0       & 2\\
    \hline
    \hline
  \end{tabular}
  \label{table:experiment:rtsupportrbmm3}
  \normalsize
\end{table}

Before discussing the results of the other programs,
we show detailed information about the disj frames and the ite frames
that are used in the benchmark programs \emph{per iteration}.
This information
is in Table~\ref{table:experiment:rtsupportrbmm12} for rbmm1 and rbmm2
(which always behave the same in these respects)
and in Table~\ref{table:experiment:rtsupportrbmm3} for rbmm3.
Both tables include only the programs
that use at least two frames during their runtime.

The five columns related to disj frames are as follows:
Total is the total number of disj frames used in one iteration;
M is the maximal number of disj frames coexisting at some point;
\# Words is the total number of words used for all the disj frames;
Mw is the maximal number of words used at some point; and
Sr is the total number of size records saved.
No regions are protected at semidet disjunctions in these benchmarks.
For ite frames,
the first five columns have meanings analogous to those for disj frames,
while the last column gives the total number of regions
that are protected by the ite frames
by having their handles saved in these frames.
The Mw columns show that
the memory used by both these kinds of embedded frames
is negligible in all benchmarks.
We do not show information about commit frames at all because
each nondeterministic program uses just one commit frame of four words
and no dynamic information is saved in them.

The rbmm3 system is only a little faster than boehm on \bench{crypt}.
Despite being a nondeterministic program,
the runtime support for backtracking it needs is rather cheap
(see Table~\ref{table:experiment:rtsupportrbmm3}).
However, the program handles a large number of small regions,
more than 125 million regions in total
(417 regions in each of 300,000 iterations),
with an average of just over eight words per region,
and the largest region being 64 words.
The cost of creating and destroying the region
has to be amortized over the words stored in the region.
In large regions,
the proportion of this overhead falling on any one word is negligible, but
in small regions, it can be substantial.
So rbmm3's gain due to avoiding runtime garbage collection
is almost exactly counterbalanced
by the overhead of handling many small regions,
resulting in just a small overall speedup.

This problem also manifests itself to various extents
in the other programs that handle many small-to-medium size regions
(more than ten million of them).
This can be seen in programs such as
\bench{dna}, \bench{life}, \bench{healthy}, \bench{sudoku},
\bench{rdna}, and \bench{rlife},
where we still have clear speedups
but they are not as good
as the speedups for programs with fewer, larger regions.
The memory results in Table~\ref{table:experiment:memory}
show that with rbmm3,
\bench{rdna} indeed needs much less memory than \bench{dna},
since it can reuse memory better with the help of its copying predicate.
Unfortunately, the overhead of copying
still causes \bench{rdna} to be about 12\% slower than \bench{dna},
though the slowdown for rbmm3 is less than for boehm (where it is 20\%).
However, compared to \bench{crypt},
\bench{queens} has many more nondet disjunctions
so it has to pay the cost of supporting backtracking within them many times
(see Tables~\ref{table:experiment:rtsupportrbmm12}
and~\ref{table:experiment:rtsupportrbmm3}),
\emph{and} it has to pay for handling many small regions
(68M regions with an average of about 27 words each),
and yet rbmm3 gets a speedup of 8.7\% over boehm on this benchmark.

The two worst cases for rbmm3 are \bench{rqueens} and \bench{boyer}.
\bench{rqueens} uses about five times as many regions as \bench{queens},
which makes the average region much smaller
than the already too small regions in \bench{queens}.
This is the negative side-effect of copying terms to new regions
to allow their old ones to be freed earlier.
That copying does achieve its objective;
we can see in Table~\ref{table:experiment:impact_rts}
that the memory \bench{queens} recovers from within regions
is recovered by \bench{rqueens} in the form of whole regions.
\bench{rqueens} actually never recovers memory from \emph{within} regions,
which means that overhead it pays for trying to do that
(saving size records at disj frames)
is useless while being quite expensive.
The slowdown in \bench{boyer} is mainly due to
the cost of saving size records (more than 306 million of them) at ite frames,
which are also all in vain.
A closer look at \bench{boyer} reveals that
it contains some semidet procedures that allocate into their input regions,
and the conditions of some if-then-elses call these procedures.
So the compiler needs to save the size records of those regions
if it wants to have instant reclaiming.
However, for the specific input used in our benchmark,
the calls to these semidet predicates all succeed,
so instant reclaiming has no words to reclaim.
See Section~\ref{seCfuture} for an idea
that would allow us to eliminate such unprofitable overhead.

Comparing the runtime results for rbmm2 and rbmm3
gives us an idea about the usefulness of tracking allocated regions.
While the reduction in the number of region arguments
does not have a strong impact in these benchmarks,
having less supporting code for backtracking
shows marked speedups for \bench{life}, \bench{healthy} and \bench{rlife}.
This enhanced performance corresponds with the reductions in
Table~\ref{table:experiment:rtsupportrbmm3} compared to
Table~\ref{table:experiment:rtsupportrbmm12}.
We can see that the main impact is on the ite frames.
For \bench{filrev} and \bench{life}, we can get rid of them completely,
except for one needed by the benchmarking mechanism itself.
For some others, we no longer have to save any size records to ite frames.
This is very important because
while nondet disjunctions are rare in Mercury programs,
if-then-elses are very common.
Ensuring their efficiency is therefore vital
to the efficiency of Mercury programs as a whole.
However, tracking of allocated regions cannot help in all cases,
such as in the case of \bench{boyer}.
For the programs for which rbmm3 seems slower than rbmm2,
this is purely a chance cache effect.
We have examined the C files generated by the Mercury compiler,
and for each such benchmark,
the only difference between the two versions is that
the rbmm2 version executes some statements that the rbmm3 version does not,
while using larger stack frames.

Comparing runtimes for rbmm1 and rbmm2,
we see that in the programs that use runtime support for backtracking,
using macros to implement that support may improve performance.
Table~\ref{table:experiment:runtime} shows that
\bench{boyer}, \bench{life}, \bench{healthy},
\bench{queens}, \bench{rlife} and \bench{rqueens}
are all at least 5\% faster in rbmm1 than they are in rbmm2.
This is because using macros avoids the cost of calling functions,
and because these programs are so small
that the increase in code size
does not adversely affect instruction cache behavior.
However, we expect that for larger programs,
the slowdown due to the reduction in the effectiveness of the instruction cache
will outweigh the cost of the calls.
However, in multi-module programs,
it should be possible to compile most modules with function calls
while compiling with macros
the modules in which the program spends most of its time,
thus getting the best of both worlds.

\subsection{The Impact of Sharing on Reusing Regions}
\label{seCevalSharing}

One can argue that sharing is the most basic and natural form of memory reuse.
However, sharing can conflict with RBMM,
because in RBMM we want terms with different lifetimes
to be stored in different regions,
and a subterm shared between two terms of different lifetimes
obviously cannot be stored in two different regions at once.
In this section we study in detail some benchmark programs
that we selected specifically for insights about the impact of sharing on RBMM.
Some of them are known difficult cases for RBMM
such as \bench{dna} and \bench{life}.
Some others create sharing that make it hard for in-place updating
such as \bench{isort}, \bench{bigcatch},
and \bench{filrev}~\cite{Aspinall08usageaspect}.

In our region points-to analysis,
we essentially put two program variables into the same region in two cases:
when there is an assignment between them,
or they are bound to a term and its same-type subterm
in a recursive data structure (Section~\ref{seCrpta}).
When the variables in a region have different lifetimes,
we will have a sort of memory leak,
because the memory of the variables with shorter lifetimes
will not be reclaimed until the longest lived variable dies.

One solution for this is to
copy the live data in the region to a different region,
so that the space used by the dead data can be reclaimed.
We experiment with this approach in \bench{rdna},
\bench{rlife}, \bench{risort}, and \bench{rqueens}.

The \bench{life} benchmark encodes the Game of Life
in which a new generation is generated from a previous one
based on a set of production rules.
From an initial generation,
it uses a loop (in the \code{life} predicate)
to produce several intermediate ones
before reaching the final generation, which is the wanted output.
We represent a generation by a list of \emph{live} cells,
with each cell being represented by its row and column in a 20x20 board.
To store a generation, we need two regions,
one for the skeleton and the other for the cells.
In the program,
the list skeletons of two successive generations are independent
while their cells may share.
In the recursive case of the predicate \code{life},
we first call \code{next\_gen} to compute the next generation,
whose skeleton could be in a different region,
and then we call \code{life} recursively
with the next generation as input.
In the base case,
we assign the current generation,
which is the ``next'' generation created by the caller,
to the output generation.
The computation is summarized in Figure~\ref{fig:life}.
\begin{figure}
\centering
\includegraphics[scale=0.8]{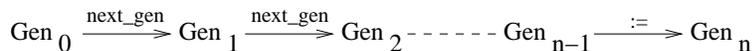}
\caption{The computation of generations in \bench{life}.}
\label{fig:life}
\end{figure}
Due to the assignment in the base case,
which creates sharing only between
the \emph{last} intermediate generation and the output generation,
our region points-to analysis decides
that the skeletons of the input and output generations
in the \code{life} predicate are in the same region,
and then enforces this for \emph{all} the (recursive) calls to \code{life}.
This eventually means that the skeletons of \emph{all} the generations
are placed in one big region with a size of 6,486 words.
In \bench{rlife}, we replace the assignment in the base case
with a call to a copying predicate that does not create any sharing,
thus allowing the compiler
to store the skeleton of each generation in a separate region,
which then can be reclaimed in time.
We see in Table~\ref{table:experiment:memory} that
the maximum amount of memory needed by \bench{rlife} is 2,056 words,
which is a 75\% reduction compared to \bench{life}'s 8,208 words.
This is because in \bench{rlife},
the skeletons of the old generations are reclaimed at each step.

The program \bench{dna} simulates the matching of a given DNA sequence
to each of the DNA sequences in a predefined set.
The matching degree of two sequences is represented by a similarity,
which is computed based on the similarities of their elements
with respect to the spatial relation among them.
The similarities between two sequences are calculated one by one
and put in an ordered tree, which is a recursive data structure.
To store a tree, we need two regions, one for the tree nodes
and the other for the structures where the similarities are stored.
Other than that, in this program,
there are assignments in several predicates
that establish sharing among the similarity structures in such a way that
all the similarities ever computed end up in the same large region of 4M words.
The maximal number of words in use during a run of the program is about 4.6M.
In the so-called region-friendly version \bench{rdna},
we make a fresh copy of each similarity and add the copy to the tree.
This allows the region analysis to decide
that the region to which the similarity is copied
is the region of the nodes of the tree,
and it can reclaim its previous region containing
all the temporary similarities involving in its computation.
The maximum amount of memory needed drops from 4.6M words to only 0.5M.
The size of the largest region also drops from 4M words to 0.43M words;
in \bench{rdna}, it contains only the skeleton of the tree.

In \cite{Phan09ppdp} we proposed a more desirable solution,
a more refined region analysis that,
by taking into account different execution paths,
can keep apart the regions of the variables in an assignment.
A dedicated implementation of the improved analysis
should achieve the same effect as
changing \bench{life} into \bench{rlife}
changing \bench{dna} into \bench{rdna},
without either requiring manual rewriting of the program
or incurring the cost of copying.

Another issue that we found was that
one of the Mercury compiler's existing optimizations, common structure reuse,
was reducing the effectiveness of our region analysis.
This optimization looks for conjunctions
in which the same term is assigned to two or more variables,
and then changes the code so that the term is constructed just once,
and then it is assigned to all the variables.
This is always an optimization for the boehm system,
but in cases where our region analysis
would want to assign those variables to different regions,
making them refer to the same memory cell creates unwanted sharing,
requiring our region analysis to merge the two variables' regions.
In general, the unmerged regions would be reclaimed at different times.
Therefore merging the two regions can delay
the reclamation of an unbounded amount of memory
by an unbounded amount of time.
The best way to avoid this problem
is to teach the optimization about regions,
and make it perform the transformation
only if the variables involved are in the same region.

The problems with memory reuse in RBMM in \bench{isort} and \bench{queens}
are typical for programs
that use recursive data structures such as lists and trees,
and continuously update them by adding to them and deleting from them.
Because the updated structure normally shares most parts of the original,
they are stored in the same regions,
which prevents us from reclaiming
the now-obsolete parts of the original structure.
In \bench{risort} and \bench{rqueens}, we try to improve memory reuse
by adding a predicate to copy the modified structure
so that the original region can be reclaimed after the copying.
In \bench{risort}, the copying happens after an integer is inserted,
while in \bench{rqueens}, it happens after a queen is deleted.
This modification obtains optimal memory management for \bench{risort}
(see Table~\ref{table:experiment:memory}).
In \bench{rqueens}, compared to \bench{queens},
the peak memory usage is higher.
This is due to region protection:
some disj-protected regions are removed but not reclaimed,
and instant reclaiming does not recover their memory until later.
However, the size of largest region drops to 24 words,
which is the storage needed to represent a list of 12 queens.

While memory reuse can be improved by this copying approach,
its runtime overhead is very expensive.
We see a 63\% increase of runtime
for \bench{rqueens} compared to \bench{queens}
in Table~\ref{table:experiment:runtime},
and for \bench{risort} we have to reduce the input size by a factor of ten
(to 1,000 integers, compared to 10,000 integers in \bench{isort})
to allow the program to finish in a reasonable time.
Similar problems with memory reuse in the presence of recursive data structures
can also be seen in \bench{dna} and \bench{rdna},
which insert similarity structures into trees,
and in \bench{bsolver}, which reduces the domains of the integral variables,
with the domains being represented as lists of integers.

The reason why \bench{bigcatch} and \bench{filrev} are not even faster
is also related to recursive data structures.
In this case the structures are not updated but only part of them is used,
i.e.\ only a part is live data,
but that still requires us to keep the whole region alive.
Copying the live data out of the region
would work just as well to recover memory,
and at just as high an overhead,
as in the previous case.
We do not have an automatic solution
for the problems related to the use of recursive data structures
in RBMM-only systems, but then, neither does anyone else.
The problem is well-known
among researchers who use type systems or type inference
to reason about memory structures
~\cite{Baker90,Chase90,TofteTalpin97,Henglein01hmn},
who nevertheless have to accept the loss of precision
as the price of having a finite model.
To improve storage use in such cases,
one can combine RBMM with other techniques,
such as runtime or compile time garbage collection.
The copying approach used by our region-friendly benchmarks
can be viewed as a simulation of runtime copying garbage collection.
Combining RBMM with copying garbage collection
has been realized in the MLKit~\cite{Hallenberg02}.

\section{Related Work}
\label{seCrelatedwork}

In this section, we only mention the most important and most related papers.
It is not our intention to give a detailed overview of the
research on RBMM for other programming paradigms.
An in-depth review of RBMM research for functional programming
can be found in~\cite{Tofte04Retro}.

The research on automated region-based memory management for programming
languages started with the work of Tofte and Talpin
\cite{TofteTalpin97} for functional programming,
in particular for a simplified call-by-value lambda calculus.
They divide program terms into regions
using a technique similar to unification-based type inference
in which the types have been annotated with region variables.
The lifetimes of the regions
are computed based on the lexical scope of the expressions and
the regions themselves are forced to follow stack discipline,
with the last region created always being the first one destroyed.
While lexically-scoped regions and stack discipline
seem natural for the evaluation of lambda expressions
and they simplify the task of deciding region lifetimes,
they often give regions lifetimes that are longer than needed,
increasing the program's memory requirements.
Possibly even more important, the cleanup they often require after a tail call
also spoils tail call optimization.
\cite{Birkedal96from} refined this system in several ways,
the most important being \emph{Storage Mode Analysis},
which mitigates the problems caused by the stack discipline
by resetting regions to zero size when their contents are no longer needed.
However, to make this region resetting possible,
programmers often have to rewrite their programs in unusual ways.

While Aiken et al.\ also used a stack in their inference algorithm,
they nevertheless thought that
forcing stack discipline on the lifetimes of regions
is too strict~\cite{Aiken95better},
and they decoupled region creation and removal,
allowing regions to have arbitrarily overlapped lifetimes.
Going even further in this direction,
Henglein, Makholm, and Niss in~\cite{Henglein01hmn}
proposed an imperative sublanguage on regions.
In their system, regions are allowed not only to have arbitrary lifetimes
but also to change their bindings.
Their regions also contain reference counters
that can give their system more flexibility in controlling their lifetimes.
The most complete functional programming system with RBMM
is the MLKit~\cite{MLKit}, which manages storage solely by RBMM.
This system, while still using stack discipline for the lifetimes of regions,
supports both resetting regions to zero size
and runtime garbage collection within regions.
Its performance is competitive with other state-of-the-art SML compilers.

Our static region analysis and transformation for Mercury
were inspired by the work in~\cite{Cherem04},
which also allowed arbitrarily-overlapped region lifetimes.
The analyses in that paper take into account
the data flow in a Java program
in order to determine the set of needed regions and their lifetimes.
Therefore the analyses had to be
redefined for Mercury to deal with unification and a control flow that are
fundamentally different from object manipulation and control flow in Java.
Cherem and Rugina use the classes of Java
to achieve a finite representation of the storage of (recursive) structures
in terms of regions,
but their starting assumptions are different from ours.
In our analysis,
we start by associating each variable with as many regions
as its type requires (e.g.\ skeletons and elements for \code{list\_int})
whereas they start by associating each variable with only one region
(the one for its class),
and add the other nodes later, on demand.
In the case of recursive types,
we know from the start that e.g.\ all the list skeleton nodes
of a given variable are in the same region.
Given a variable $v$ of class $c$ whose fields include,
directly or indirectly, other variables of class $c$,
they initially allocate different nodes in the region graph
to $v$ and those other variables,
and merge some of those nodes only when they see a link between them.
This complicates their analysis,
though in some cases it allows them to keep the regions separate
and thus free some memory earlier.
In logic programs, recursive types
are almost always processed using recursive procedures,
and such cases would be vanishingly rare.

Another difference between the two systems
that is likely to be more important in practice
is that the liveness information we derive in Section~\ref{seClra}
allows interprocedural creation of regions,
something that was not handled in \cite{Cherem04}.
This can give finer lifetimes to regions,
which can result in better memory reuse in certain situations.
For example, for a region like \code{R1} in \code{p}
in Figure~\ref{fig:recreate:annotated},
the system in~\cite{Cherem04}
would force \code{R1} to be live throughout \code{p}.
If we had replaced the atom at (4) with a recursive call to \code{p}
(such as \code{p(A - 1, B)})
their system would build up
all the temporary memory allocated at (1) in \code{R1}.

Note that using graphs to model storage is not at all new
in research about heap structures \cite{Chase90,Steensgaard96}.
Our graphs share many features with annotated types
where the annotation on each type constructor
is a location or region; see e.g.\ \cite{Baker90,TofteTalpin97}.
Baker in~\cite{Baker90} and many others pointed out that
such annotated types can also give information about sharing,
very similar to the concept of \emph{region-sharing} in this paper.

The first application of RBMM to logic programming
was the work of Makholm for Prolog,
described in \cite{Makholm00master} and~\cite{Makholm00}.
He realized that backtracking can be handled completely by runtime support,
which can keep the region inference simple.
However, the Prolog system he used
was not based on the usual implementation technology for Prolog,
the Warren Abstract Machine or WAM.
This shortcoming was fixed in~\cite{MakholmSagonas02}
where Makholm and Sagonas extended the WAM
to enable region-based memory management.
The main differences between their work and ours are that
Mercury supports if-then-elses with conditions
that can succeed more than once,
and the Mercury implementation generates specialized code
for many situations that Prolog handles with a more general mechanism.
(For example, Mercury has separate implementations
for nondet disjunctions and for semidet disjunctions.)
The first difference required new algorithms,
while the second posed a tough engineering challenge in keeping overheads down,
since due to Mercury's higher speed,
any given overhead would hurt Mercury more than Prolog.

\section{Future Work}
\label{seCfuture}

Our RBMM implementation already has some support for profiling.
When given a certain option,
the Mercury compiler will augment the RBMM support code it generates
with code that counts and keeps track of several things:
the number of region creations and removals,
the amount of memory allocated in regions,
the maximum size of regions,
the number and size of the embedded disj, ite and commit frames, and so on.
This option was the source of the information in
Tables~\ref{table:experiment:memory},~\ref{table:experiment:impact_rts},
\ref{table:experiment:rtsupportrbmm12},
and~\ref{table:experiment:rtsupportrbmm3}.
We would like to modify this profiling mechanism to also report,
for each region variable (both old and new) at each resume point,
the number of instant reclaiming attempts made
at that point for that region variable,
and the amount of memory recovered in those attempts.
We would like to then feed this information back to the compiler,
so that it can find out which attempts are too expensive
for the amount of memory they recover,
so it can simply avoid generating them.

Our current system prevents the reclamation of regions
that are forward dead but backward live entirely at runtime.
Such runtime protection is in fact necessary in general.
Given a procedure $p$ and a region $r$ with $r \in \deadR{p}$,
$p$ cannot know whether some disjunction to the left of its caller
makes $r$ backward live or not.
We could handle this situation by generating three versions of $p$.
The first version would assume that $r$ is backward live
and therefore never reclaim $r$,
the second version would assume that $r$ is backward dead,
and therefore always reclaim $r$,
and the third version
would make neither assumption
and would reclaim $r$ only if it is not protected,
as in our current system.
The caller would call the first version
if it itself makes the region backward live
(e.g. the call may be in one disjunct,
and a later disjunct in that disjunction may need the region),
or because the caller itself is a specialized version
that assumes that the region is backward live.
The caller would call the second version if it itself created the region,
and if there is no nondet construct between that creation and the call
that could make the region backward live.

Unfortunately, a procedure's $\deadRegs$ set may contain several regions,
and given $n$ regions, we may need up to $3^n$ copies of the procedure,
which is far too many,
since that many copies
would significantly degrade the effectiveness of the instruction cache.
Nevertheless, in some situations,
the fraction of execution time spent in the procedure
may justify creation of one or more specialized copies of the procedure.
We intend eventually to implement an optimization
that figures out which of the possible specialized versions can ever be called,
attempts to compare their cost in lost locality
to the speedup we can expect
from optimizing away unnecessary remove instructions,
and creates the specialized versions if and only if
the comparison indicates that it is beneficial to do so.
If a specialized version is not worth it,
the caller can call the original version of the procedure;
since this does runtime tests on all the removed regions
before reclaiming them, it still works in all cases.

What we \emph{could} improve without considering such complicated tradeoffs
are situations where the instruction that removes a region is in a procedure
that itself makes the region unconditionally protected at the removal site.
In such cases, we know statically
that the removal will not actually reclaim the region,
and that therefore we can simply optimize it away.
If such protection is only conditional,
we do have to consider the tradeoff.
Since we cannot guarantee optimizing away all protected removals,
the mechanisms we described in Section~\ref{seCsupportnondet}
will always be needed.

The main limitation of our work is that currently,
the program analysis underlying our system supports only a subset of Mercury.
We intend to work on extending the analysis to handle the rest of the language.
Since we already handle almost all of Mercury,
``the rest of the language'' covers only a few features:
Mercury procedures defined in foreign languages,
multi-module programs, and higher-order code.
To handle them, we need to ensure two things.
First, that the callers and callees involved in
calls to foreign language code, cross-module calls and higher order calls
all agree on the liveness of the regions involved in the call;
second, that they all agree on the sharing between those regions.
The first one is relatively easy to achieve
by simply setting the $\bornRegs$ and $\deadRegs$ sets of those calls to empty.
This will work;
any creations and removals of the regions
that \emph{would} have been in those sets will happen around the call.
The cost is that it may increase the program's memory consumption,
though only to the levels seen in some other RBMM systems.
The real problem is the second issue:
getting consensus between callers and callees on sharing.

\noindent\textbf{Handling foreign language procedures.}
Always setting the $\bornRegs$ and $\deadRegs$ sets
of foreign language procedures to the empty set
avoids burdening programmers with the responsibility
for managing the creation and removal of regions.
Since most foreign language procedures do not allocate any memory,
their writers do not need to know anything about regions at all.
The foreign language procedures that \emph{do} allocate memory
need to know where the allocation of each cell should happen.
In a hybrid system that combines RBMM with the Boehm collector,
it is simple enough to let such foreign procedures
keep doing what they do now,
which is doing all their allocation on the Boehm heap.
An RBMM-only system would need to make
the region arguments added to each procedure by our transformation
visible to the programmer,
and document which of these region variables represent
which part of each of the arguments originally created by the programmer,
so that when he or she writes code to create a new cell
that will become part of a term that will be bound to an output argument,
they can allocate it in the right region.
We would also need to give programmers a mechanism
that they can use to tell the compiler about any sharing they create
between the regions;
our Algorithm~\ref{algo:interproc} could then take this information on trust.
As for temporary structures that can never become part of an output argument,
programmers can put them where they wish.
They can put them in memory managed
by \code{malloc} and \code{free} (if the foreign language is C)
and their equivalents (if the foreign language is something else),
or, if we expose the functions for creating and removing regions,
they can put them in one or more programmer-managed regions instead.

\noindent\textbf{Handling multi-module programs.}
Our current implementation actually allows cross module calls;
if a program cannot call the procedures in the standard library's I/O module,
then it cannot print out its results.
The reason why we cannot yet handle multi-module programs in general
is that currently we do not do any region analyses across modules,
and hence we never pass region variables or any other information about regions
from one module to another.

The reason why implementing region analysis in multi-module programs is hard
is that the fixpoint computation in Algorithm~\ref{algo:interproc}
is inherently incompatible with separate compilation.
Mercury's compilation system ensures that when a module changes,
all other modules dependent on its interface will be recompiled
before the building of the executable,
but it guarantees that this will take a bounded number of steps.
As it is, Algorithm~\ref{algo:interproc} cannot provide a similar guarantee;
the procedures in a single SCC may be in different modules,
and each iteration of the search for the fixpoint
must analyze code in each of those modules.
We therefore need to either change the algorithm,
or make the compilation system flexible enough
to encompass fixpoint computations that need an unbounded number of iterations.
We have looked at the second option in the past,
using the ideas of \cite{latexmodel} as the basis,
but even if it were implemented,
being able to limit the number of iterations
would help compile programs more quickly.
There are some assumptions we can make that can help with that.
For example, we can assume that all input variables of cross-module calls
are in regions that the callee will not allocate in or remove;
if their last use is during the call, the caller will remove them upon return.
This loses some precision and therefore reduces the efficiency of memory reuse,
but this is a known and fairly widespread problem:
\emph{most} program analysis and optimizations
lose precision at module boundaries,
and in almost every case this is seen as an acceptable tradeoff.
The challenge will be in coming up with mechanisms
for handling the regions of output variables
that still allow memory to be recovered effectively enough.
We have some ideas, but no solutions yet.

\noindent\textbf{Handling higher order code.}
Mercury supports two forms of higher order calls:
calling an ordinary higher order term (a closure),
and calling a typeclass method.
The challenge in both cases is that
the identity of the called procedure
may not be apparent when the calling module is compiled,
which prevents Algorithm~\ref{algo:interproc} from analyzing it.
There are two avenues of possible solutions.
First, the Mercury compiler already contains an analysis
that attempts to find out which procedures each higher order value may call.
If this analysis succeeds, an adapted version of Algorithm~\ref{algo:interproc}
can convey the requirements of the calling context to these procedures,
and convey to the caller
the worst-case demands that any of the callees may make
(e.g.\ in terms of which nodes they need merged to reflect their sharing).
Second, in case the analysis fails
(which may happen e.g.\ because
the caller picks up those higher order values
from a data structure created elsewhere),
we need an interface between caller and callee
that is standard and thus does not require negotiation
(which is what the fixpoint iteration
in Algorithm~\ref{algo:interproc} represents).

Our search for this standard interface
will not be restricted to RBMM-only systems.
We will also look at hybrid systems
in which RBMM coexists with the Boehm general purpose garbage collector,
each looking after some of the program's memory.
Hybrid system that combine RBMM with a runtime collector
have proven useful in other contexts \cite{Hallenberg02},
and they may prove useful in this one as well.
We do not intend to look at hybrid schemes that integrate RBMM
with Mercury's accurate garbage collector
since that collector is actually significantly slower than the Boehm collector
\cite{Fergus02gc}.
We do however intend to look at integrating our RBMM system
with the compile time garbage collection scheme
reported in \cite{Mazur00,Mazur01,Mazur04phd}.

\section{Conclusion}
\label{seCconclusion}

We have made region-based memory management available
as an alternative storage management technique
for programs written in a very large subset of Mercury.
This involved the design and implementation of two program analyses
(region points-to analysis and region liveness analysis)
and a program transformation,
the modification of the Mercury code generator
to use the information produced by the analyses and transformation
to generate code that uses RBMM to manage its memory,
and the implementation of the primitive operations used by the generated code.

We provide termination and correctness theorems
for our region analyses and our transformation algorithms.
These ensure the safety of memory accesses and region operations
with respect to forward liveness.
Our discussions in Section~\ref{seCsupportnondet} also strongly argue
that our runtime support operations
guarantee the safety of memory accesses and region operations
with respect to backward liveness
(i.e.\ in the presence of backtracking).
These operations also instantly reclaim
the memory allocated by backtracked-over computations,
which help programs to reuse memory effectively.

The main challenge for the runtime support
is to support backtracking correctly
without incurring significant overhead,
especially in deterministic code.
Our experiments show that using RBMM instead of the Boehm collector
yields nontrivial speedups for 15 out of our 18 benchmark programs,
these speedups ranging from near 10\% to a remarkable more than 60\%.
We even get large speedups for some benchmarks
that are known to be difficult cases for RBMM.
This indicates that the runtime support we provided for backtracking
incurs very modest overhead in most cases,
contributing to the overall better performance.

The memory use results of the benchmarks are also positive:
in some programs we obtain optimal memory consumption.
On average, our benchmarks require about one-twentieth the memory
with RBMM than with the Boehm collector (only 5\%),
and even if we exclude the region-friendly programs,
the figure is about one-eighteenth (5.4\%).
This even before including
any of the optimizations that have been studied for RBMM,
such as stack allocation of regions~\cite{Birkedal96from,Cherem04},
and merging regions that are removed at the same points ~\cite{Makholm00}.

Everything we have described is available
in current releases-of-the-day from the Mercury web site.
The experimental setup for this paper is available from
\code{http://www.cs.kuleuven.be/\~{}gerda/rbmm/rbmm\_benchmarks.tar};
it includes the benchmark programs as well as the benchmarking script.

\bibliography{new_rpta}
\end{document}